%% file: main_R1.tex
\documentclass[11pt]{amsart}
\usepackage{graphicx,subcaption}
\usepackage[usenames, dvipsnames]{color}
\usepackage{array}
\usepackage[margin=1in, centering]{geometry}
\usepackage{lscape}
\usepackage{multirow}
\usepackage[short-journals,short-months,backrefs]{amsrefs}
\usepackage[ruled,linesnumbered,algosection]{algorithm2e}
\usepackage[noend]{algpseudocode}
\usepackage{float}
\usepackage[colorlinks=true,pagebackref]{hyperref}
\usepackage{amsmath,amsthm,amssymb,latexsym,mathrsfs,bigints,cancel, url, tikz,enumitem, soul, bbm}
\usepackage[most]{tcolorbox}
\usetikzlibrary{angles,positioning,arrows.meta}

\newcolumntype{M}[1]{>{\centering\arraybackslash}m{#1}}
\newcolumntype{P}[1]{>{\flushleft\arraybackslash}m{#1}}

\definecolor{bg}{RGB}{220,220,220}

\theoremstyle{plain}
\newtheorem{theorem}{Theorem}
\newtheorem{lemma}{Lemma}[section]
\newtheorem{proposition}[lemma]{Proposition}
\newtheorem{corollary}[lemma]{Corollary}
\newtheorem{conjecture}[lemma]{Conjecture}

\theoremstyle{definition}
\newtheorem{definition}[lemma]{Definition}

\newtheorem{openquestion}{Open Question}

\theoremstyle{remark}
\newtheorem{remark}[lemma]{Remark}
\newtheorem*{remark*}{Remark}

\newcommand{\bi}[1]{\mathsf{b}_{#1}}
\newcommand{\uds}[1]{u^{\rm DS}_{#1}}
\newcommand{\urr}[1]{u^{\rm RR}_{#1}}
\newcommand{\srr}[1]{\sigma^{\rm RR}_{#1}}

\title[Laplacian growth \& sandpile growth on the Sierpinski gasket]{Laplacian growth \& sandpiles on the Sierpinski gasket: limit shape universality and exact solutions}

\author{Joe P.\@ Chen}
\address[Joe P.\@ Chen]{Department of Mathematics, Colgate University, Hamilton, NY 13346, USA.}
\email{jpchen@colgate.edu}
\urladdr{\url{http://math.colgate.edu/~jpchen}}

\author{Jonah Kudler-Flam}
\address[Jonah Kudler-Flam]{Kadanoff Center for Theoretical Physics, The University of Chicago, Chicago, IL 60637, USA.}
\email{jkudlerflam@uchicago.edu}

\thanks{This project was initiated while JKF was an undergraduate at Colgate University. It has been supported in part by the Research Council of Colgate University, the Simons Foundation (Collaboration Grant for Mathematicians \#523544), and the National Science Foundation (DMS-1855604).}

\date{\today}

\keywords{Laplacian growth, rotor-router aggregation, divisible sandpiles, abelian sandpiles, sandpile group, internal diffusion-limited aggregation, harmonic measure, analysis on fractals, Sierpinski arrowhead curve, self-similarity, limit shapes, exact renormalization.}
\subjclass[2010]{
05C20, 
05C81, 
05E18, 
28A80, 
31A15, 
37B15, 
60K05, 
82C24. 
}

\begin{document}

\begin{abstract}
We establish quantitative spherical shape theorems for rotor-router aggregation and abelian sandpile growth on the graphical Sierpinski gasket ($SG$) when particles are launched from the corner vertex.

In particular, the abelian sandpile growth problem is exactly solved via a recursive construction of self-similar sandpile tiles.
We show that sandpile growth and patterns exhibit a $(2\cdot 3^n)$-periodicity as a function of the initial mass. 
Moreover, the cluster explodes---increments by more than 1 in radius---at periodic intervals, a phenomenon not seen on $\mathbb{Z}^d$ or trees.
We explicitly characterize all the radial jumps, and use the renewal theorem to prove the scaling limit of the cluster radius, which satisfies a power law modulated by log-periodic oscillations.
In the course of our proofs we also establish structural identities of the sandpile groups of subgraphs of $SG$ with two different boundary conditions, notably the corresponding identity elements conjectured by Fairchild, Haim, Setra, Strichartz, and Westura.

Our main theorems, in conjunction with recent results of Chen, Huss, Sava-Huss, and Teplyaev, establish $SG$ as a positive example of a
state space which exhibits ``limit shape universality,'' in the sense of Levine and Peres, among the four Laplacian growth models: divisible sandpiles, abelian sandpiles, rotor-router aggregation, and internal diffusion-limited aggregation (IDLA).
We conclude the paper with conjectures about radial fluctuations in IDLA on $SG$, possible extensions of limit shape universality to other state spaces, and related open problems.
\end{abstract}

\maketitle
\tableofcontents

\section{Introduction and main results} \label{sec:intro}

Internal aggregation models---such as internal diffusion-limited aggregation (IDLA) and rotor-router aggregation---and the abelian sandpile model have a long history in the statistical physics and the mathematics literature.
For reasons to be explained shortly, we prefer to call them \textbf{Laplacian growth models}, to emphasize their close association with the \textbf{combinatorial graph Laplacian}, which is defined on a connected, locally finite, undirected graph $G=(V(G),E(G))$ by the non-positive symmetric matrix
\[
\Delta_G(x,y) = \left\{\begin{array}{ll} -\deg(x), &\text{if } x=y, \\ \mathsf{N}_{xy}, &\text{if } x\neq y,\end{array}\right. \quad (x,y\in V(G))
\]
where $\mathsf{N}_{xy}$ is the number of edges connecting $x$ and $y$.
It is also convenient to introduce the graph Laplacian normalized by the vertex degree, $\Delta(x,y) := \frac{1}{\deg(x)} \Delta_G(x,y)$, sometimes also called the \textbf{probabilistic graph Laplacian}.

The present work is devoted to the solutions of two Laplacian growth models---rotor-router aggregation and the abelian sandpile growth model---on the graphical Sierpinski gasket ($SG$, see Figure \ref{fig:SG}), when particles (or ``chips'') are launched from the corner vertex of $SG$.
In particular, we solve the abelian sandpile growth problem exactly via a renormalization scheme involving self-similar sandpile configurations, or sandpile ``tiles,'' on subgraphs of $SG$.

The motivations for our study are twofold:
\begin{enumerate}[wide]
\item \emph{The limit shape universality conjecture.}
A folklore conjecture in the sandpile community is that on a fixed state space, the growing clusters associated with the four Laplacian growth models---IDLA,
rotor-router aggregation, divisible sandpiles,
and abelian sandpiles---have the same limit shape.
This ``limit shape universality'' conjecture does not hold in general. In fact, exhibiting even a positive example of a state space beyond $\mathbb{Z}$ is difficult. On $\mathbb{Z}^d$, it has been proven that the first three models have Euclidean balls as limit shapes \cites{LBG92, LevinePeres09}, but numerical evidence strongly suggests that the limit shape in the abelian sandpile model on $\mathbb{Z}^2$ is closer to a polygon than an Euclidean ball, see Figure \ref{fig:ASMZ2}. Table \ref{table:Zdshape} summarizes the state of the art on $\mathbb{Z}^d$. See the excellent survey \cite{LevinePeres17} of Levine and Peres for an overview of the basic mechanisms behind each of the Laplacian growth models and the lack of limit shape universality on $\mathbb{Z}^d$.

We will show that limit shape universality holds on $SG$, see Theorem \ref{thm:shapeuniv} below.

\item \emph{From ``fractals in a sandpile'' to ``sandpile on a fractal.''}
The abelian sandpile model, introduced by Bak, Tang, and Wiesenfeld \cites{BTW87, BTW88}, exhibits nontrivial fractal patterns on the Euclidean lattice $\mathbb{Z}^d$, see Figure \ref{fig:ASMZ2} again.
While the fractal nature of the sandpile patterns has been long recognized, rigorous proofs did not arrive until recently.
In an important breakthrough, Pegden and Smart \cite{PegdenSmart} showed that the scaling limit of the patterns exists in the sense of weak-$*$ $L^\infty(\mathbb{R}^d)$ convergence.
Shortly after, the seminal works of Levine, Pegden, and Smart \cites{LPS16, LPS17} established the existence of the Apollonian structure in the sandpile patterns on $\mathbb{Z}^2$, via the analysis of $\mathbb{Z}$-valued superharmonic matrices. (See Figure \ref{fig:Apollonian} for a picture of an Apollonian gasket.)

Inspired by the ``analysis of fractals in sandpiles,'' and following the lead of Strichartz, we are prompted to study ``analysis of sandpiles on fractals.'' 
We are grateful to Strichartz and his undergraduate students for making several key numerical findings and conjectures in \cite{ASMSGStr}.
Indeed, the exact renormalization scheme we present here stems from an insight which appeared in \cite{ASMSGStr}, namely, Dhar's multiplication by identity test applied to cut points on a nested fractal graph. 
This is the key mechanism behind the production of self-similar sandpile tiles on $SG$, which enables us to establish the identity element of the sandpile groups of subgraphs of $SG$ (Theorem \ref{thm:groupSG}), and to solve the sandpile growth problem exactly (Theorems \ref{thm:tail} and \ref{thm:radialcycle}).
Let us mention that the notion of sandpile tiles has also appeared in the Euclidean setting \cites{CPS15,PegdenSmart17}, and its connection to curves in tropical geometry has been explored in \cites{KS16}.
\end{enumerate}

\begin{figure}
\centering
\includegraphics[width=0.5\textwidth]{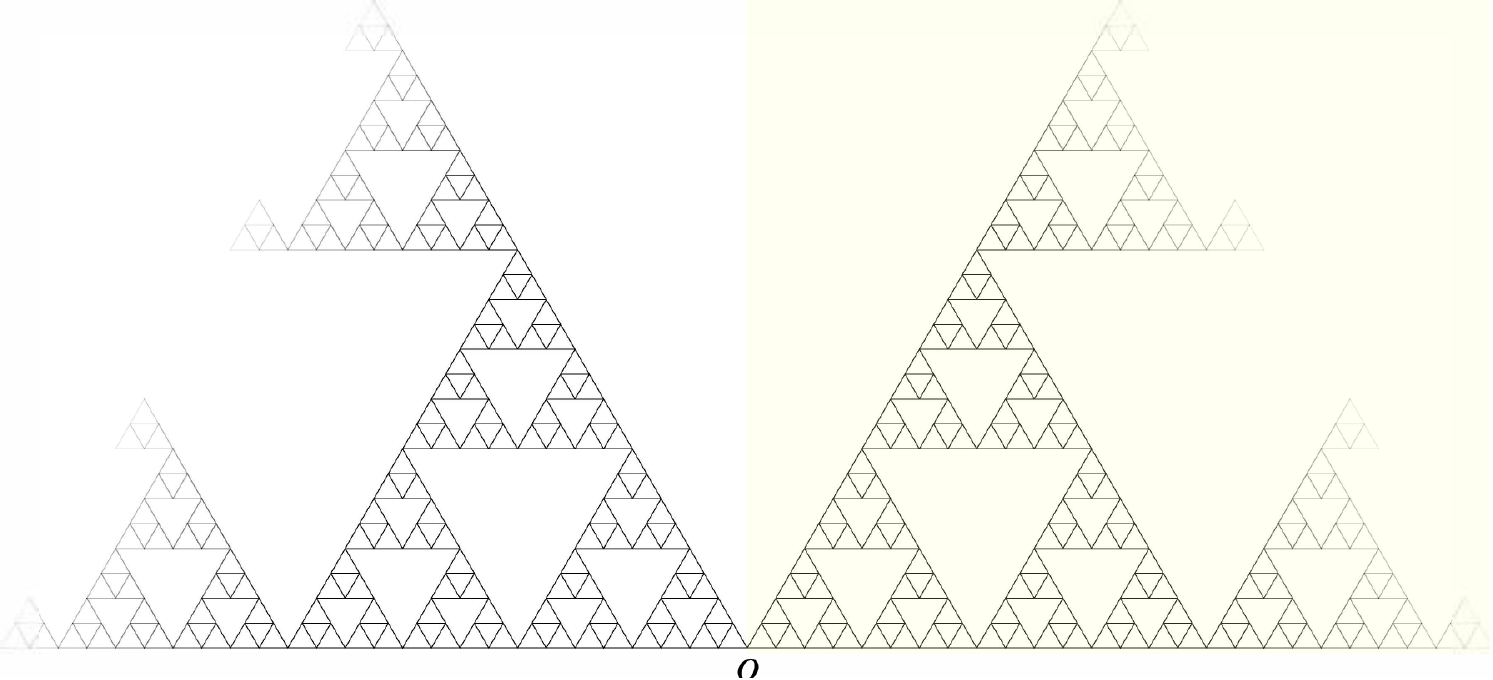}
\caption{The double-sided Sierpinski gasket graph $SG$. The single-sided gasket graph (shaded) is just one half of the double-sided one. The vertex $o$ is the source vertex from which particles are launched.}
\label{fig:SG}
\end{figure}

\begin{figure}
\centering
\begin{tabular}{m{0.35\textwidth} m{0.35\textwidth} m{0.1\textwidth}}
\includegraphics[width=0.3\textwidth]{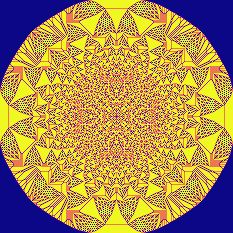}
&
\includegraphics[width=0.3\textwidth]{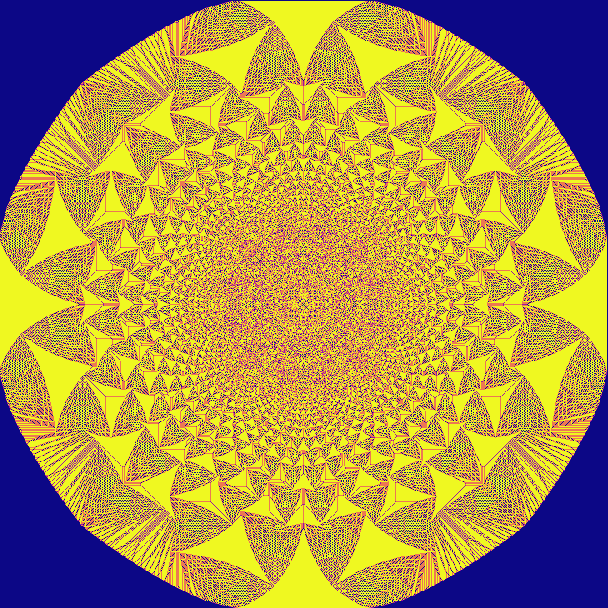}
&
\begin{tabular}{lm{10px}}
0 & \includegraphics[width=10px]{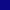} \\
1 & \includegraphics[width=10px]{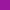} \\
2 & \includegraphics[width=10px]{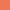} \\
3 & \includegraphics[width=10px]{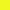}
\end{tabular}
\end{tabular}
\caption{Abelian sandpile cluster on $\mathbb{Z}^2$ starting with $10^5$ chips (left) and $10^6$ chips (right) at the origin. Each vertex is colored according to the number of chips there.}
\label{fig:ASMZ2}
\end{figure}

\begin{figure}
\centering
\includegraphics[width=0.3\textwidth]{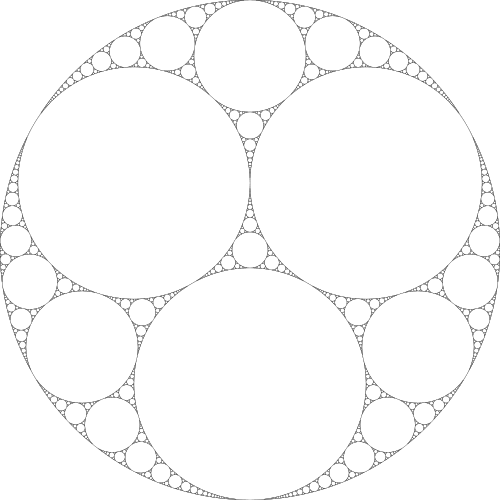}
\caption{
An Apollonian gasket. Picture by Time3000, \url{https://commons.wikimedia.org/wiki/File:Apollonian_gasket.svg} [\href{http://www.gnu.org/copyleft/fdl.html}{GFDL} or \href{https://creativecommons.org/licenses/by-sa/4.0}{CC BY-SA 4.0}], from Wikimedia Commons.
}
\label{fig:Apollonian}
\end{figure}

\emph{Notation.} Throughout the paper, unless noted otherwise, a graph $G=(V(G),E(G))$ is assumed to be undirected, locally finite, and connected. 
When $xy\in E(G)$ we write $x\sim y$.
Let $d: V(G)\times V(G)\to \mathbb{N}_0$ be the graph metric on $G$, and $B_x(r) :=\{y\in V(G): d(x,y)\leq r\}$ be the closed ball of radius $r$ centered at $x$. The cardinality of a finite set $S$ is denoted $|S|$.

\subsection{Sierpinski gasket graph ($SG$)}

\begin{figure}
\centering
\includegraphics[width=0.8\textwidth]{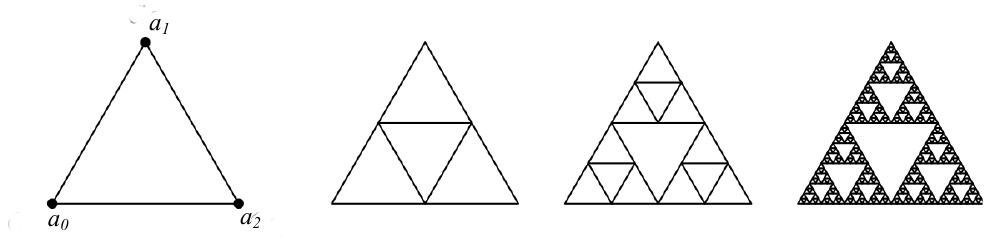}
\caption{The pre-fractal Sierpinski gasket graphs of level $0$, $1$, $2$, and $5$.}
\label{fig:SGImage}
\end{figure}

We define the Sierpinski gasket and the associated pre-fractal graph.
Let $a_0=(0,0)=:o$, $a_1=(\frac{1}{2}, \frac{\sqrt{3}}{2})$, and $a_2=(1,0)$ be the vertices of a unit equilateral triangle in $\mathbb{R}^2$, and $\mathfrak{G}_0$ be the complete graph on the vertex set $V_0 = \{a_0,a_1,a_2\}$, as shown on the left in Figure \ref{fig:SGImage}. 
We introduce three contracting similitudes $\Psi_i : \mathbb{R}^2 \to\mathbb{R}^2$, $\Psi_i(x) = \frac{1}{2}(x-a_i) + a_i $ for each $i\in \{0,1,2\}$. 
The Sierpinski gasket fractal $K$ is the unique nonempty compact set $K$ such that $K= \bigcup_{i=0}^2 \Psi_i(K)$.
To obtain the associated level-$n$ pre-fractal graph $\mathfrak{G}_n$, $n\geq 1$, we define by induction $\mathfrak{G}_n = \bigcup_{i=0}^2 \Psi_i(\mathfrak{G}_{n-1})$; see Figure \ref{fig:SGImage}.
To make all edges of the graph have unit length, we consider $G_n:= 2^n \mathfrak{G}_n$, where for $\alpha >0$ and $S\subset \mathbb{R}^2$ we denote $\alpha S:= \{ \alpha x: x\in S\}$.
The one-sided Sierpinski gasket graph $SG$ is then defined to be the infinite graph $G_\infty:=\bigcup_{n=0}^\infty G_n$, and
the double-sided gasket graph, $G_\infty \cup \mathfrak{R}(G_\infty)$, where $\mathfrak{R}$ is the reflection about $x_1=0$; see Figure \ref{fig:SG}. 

Given a set $U\subset \mathbb{R}^n$, let ${\rm diam}(U) := \sup\{|x-y|: x,y \in U\}$ denote the diameter of $U$ in the Euclidean metric.
The Hausdorff dimension of a set $F \subset \mathbb{R}^n$ is defined as
\begin{align}
d_H(F) := \inf\{s\geq 0: \mathcal{H}^s(F)=0\} \overset{\text{or}}{=}\sup\{s\geq 0: \mathcal{H}^s(F)=\infty\},
\end{align}
where
\begin{align}
\mathcal{H}^s(F) := \lim_{\delta\downarrow 0} \inf\left\{ \sum_{i=1}^\infty [{\rm diam}(U_i)]^s ~\bigg|~ \{U_i\}_i \text{ is a $\delta$-cover of } F\right\}
\end{align}
is the $s$-dimensional Hausdorff measure of $F$.
See \emph{e.g.\@} \cite[Chapter 2]{Falconer} for more details.
It is well known that the Hausdorff dimension of the Sierpinski gasket $K$ is $\frac{\log 3}{\log 2}$.

\subsection{Laplacian growth models and shape theorems on $SG$}

We now define the Laplacian growth models under study. All of them belong to the so-called \emph{abelian networks} introduced by Bond and Levine \cite{BondLevine}; for the rationale behind the abelian property we refer the reader to \cites{DF91, HLMPPW, BondLevine, LevinePeres17}.

In this subsection, $G$ is the one-sided gasket graph $SG$, and the ``origin'' $o$ is the corner vertex of $SG$; see again Figure \ref{fig:SG}.
Propositions \ref{prop:IDLAShapeThm} and \ref{prop:DSShapeThm}, which are the main theorems of \cites{IDLASG} and \cite{HSH17}, respectively, were proved on the double-sided $SG$. It is straightforward to modify the proof to work on the one-sided $SG$, which results in no change in the statement.

\subsubsection{Internal diffusion-limited aggregation (IDLA)}

Launch $m$ particles successively from $o$, and let each of them perform i.i.d.\@ random walks until reaching a site previously unvisited. 
Recall that a random walk on $G$ is a Markov chain on the state space $V(G)$  with infinitesimal generator $\Delta$.
The resulting (random) set of occupied vertices is called an IDLA cluster, denoted $\mathcal{I}(m)$.

\begin{proposition}[Shape theorem for IDLA on $SG$ \cite{IDLASG}*{Theorem 1.1}]
\label{prop:IDLAShapeThm}
For all $\epsilon>0$, we have
\[
B_o(n(1-\epsilon)) \subset \mathcal{I}(|B_o(n)|) \subset B_o(n(1+\epsilon))
\]
for all $n$ sufficiently large, with probability $1$.
\end{proposition}

\subsubsection{Divisible sandpiles}
\label{sec:DSShapeThm}

Divisible sandpiles were introduced by \cite{LevinePeres09}. Start with $m$ amount of sand at the origin $o$. Whenever the amount of sand $s(x)$ at vertex $x$ exceeds $1$, we topple the excess amount $s(x)-1$ and distribute it equally among the neighboring vertices $y\sim x$, \emph{i.e.,}
the resulting configuration is $s' = s + \max\left(s(x)-1,0\right)\Delta(x,\cdot)$.
Continue this procedure until the amount of sand is $\leq 1$ for all $x\in V(G)$, in which case we say that the sandpile has stabilized. 

Let $\mathcal{D}(m)$ denote the set of vertices which have toppled in the process, which we refer to as the \emph{divisible sandpile cluster}. 
Let $\bi{n} :=|B_o(n)| -\frac{1}{2}|\partial_I B_o(n)|$, where $\partial_I A = \{x\in A: \exists y\in A^c \text{ with } x\sim y\}$ is the inner boundary of a set $A \subset V(G)$.

\begin{proposition}[Shape theorem for divisible sandpiles on $SG$ \cite{HSH17}*{Theorem 1.1}]
\label{prop:DSShapeThm}
Let $n_m = \max\{k\geq 0: \bi{k} \leq m\}$. Then $B_o(n_m-1) \subset \mathcal{D}(m) \subset B_o(n_m)$.
\end{proposition}

\subsubsection{Rotor-router aggregation}

A rotor(-router) walk on a graph, introduced by Propp, is a derandomized version of a random walk on a graph. To begin, each vertex, $x$, is equipped with an arrow (rotor) which targets
its neighboring vertices in a periodic sequence. This periodic sequence is called a \emph{rotor mechanism}. A rotor mechanism is said to be \emph{simple} if each neighbor of $x$ occurs exactly once in a period.
A particle performing rotor walk first changes the rotor at the current position to point to the next neighbor, according to the simple periodic rotor mechanism, and then moves to the neighboring vertex the rotor points towards.
The rotor-router action $\rho$ on a particle configuration on $G$ is governed by the \textbf{stack Laplacian} $\Delta_\rho$ which, unlike the usual Laplacian $\Delta$, is a \emph{non}linear operator. See Definition \ref{def:stackLap} below.

In rotor-router aggregation, we launch $m$ rotor walks successively from the origin $o$, and let each of them perform rotor walks until reaching a site previously unvisited. We assume that each vertex carries a rotor mechanism which is periodic and simple.  
Let $\mathcal{R}(m)$ and $\sigma(m)$ denote, respectively, the set of vertices which have fired and the set of vertices occupied by the rotor walkers. 

\begin{theorem}[Shape theorem for rotor-router aggregation on $SG$]
Let $n_m=\max\{k\geq 0: \bi{k}\leq m\}$. Then for any periodic simple rotor mechanism,
\[
B_o(n_m-2) \subset \mathcal{R}(m) \subset B_o(n_m) \quad\text{and}\quad B_o(n_m-1) \subset \sigma(m) \subset B_o(n_m+1)
\]
for all $m\in\mathbb{N}$.
\label{thm:RRA}
\end{theorem}

\begin{remark}
When considered over all possible periodic simple rotor mechanisms, our Theorem \ref{thm:RRA} is sharp, namely, that the difference between the out-radius and the in-radius may equal, but never exceeds, $2$.
This has been confirmed by our simulations, see Figure \ref{fig:RRA}.
In fact, simulations indicate that the cluster growth tends to fill up the sphere of radius $n$ before entering the sphere of radius $n+1$, with occasional exceptions resulting in (temporary) outer$-$inner radial difference of $2$.
Moreover, the second-named author has numerical evidence \cite{AutomataSG} that starting from a nearly symmetric rotor configuration, assigning to all vertices the clockwise (resp.\@ counter-clockwise) rotor mechanisms, the radial difference appears never to exceed $1$. 
\end{remark}

\begin{figure}
\centering
\includegraphics[width=0.4\textwidth]{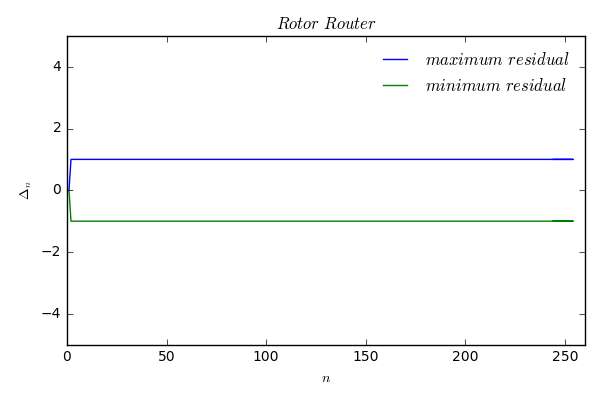}
\caption{Fluctuations of the rotor-router cluster $\sigma(|B_o(n)|)$ about the expected radius $n$, as a function of $n \in \{1,2,\cdots, 256\}$. For each $n$, 1000 instances of the rotor-router cluster associated with random initial rotor configurations were generated. The data points indicate the smallest possible out-radius (blue) and the largest possible in-radius (green) among the 1000 samples.}
\label{fig:RRA}
\end{figure}

\subsubsection{The abelian sandpile (or chip-firing) model} \label{sec:sandpileintro}
Start with $m$ chips at the origin $o$. Whenever the number of chips $\eta(x)$ at vertex $x$ equals or exceeds $\deg(x)$, we \textbf{topple} or \textbf{fire} at $x$ by sending one chip to each of the neighboring vertices $y\sim x$. 
Put in another way, toppling at $x$ on the sandpile configuration $\eta: V(G)\to \mathbb{N}_0$ produces the new configuration $\eta' =\eta + \Delta_G(x,\cdot)$.
Continue this procedure until the number of chips at $x$ is fewer than $\deg(x)$ for all $x\in V(G)$, in which case we say that the sandpile has stabilized. It is well-known that the order of topplings does not affect the stable configuration, which lends to the term ``abelian'' in the abelian sandpile model. 

Note that for the one-sided gasket graph, since $\deg(o)=2$, the sandpile configurations at mass $2k$ and $2k+1$, $k\in \mathbb{N}_0$, only differ at the origin $o$ (which carries $0$ and $1$ chip, respectively).
Therefore it suffices to study the growing sandpile cluster starting with an even number of chips at $o$.

Let $S(m)$ be the set of all vertices which have received at least a chip during the toppling process, which we refer to as the \textbf{receiving set}. This is to be distinguished from the \textbf{firing set} $A(m)$, the set of vertices which have fired at least once. Trivially $A(m)\subset S(m)$.

Recall $d_H=\frac{\log 3}{\log 2} = 1.58496\ldots$ is the Hausdorff dimension of $SG$.

\begin{theorem}[Shape theorem for the abelian sandpile cluster on $SG$]
The following hold for the abelian sandpile cluster on $SG$ with initial configuration $m\mathbbm{1}_o$:
\begin{enumerate}
\item \label{item:ball} For every $m\in \mathbb{N}$, there exists a radius $r_m\in \mathbb{N}_0$ such that $B_o(r_m-1)\subset A(m) \subset B_o(r_m) = S(m)$. (It is understood that $B_o(-1)=\emptyset$.)
\item \label{item:renewal} Let $r: [0,\infty)\to [0,\infty)$ be defined by $r(x) = r_{\lfloor x\rfloor}$. Then
\begin{align}
\label{eq:asymprm}
r(x) = x^{1/d_H} [\mathcal{G}(\log x)+o(1)] \quad \text{as}~x\to\infty,
\end{align}
where $\mathcal{G}$ is a nonconstant $(\log 3)$-periodic function having a finite number of well-defined discontinuities within each period (see Theorems \ref{thm:tail} and \ref{thm:radialcycle}, and Figure \ref{fig:tailpattern} below).
In particular, when $x\in [\frac{10}{9}, \frac{4}{3})$,
\begin{align}
\label{Gest}
\mathcal{G}(\log x) =\frac{1}{2}x^{-1/d_H} \in \left(0.4170,0.4679\right).
\end{align}
We also have the global estimate
\begin{align}
\label{eq:globalest}
0.3871 \leq \left(\frac{2}{9}\right)^{1/d_H}  \leq \mathcal{G}(\log x) \leq \frac{3}{4},
\end{align}
\end{enumerate}
\label{thm:ASM}
\end{theorem}

Figure \ref{fig:ASMScaling} shows the radius-to-mass scaling of Theorem \ref{thm:ASM}.

\begin{figure}
\centering
\begin{subfigure}[t]{0.5\textwidth}
\centering
\includegraphics[height=0.2\textheight]{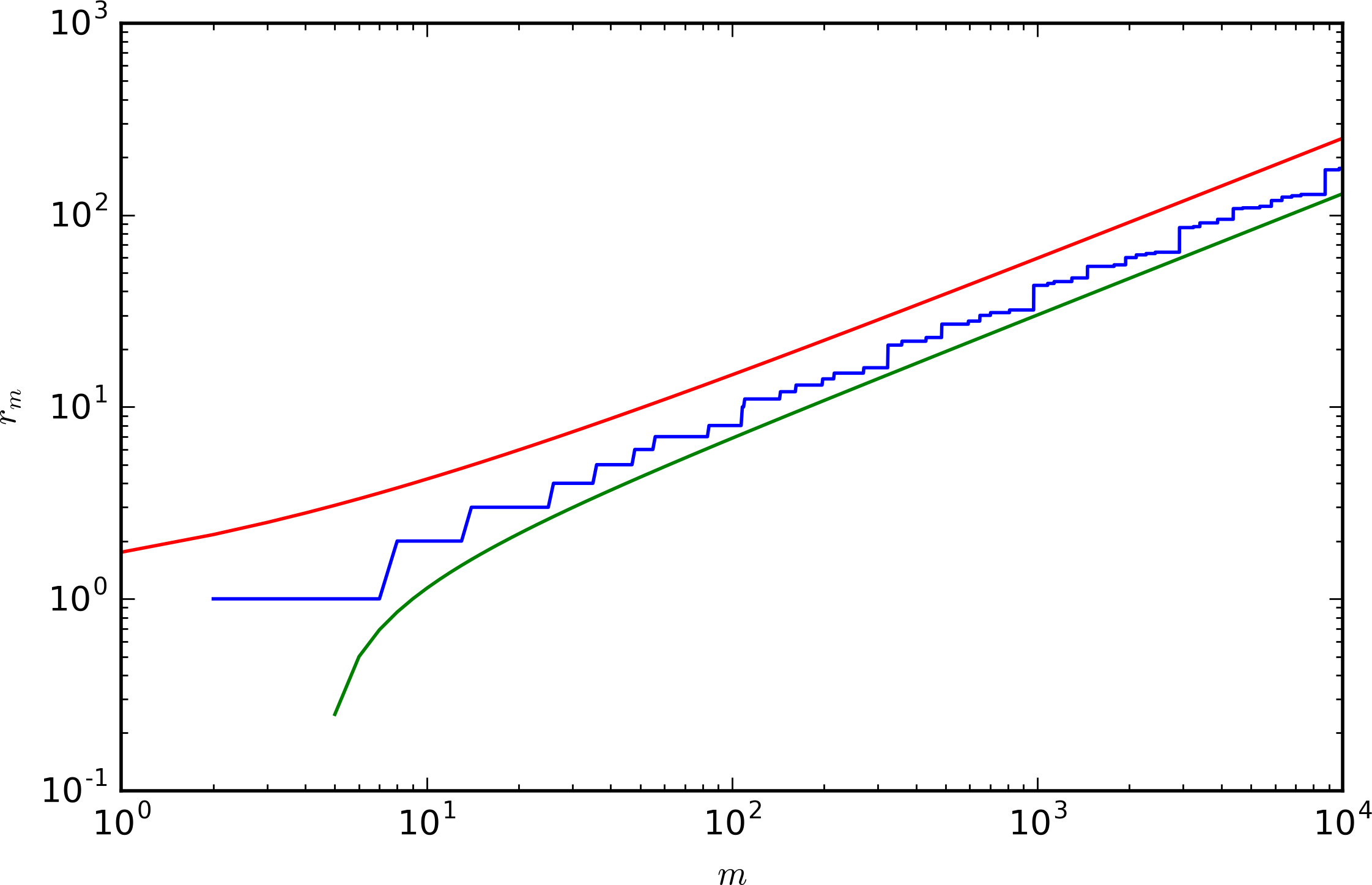}
\end{subfigure}%
~
\begin{subfigure}[t]{0.5\textwidth}
\centering
\includegraphics[height=0.225\textheight]{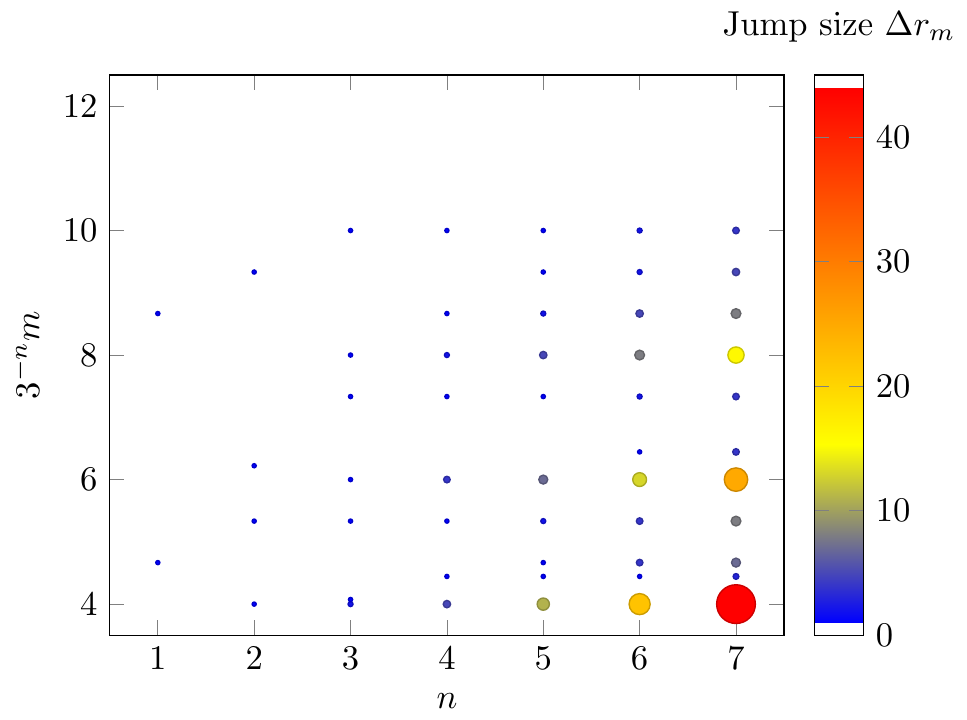}
\end{subfigure}
\caption{(Left) The radius $r_m$ of the abelian sandpile cluster $(m\mathbbm{1}_o)^\circ$ as a function of the initial mass $m$. Upper and lower bounds are given by the global estimate \eqref{eq:globalest}.
(Right) Jumps in $r_m$, \emph{cf.\@} Table \ref{table:spectrum}.
Enumeration of the radial jumps is given in Theorems \ref{thm:tail} and \ref{thm:radialcycle}; see also Figure \ref{fig:tailpattern}.
}
\label{fig:ASMScaling}
\end{figure}

\begin{remark}
\label{rem:sandpileSG}
As alluded to earlier, there have been previous works on the abelian sandpile model on $SG$ (or its variant).
One type of problem is to study sandpile dynamics under stationarity (\emph{a.k.a.\@} sandpile Markov chains) and obtain critical exponents of sandpile avalanches.
This was the focus of the numerical works \cites{Stanley96, DV98,DPV01} in the late 90s.
A rigorous analysis of sandpile height correlations was carried out by Matter \cite{MatterThesis}*{Chapter 5} on the Hanoi tower graph (Figure \ref{fig:Hanoi}), which is a different approximation of the Sierpinski gasket fractal.
The other type of problem is deterministic single-source sandpile growth, which was studied by Fairchild, Haim, Setra, Strichartz, and Westura \cite{ASMSGStr} and is the main focus of the present work.

\begin{figure}
\centering
\includegraphics[width=0.2\textwidth]{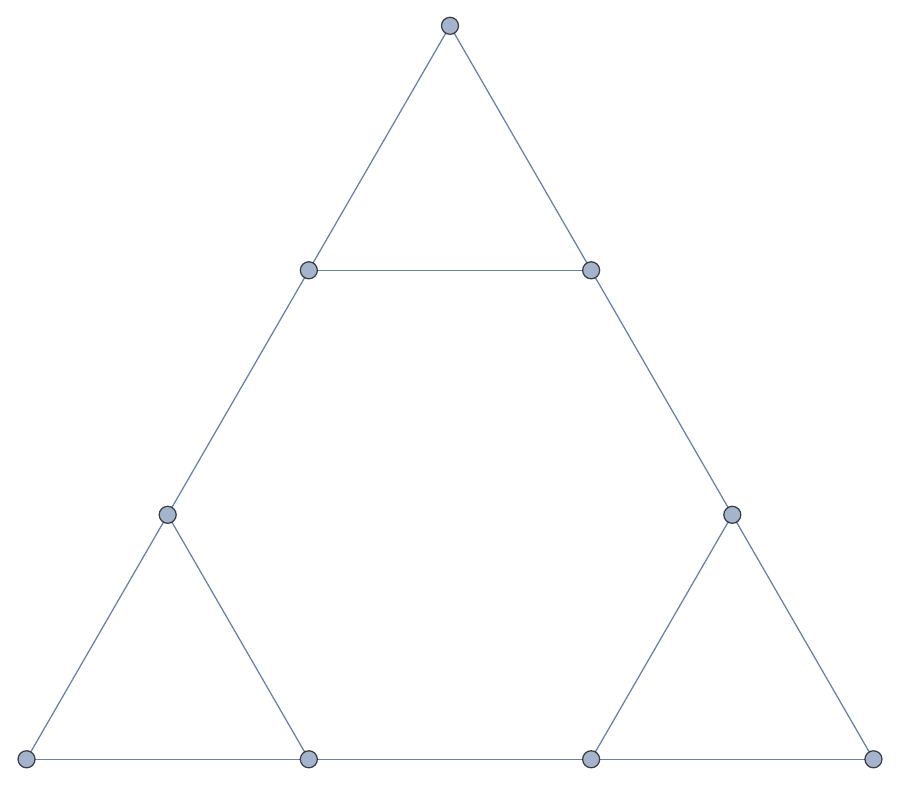}
\includegraphics[width=0.2\textwidth]{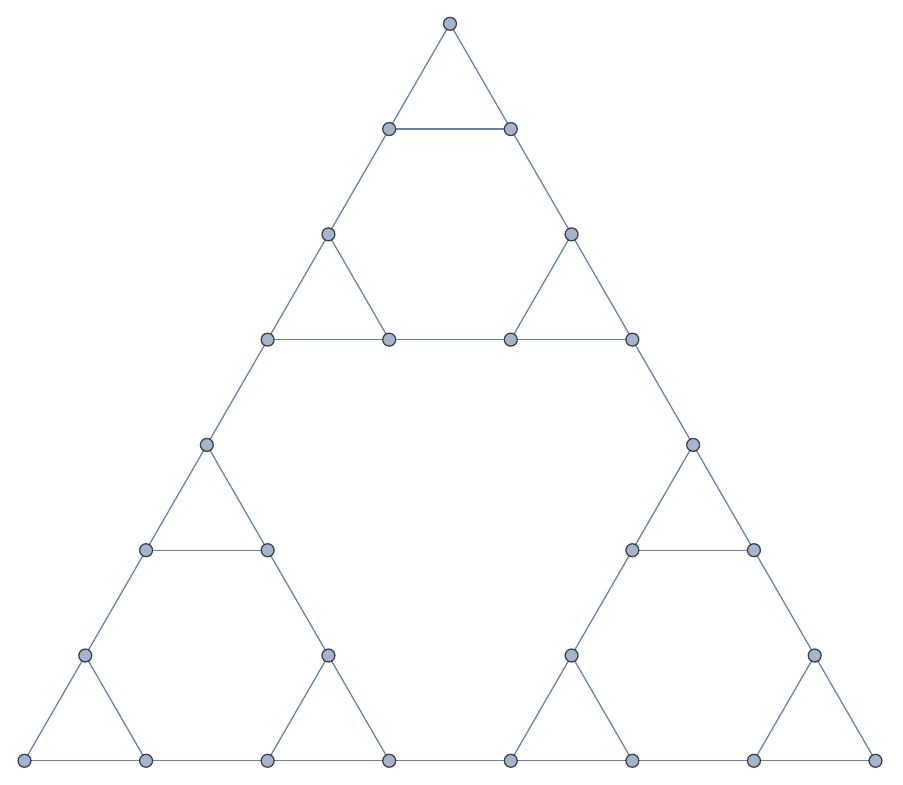}
\includegraphics[width=0.2\textwidth]{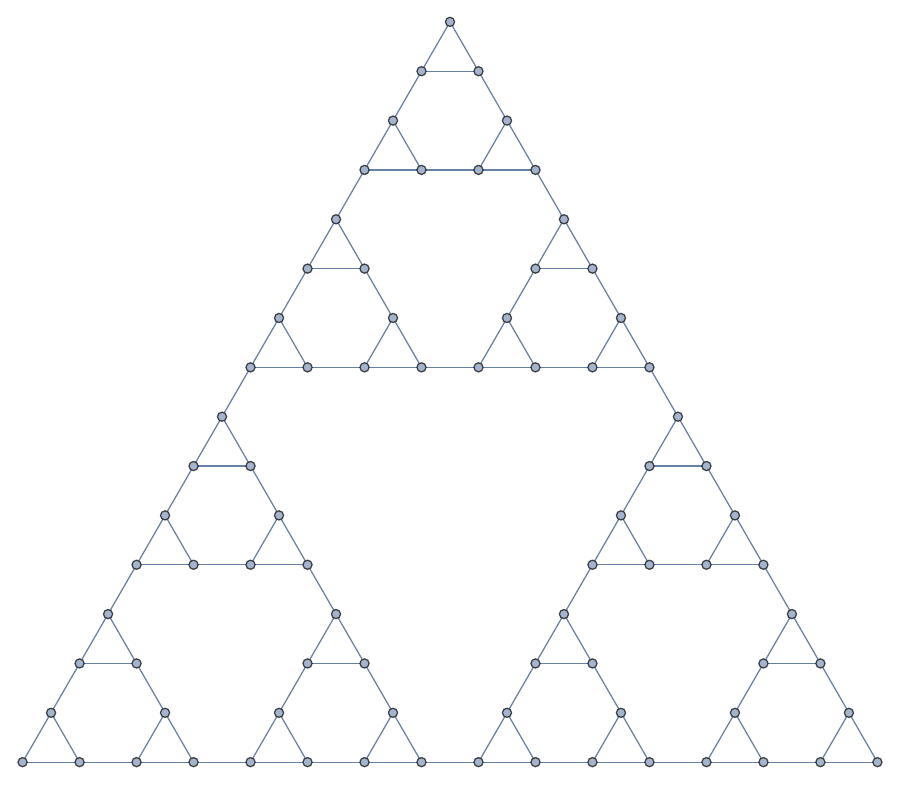}
\includegraphics[width=0.2\textwidth]{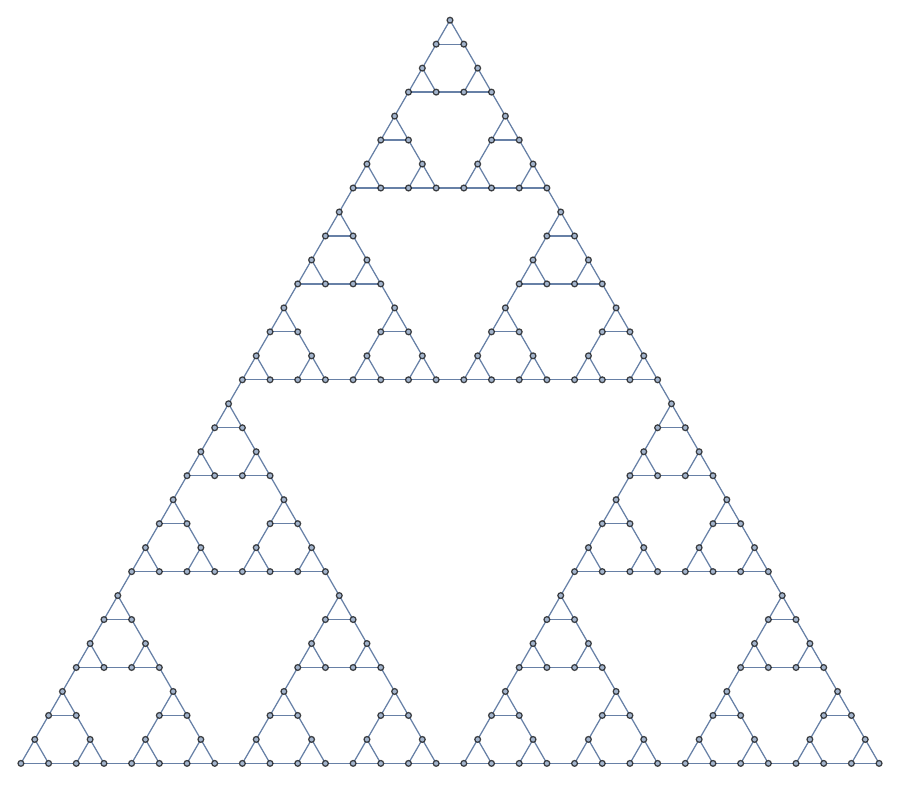}
\caption{The Hanoi-tower graphs of level 1, 2, 3, and 4.}
\label{fig:Hanoi}
\end{figure}

An expected feature appearing in both types of problems is that the relevant observable (be it the number of avalanche events, the cluster radius, etc.\@) asymptotically follows a \textbf{power law modulated by log-periodic oscillations}. 
This is a phenomenon expected on state spaces that possess discrete scale invariance, $SG$ being a prime example, although to our best knowledge there has been no rigorous proof prior to the present work.

In terms of cluster radial asymptotics, it was proved in \cite{ASMSGStr} that $r_m = \mathcal{O}(m^{1/d_H})$. 
Our Theorem \ref{thm:ASM} improves upon this result by showing rigorously that the radius follows a power law (with exponent $1/d_H$) modulated by a nonconstant log-periodic function $\mathcal{G}$.
This $\mathcal{G}$ function has a Fourier series representation with explicit Fourier coefficients, \emph{cf.\@} \eqref{eq:G}, and we can evaluate it in certain intervals of $x$, \emph{cf.\@} \eqref{Gest}.
\end{remark}

\subsection{Limit shape universality on $SG$}

Propositions \ref{prop:IDLAShapeThm} and \ref{prop:DSShapeThm} and Theorems \ref{thm:RRA} and \ref{thm:ASM} together imply the following ``limit shape universality'' result on $SG$, summarized in Table \ref{table:summary}.

\renewcommand*{\arraystretch}{1.2}
\begin{table}
\begin{tabular}{c|c}
\emph{Growth model} & \emph{Shape {\color{purple} theorem}/{\color{blue} conjecture}} \\ \hline \hline
IDLA & In/out-radius 
{\color{purple} $\left\{\begin{array}{ll}n\pm \mathcal{O}(\log n),& d= 2 \\ n\pm\mathcal{O}(\sqrt{\log n}), & d\geq 3 \end{array}\right\}$} \cites{LBG92, Lawler95, AG1, AG2, JLS1, JLS2} \\ \hline
Rotor-router aggregation & 
\begin{tabular}{@{}c@{}}{\color{purple} In-radius $n-c\log n$, out-radius $n+c'\log n $} \cites{LevinePeres09, LevinePeres17} \\ ($c,c'$ indep of $n$)\end{tabular}  \\ \hline
Divisible sandpiles & 
\begin{tabular}{@{}c@{}}{\color{purple} In-radius $n-c$, out-radius $n+c'$} \cite{LevinePeres09} \\ ($c,c'$ indep of $n$) \end{tabular} \\ \hline
Abelian sandpiles & \begin{tabular}{@{}c@{}} {\color{purple} Limit shape is not spherical}, {\color{blue} appears to be a polygon} \\
  {\color{purple}  Rigorous outer/inner spherical bounds (with a gap)} \cites{LevinePeres09,FLP10}
\end{tabular}
\end{tabular}
\vspace{10pt}
\caption{A summary of shape results for Laplacian growth models on $\mathbb{Z}^d$, starting with $|B_o(n)|$ particles at the origin $o$.}
\label{table:Zdshape}
\end{table}

\begin{table}
\centering
\begin{tabular}{c|c|c}
\emph{Growth model} & \emph{Initial \# of chips} & \emph{Shape {\color{purple} theorem}/{\color{blue} conjecture}} \\ \hline \hline
IDLA  & $|B_o(n)|$ & In/out-radius ${\color{purple} n} {\color{blue} \pm \mathcal{O}(\sqrt{\log n})}$ \cite{IDLASG}, [$\Diamond$] \\ \hline
Rotor-router aggregation & $m$ & {\color{purple} In-radius $n_m-2$, out-radius $n_m$} [$\Diamond$] \\ \hline 
Divisible sandpiles  & $m$ &
{\color{purple} In-radius $n_m-1$, out-radius $n_m$} \cite{HSH17}  \\ \hline
Abelian sandpiles & $m$ & \begin{tabular}{@{}c@{}} {\color{purple} Receiving set $S(m) = B_o(r_m)$} \\  {\color{purple} $r_m=m^{1/d_H}[\mathcal{G}(\log m)+o(1)]$ as $m\to\infty$} [$\Diamond$] \\
($\mathcal{G}$ is an explicit $(\log 3)$-periodic function)
\end{tabular}
\end{tabular}
\vspace{10pt}
\caption{A summary of shape results for Laplacian growth models on $SG$, with a single source at the origin $o$. [$\Diamond$] denotes results addressed in the present paper. Here $n_m=\max\{k\geq 0: |B_o(k)| - \frac{1}{2}|\partial_I B_o(k)|\leq m\}$, where $\partial_I A = \{x\in A: \exists y\in SG\setminus A,~ y\sim x\}$ denotes the inner boundary of $A\subset SG$.}
\label{table:summary}
\end{table}

\begin{theorem}[Limit shape universality on $SG$]
\label{thm:shapeuniv}
On $SG$, the four single-source Laplacian growth models---IDLA, rotor-router aggregation, divisible sandpiles, and abelian sandpiles---launched from the corner vertex $o$ all fill balls in the graph metric centered at $o$. 
\end{theorem}

As far as we are aware, this is the first non-tree state space whereupon the limit shapes of all four Laplacian growth models have been rigorously proven to coincide.
The fact that the abelian sandpile cluster on $SG$ is a ball is not entirely trivial, but nonetheless can be established easily via induction.
On other lattices and graphs the sandpile limit shape problem is largely open.


\subsection{Key ideas behind the proofs}

Our proof methods combine ideas and tools from analysis, probability, algorithms, combinatorics, algebra, and geometry.
For the reader's convenience we summarize the key elements of our proofs in this subsection.

\subsubsection{From divisible sandpiles to IDLA and rotor-router aggregation}
\label{sec:harmonicmeasure}

Let $(X_t)_{t\geq 0}$ be a continuous-time random walk on $G$ with infinitesimal generator $\Delta$, and $\mathbb{P}_o$ denote its law started at $o$.
Given a bounded subset $D \ni o$ of $V(G)$ which has boundary $\partial D$, one would like to estimate the hitting distribution of $X_t$ on $\partial D$ as it exits $D$.
This is known as the \textbf{harmonic measure} $\nu$ on $\partial D$:
\[ \nu(S) = \mathbb{P}_o[X_\tau \in S]  \qquad (S\subset \partial D)\] 
where $\tau=\inf\{t>0: X_t\in \partial D\}$.

In Laplacian growth models, the first particle occupies $o$, and inductively the $m$th particle is launched from $o$ and occupies the first vertex upon exiting the cluster $\Omega(m-1)$ formed by the first $m-1$ particles.
So the analysis involves the harmonic measure on the boundary $\partial \Omega(m-1)$, which evolves with $m$.
We need to make an educated guess of $\Omega(m)$ at special values of $m$, for example, when $m=|B_o(n)|$.

The strategy suggested by Levine and Peres \cite{LevinePeres09} is to first solve the divisible sandpile problem, \emph{i.e.,} the variational problem
\begin{align}
\label{eq:DSvar}
u_\infty(x) = \inf\left\{ w(x) ~|~ w: V(G)\to [0,\infty) \text{ satisfies } m\mathbbm{1}_o + \Delta w \leq 1\right\},
\end{align}
also known as the \emph{least action principle} for divisible sandpiles. Here $u_\infty$ is the divisible sandpile \textbf{odometer function} associated with the initial configuration $m\mathbbm{1}_o$, $\sigma_\infty= m\mathbbm{1}_o + \Delta u_\infty$ is the final configuration, and the divisible sandpile cluster $D(m)$ is the support of $u_\infty$.
See \cite{LevinePeres17}*{\S2} for a discussion.

In \cite{HSH17} Huss and Sava-Huss solved the divisible sandpile problem on $SG$.
Specifically they showed that at $m=\bi{n}$, the odometer function $u_\infty$ is the unique solution of a Dirichlet problem on the ball $D(\bi{n})=B_o(n)$ with boundary $\partial_I B_o(n)$.
See \cite{HSH17}*{Theorem 4.2} or Lemma \ref{lem:divisibleodometer} below.

The key philosophy to follow is that the divisible sandpile cluster gives a very good approximation of the bulk of the cluster in rotor-router aggregation and in IDLA.
This is best illustrated by the exact algorithm of Friedrich and Levine \cite{FL13} for fast simulations of large-scale growth models, see Algorithm \ref{alg:FL} below. 
In this algorithm, one can input any odometer function to produce an (incorrect) configuration, and then correct it by successive firings and unfirings and reverse cycle-popping.
By using the divisible sandpile odometer $u_\infty$ as the input odometer, the ensuing error corrections predominantly involve rotors near the boundary of the putative cluster.

This idea, combined with the fact that the harmonic measure on spheres in $SG$ is uniform, is central to our proof of the rotor-router shape Theorem \ref{thm:RRA}. At $m=\bi{n}$, we use the corresponding divisible sandpile odometer $u_\infty$ as the input to the Friedrich-Levine algorithm, and obtain an outer bound for the rotor-router cluster $\mathcal{R}(\bi{n})$. 
We then take advantage of the structure of $SG$ and the constancy of $u_\infty$ along the cluster boundary $\partial_I B_o(n-1)$ to carry out the error corrections exactly, leading to the inner bound.

Computationally this algorithm can also be implemented for IDLA, but analytically it appears not as useful. Indeed, one has to take into account randomization of rotors in calculating the \emph{random} approximate odometer, and then carry out the \emph{random} error corrections, neither of which is straightforward to analyze. 
To our best knowledge, the best techniques for analyzing IDLA are still based on those in \cites{LBG92, Lawler95}, combined with properties of random walks on graphs (Green's function estimates, elliptic Harnack inequality). 
In particular one needs to prove a mean-value inequality for the Dirichlet Green's function over balls, which can be obtained from solving the divisible sandpile problem.
In \cite{IDLASG} the first-named author, Huss, Sava-Huss, and Teplyaev adopted this idea to prove the inner bound of the IDLA cluster (Proposition \ref{prop:IDLAShapeThm}); see \cite{IDLASG}*{Section 3.1}.
The inner bound was then used to prove a matching outer bound using arguments similar to \cite{DCLYY} .

\subsubsection{Abelian sandpile growth}

The analysis of the abelian sandpile model is carried out differently from the other three growth models.
Since the model is predicated upon the \textbf{integrality} of the sandpile configuration (or height function), $\eta: V(G) \to \mathbb{N}_0$, we cannot directly apply tools from analysis of real-valued functions on fractals \cites{BarlowStFlour, KigamiBook, StrichartzBook}.
Instead, we analyze integer-valued functions on subgraphs of $SG$ endowed with suitable boundary conditions (sinks),
which brings us to the notion of a \textbf{sandpile group}.
For a modern introduction to this subject, see the excellent surveys \cite{Jarai} and \cite{PPW13}, as well as references therein.

\emph{A quick primer on the sandpile group.}
Let $G=(V \cup \{s\},E)$ be a finite, connected, undirected graph with a distinguished vertex $s$ (or possibly a set of distinguished vertices identified together) called the \textbf{sink}.
A sandpile configuration on $G$ is a map $\eta: V\to \mathbb{N}_0$. 
\textbf{Toppling} the configuration $\eta$ at the vertex $x\in V$ produces the new configuration $\eta' = \eta + \Delta'_G(x,\cdot)$, where $\Delta'_G$ is the combinatorial graph Laplacian:
\begin{align}
\label{def:Lap}
\Delta'_G(x,y) = 
\left\{\begin{array}{ll} -\deg_G(x), & \text{if } x=y \in V, \\ \mathsf{N}_{xy}, &\text{if } x\neq y,~x,y\in V. \end{array}\right.
\end{align}
Here $\deg_G(x)$ is the degree of the vertex $x$ in $G$.
Note that the sink $s$ plays a distinguished role in that chips that fall into $s$ are lost and do not return to $V$.
Accordingly the Laplacian $\Delta_G'$ is endowed with Dirichlet boundary condition on $s$, and is notated differently from the aforementioned Laplacian $\Delta_G$.
(That said, in our proofs to follow, it is necessary to keep track of the number of chips received by the sink after stabilization.) 

A configuration $\eta$ is \textbf{stable} if $\eta(x)<\deg_G(x)$ for all $x\in V$.
Denote $\Omega_G$ as the set of all stable configurations on $G$.
If $\eta$ is unstable, we can stabilize it by executing successive (legal) topplings until it reaches a unique stable configuration $\eta^\circ \in \Omega_G$.
This is guaranteed by the existence of the sink $s$.
Let us define the binary operation of pointwise addition of two stable configurations followed by stabilization
\[
\oplus: \Omega_G \times \Omega_G \to \Omega_G, \quad \eta \oplus \xi = (\eta+\xi)^\circ.
\]
This makes $(\Omega_G, \oplus)$ into a commutative monoid.

Now define a Markov chain on $\Omega_G$ with transitions
\[
\eta \to \eta\oplus \mathbbm{1}_x \quad \text{ with probability } p(x),
\]
where $p(x)>0$ for all $x\in V$ and $\sum_{x\in V}p(x)=1$.
Using the standard Markov chain language, we say that $\eta\in \Omega_G$ is \textbf{recurrent} if starting from $\eta$, the Markov chain returns to $\eta$ with probability $1$. The following facts are well-known \cite{Jarai}*{\S2}: there is exactly one recurrent communication class $\mathcal{R}_G$ in $\Omega_G$, and that $\eta\in \mathcal{R}_G$ if and only if
for any sandpile $\sigma$, there exists a sandpile $\zeta$ such that $\eta= \sigma\oplus\zeta$. One may check membership in $\mathcal{R}_G$ using the \textbf{burning test} of Dhar and Majumdar \cites{Dhar90, MajumdarDhar}, see \emph{e.g.\@} \cite{Jarai}*{\S4.1} for description of the \textbf{burning bijection} between recurrent configurations in $\mathcal{R}_G$ and spanning trees on $G$ rooted at $s$.

Taking $\mathbb{Z}^V$ as an abelian group, the integer row span $\mathbb{Z}^V \Delta_G'$ of $\Delta_G'$ forms a subgroup of $\mathbb{Z}^V$.
Put in another way, we define an equivalence relation on $\mathbb{Z}^V$ by declaring that 
\[
\xi \sim \zeta \quad \Longleftrightarrow \quad \xi-\zeta \in \mathbb{Z}^V \Delta_G',
\]
that is, two configuratons are equivalent if one can be obtained from the other via successive (possibly illegal) topplings.
The equivalence classes under $\sim$ form an abelian group $K_G:= \mathbb{Z}^V / \mathbb{Z}^V \Delta_G'$, which is called the \textbf{sandpile group} of $G$.
Each equivalence class in $K_G$ corresponds to one and exactly one recurrent configuration in $\mathcal{R}_G$.
In other words, $(\mathcal{R}_G, \oplus)$, which is the minimal ideal of $(\Omega_G, \oplus)$, forms an abelian group which is isomorphic to $K_G$.
As a finite abelian group, $K_G$ can be expressed as the direct sum of cyclic groups.
For a systematic discussion of the sandpile group, including its computation via the Smith normal form of $\Delta_G'$, see \cites{Biggs, Lorenzini}. 

\emph{Dhar's multiplication by identity test and its application.}
The next result is due to Dhar \cite{Dhar90}.
\begin{lemma}[Multiplication by identity test]
\label{lem:mit}
Let $\eta \in \mathcal{R}_G$. Then $\eta \oplus \sum_{y\in V} \mathsf{N}_{sy} \mathbbm{1}_y = \eta$, and each vertex topples exactly once upon stabilization.
\end{lemma}
\begin{proof}
Using \eqref{def:Lap} we have that for each $y\in V$,
\[
\mathsf{N}_{sy} + \sum_{x\in V}\Delta_G'(x,y) = \mathsf{N}_{sy} + \left( \sum_{\substack{x\in V\\x\neq y}} \mathsf{N}_{xy} - \deg_G(y) \right)= 0.
\]
This implies that on a recurrent configuration $\eta$, after we add $\mathsf{N}_{sy}$ chips to each vertex $y\in V$, and then topple once at every vertex, the same configuration $\eta$ is returned. 
\end{proof}

We now explain how Lemma \ref{lem:mit} is applied to the state space $SG$ (the same idea also works on other nested fractal graphs, see \cite{ASMSGStr}).
Let $G_n$ be the level-$n$ pre-fractal Sierpinski gasket graph; see Figure \ref{fig:SGImage}, and observe the three corner vertices $o$ (origin), $x$, and $y$.
We denote by $\partial G_n=\{x,y\}$ the set of two non-origin corner vertices which are distance $2^n$ from $o$ in the graph metric.
The two types of sinked graphs we consider are:
$G_n^{(s)}$, where we designate $\partial G_n$ as the sink; and $G_n^{(o)}$, where we designate $o$ as the sink. 
We denote the set of recurrent configurations on $G_n^{(s)}$ (resp.\@ $G_n^{(o)}$) by $\mathcal{R}_n^{(s)}$ (resp.\@ $\mathcal{R}_n^{(o)}$).

In the process of stabilizing $m\mathbbm{1}_o$, chips will occupy the graph $G_n$ for some $n$.
Instead of stabilizing all vertices in $SG$ at once, we can use the abelian property to stabilize only the vertices in $G_n \setminus \partial G_n$, and pause any excess chips on $\partial G_n$.
(By the axial symmetry of the initial configuration, an equal number of chips will reach $x$ and $y$.)
This creates a recurrent configuration $\eta \in \mathcal{R}_n^{(s)}$ on $G_n^{(s)}$.
Now we topple at the two cut vertices $x$ and $y$ in $\partial G_n$---think of them as sinks of $G_n^{(s)}$, but sources for producing the ``tail'' configuration in $(G_n)^c$---which leads to 1 chip being added to each vertex adjacent to $\partial G_n$.
By Lemma \ref{lem:mit}, we can then topple at every vertex in $G_n\setminus \partial G_n$, and this results in no change in the configuration $\eta$ on $G_n^{(s)}$.
Overall, $x$ and $y$ each loses 2 chips to the vertices in the ``tail'' $(G_n)^c$.
Continue this process until a stable configuration is reached.

In summary, we obtain the fundamental diagram
\begin{align}
\label{eq:FD}
\begin{array}{lm{1.1in}l}
(m \mathbbm{1}_o)^\circ = 
&
\includegraphics[height=0.15\textwidth]{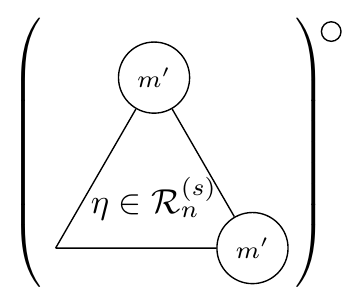}
&
\subseteq G_{n+1}.
\end{array}
\end{align}
which means that the cluster first fills $G_n$ with $m'$ chips paused at each vertex in $\partial G_n$, and then upon full stabilization the cluster fills a subset of $G_{n+1}$.
Based on \eqref{eq:FD}, we prove in Proposition \ref{prop:ball} below that the receiving set of $(m\mathbbm{1}_o)^\circ$ is always a ball $B_o(r_m)$, and for $m\geq 12$, $r_m =2^n + r_{m'-2}$.

\emph{Structural theorems on the sandpile groups of subgraphs of $SG$.}
Using the preceding ideas we obtain, via induction on $n$, a number of structural results on the sandpile groups $\mathcal{R}_n^{(s)}$ and $\mathcal{R}_n^{(o)}$.
In the diagrams to follow throughout the paper, $\bullet$ represents a sink vertex, and long arrows (or long dashed lines) are used as visual devices for indicating orientations. 

\begin{definition}
\label{def:en}
Let
\begin{tabular}{m{0.8in}lm{0.8in}l}
 \includegraphics[width=0.135\textwidth]{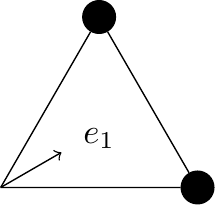}
&
$=$
&
 \includegraphics[width=0.135\textwidth]{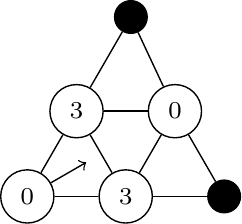}
 &,
\end{tabular}
and for each $n \geq 1$, $e_{n+1}$ is constructed by gluing three copies of $e_n$ according to the rule

\begin{center}
\begin{tabular}{m{1.3in}lm{1.3in}l}
\includegraphics[height=0.2\textwidth]{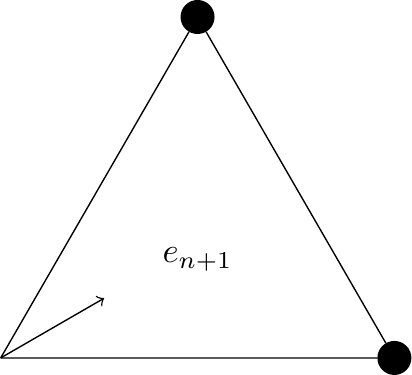}
&
$=$
&
\includegraphics[height=0.2\textwidth]{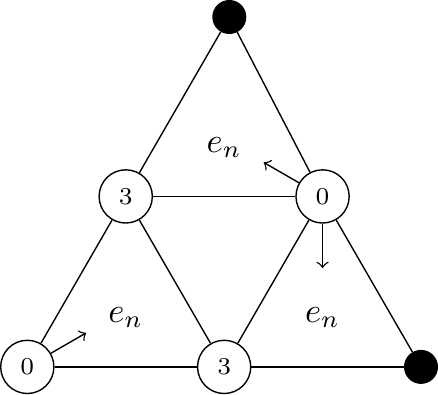}
&
.
\end{tabular}
\end{center}
\end{definition}

\begin{definition}
\label{def:Mn}
Let
\begin{tabular}{m{0.8in}lm{0.8in}l}
 \includegraphics[width=0.135\textwidth]{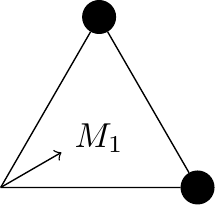}
&
$=$
&
 \includegraphics[width=0.135\textwidth]{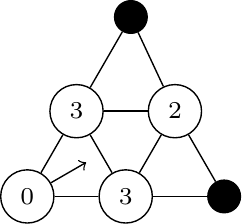}
 &,
\end{tabular}
and for each $n \geq 1$, $M_{n+1}$ is constructed by gluing three copies of $M_n$ according to the rule

\begin{center}
\begin{tabular}{m{1.3in}lm{1.3in}l}
\includegraphics[height=0.2\textwidth]{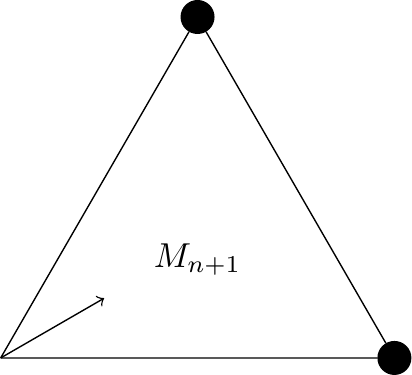}
&
$=$
&
\includegraphics[height=0.2\textwidth]{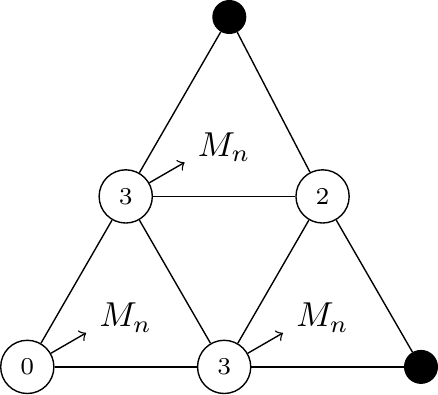}
&
.
\end{tabular}
\end{center}

Let $e_n^{(o)}$ denote the configuration 
\begin{tabular}{m{0.8in}}
 \includegraphics[width=0.135\textwidth]{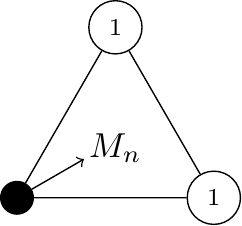}
\end{tabular}
.
\end{definition}

\begin{theorem}
\label{thm:groupSG}
We have the following identities:
\begin{enumerate}
\item \label{item:identity} $e_n$ and $e_n^{(o)}$ are, respectively, the identity element of $(\mathcal{R}_n^{(s)}, \oplus)$ and $(\mathcal{R}_n^{(o)}, \oplus)$.
\item \label{item:tops} For every $\eta \in \mathcal{R}_n^{(s)}$, $\eta \oplus (2\cdot 3^n) \mathbbm{1}_o = \eta$.
\item \label{item:topo} Let $\partial G_n=\{x,y\}$. For every $\eta\in \mathcal{R}_n^{(o)}$,
\[
\eta \oplus 3^n(\mathbbm{1}_x + \mathbbm{1}_y) = \eta,
\quad
\eta  \oplus 3^{n+1}\mathbbm{1}_x = \eta,
\quad\text{and}\quad
\eta  \oplus 3^{n+1}\mathbbm{1}_y = \eta.
\]
\end{enumerate}
\end{theorem}

See Figure \ref{fig:identities} for pictures of the identity elements, and observe the tiling construction.

\begin{figure}
\begin{tabular}{M{0.35\textwidth} M{0.35\textwidth} M{0.07\textwidth}}
\includegraphics[width=0.3\textwidth]{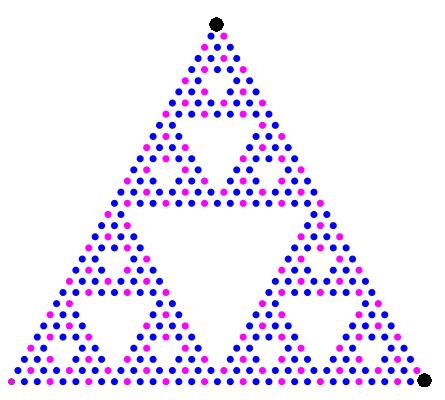}
&
\includegraphics[width=0.3\textwidth]{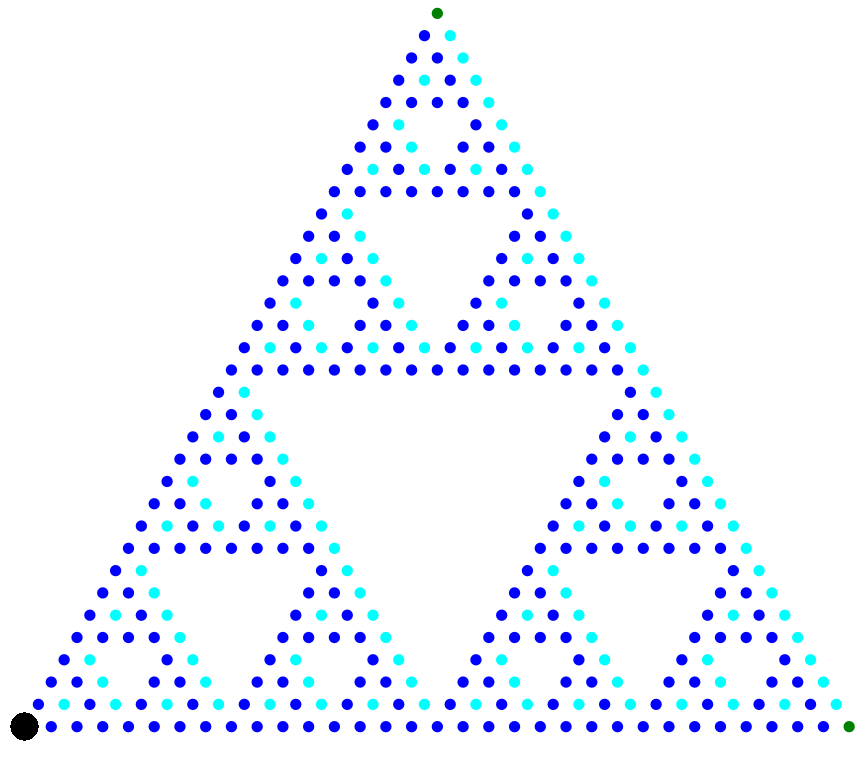}
&
\includegraphics[width=0.05\textwidth]{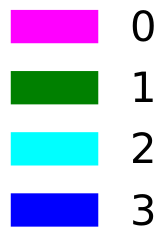}
\end{tabular}
\caption{The identity elements $e_5 \in \mathcal{R}_5^{(s)}$ and $e_5^{(o)} \in \mathcal{R}_5^{(o)}$.}
\label{fig:identities}
\end{figure}

\begin{remark}
In \cite{ASMSGStr} the authors numerically obtained the identity elements $e_n$ for $n=1,2,3$ via an application of Lemma \ref{lem:mit} to $\partial G_n$. They also computed the sandpile group of $G_n$ for $n=1,2,3$, with all three corner vertices $\{o,x,y\}$ identified as sink.
\end{remark}

\begin{remark}
Since Theorem \ref{thm:groupSG} suffices for our purposes, we do not pursue a full characterization of the sandpile group in the present work, though it is of interest to further investigate the underlying self-similar structure.
\end{remark}

\emph{Radial explosions, periodicity, and the exact solution of the sandpile growth problem.}
A somewhat surprising feature on $SG$ is that the configurations $M_n$ and $e_n$ appear in the stabilizations of $m\mathbbm{1}_o$ periodically in $m$.
We will show in \S\ref{sec:explosion} below that
\begin{center}
\begin{tabular}{rm{.85in}l}
$((4\cdot 3^n-2) \mathbbm{1}_o)^\circ = $
&
\includegraphics[height=.8in]{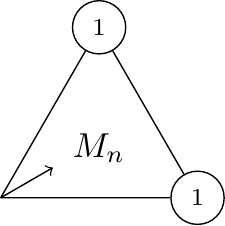}
&; and\\
$((4\cdot 3^n) \mathbbm{1}_o)^\circ = $
&
\includegraphics[height=.85in]{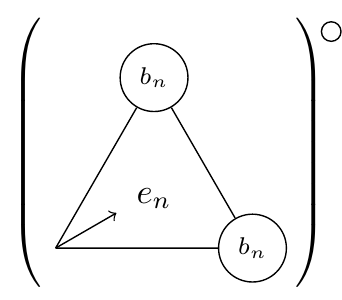}
&, where $b_n= |V(G_{n-1})|=\frac{3}{2}\left(3^{n-1}+1\right)$.
\end{tabular}
\end{center}
The proof of $M_n$ is not difficult.
The proof of $e_n$, on the other hand, is subtle, and uses the cut point structure and the axial symmetry of $SG$. See the axial reflection lemma in \S\ref{sec:reflection}.

The significance of this result is that the sandpile cluster ``explodes''---the cluster radius increments by more than $1$---at mass $4\cdot 3^n$.
Using the toppling identities stated in Theorem \ref{thm:groupSG}, we further deduce radial explosions at mass $6\cdot 3^n$, $8\cdot 3^n$, and $10\cdot 3^n$, \emph{i.e.,} the $(2\cdot 3^n)$-periodicity which is evident from Figure \ref{fig:ASMScaling}.
Details are given in \S\ref{sec:explosion}.

\begin{remark}
The $(2\cdot 3^n)$-periodicity of sandpile growth on $SG$ was already conjectured by the authors of \cite{ASMSGStr}; see their Conjecture 16, where they stated the $(4\cdot 3^n)$-periodicity on the double-sided $SG$.
The complete enumeration of radial jumps up to $n=6$, \emph{cf.\@} the right-hand diagram in Figure \ref{fig:ASMScaling}, was first obtained numerically by the second-named author.
The first-named author then discovered further patterns based on the diagram, as well as the results in \cite{ASMSGStr}, to write down the proofs of the periodicity.
\end{remark}

To describe all the other radial jumps in Figure \ref{fig:ASMScaling}, we need to identify the function $m\mapsto m'$ in the fundamental diagram \eqref{eq:FD}.
This jump function turns out to be well-defined for $m\geq 4\cdot 3^3$.

\begin{theorem}
\label{thm:tail}
For $n\geq 3$ and $m\in [4\cdot 3^n, 4\cdot 3^{n+1})$,
\begin{tabular}{lm{0.9in}}
$(m \mathbbm{1}_o)^\circ = $
&
\includegraphics[height=0.13\textwidth]{tablefig}
\end{tabular}
$\subseteq G_{n+1}$, where $m\mapsto m'$ is a piecewise constant right-continuous function which has jumps indicated in the following table:

\begin{center}
\begin{tabular}{c|c|c}
$m$ & $m'$ & \emph{Location of statement \& proof} \\ \hline
$(4+2p)\cdot 3^n$ & $b_n+ p\cdot 3^n$ & \emph{``$4$'' (Prop.\@ \ref{prop:merge})} \\ 
$(4+2p)\cdot 3^n+2$ & $(b_n+1)+p\cdot 3^n$ & \emph{``$e+2$'' (Prop.\@ \ref{prop:4+2}), ``$4\frac{4}{9}^-$'' (Prop.\@ \ref{prop:449-})} \\
$(4\frac{4}{9}+2p)\cdot 3^n$ & $2\cdot 3^{n-1}+1+p\cdot 3^n$ &\emph{``$4\frac{4}{9}$'' (Prop.\@ \ref{prop:449}), ``$4\frac{2}{3}^-$'' (Prop.\@ \ref{prop:423-})} \\
$(4\frac{2}{3}+2p)\cdot 3^n$ & $2\cdot 3^{n-1}+2+p\cdot 3^n$ & \emph{``$4\frac{2}{3}$'' (Prop.\@ \ref{prop:423}), ``$5\frac{1}{3}^-$'' (Prop.\@ \ref{prop:513-})} \\
$(5\frac{1}{3}+2p)\cdot 3^n$ & $3^n+1+p\cdot 3^n$ & \emph{``$5\frac{1}{3}$'' (Prop.\@ \ref{prop:513}), ``$6^-$'' (Prop.\@ \ref{prop:maxtopple})}
\end{tabular}
\end{center}
where $p\in \{0,1,2,3\}$, and $b_n = |V(G_{n-1})| =\frac{3}{2}(3^{n-1}+1)$. 
\end{theorem}

\begin{remark}
Regarding the statements in the right-most column,
``$a^-$'' and ``$a$'' stand for, respectively, ``just below the jump at $a$'' and ``at the jump at $a$.''
All statements refer to the case $p=0$. Once proved, we can establish the case $p\in \{1,2,3\}$ easily. 
We have indicated above where these results are stated and proved in \S\ref{sec:enumeration} and \S\ref{sec:enumeration2}.
\end{remark}

The proof of Theorem \ref{thm:tail} is fairly technical, and relies mainly upon the identification of sandpile tiles which, when glued together in a self-similar and symmetric way, produce the sandpile configuration on $G_n^{(s)}$, see Figure \ref{fig:tailpattern}.
In the physics parlance this procedure can be considered an \textbf{exact block renormalization}, where the blocks are the sandpile tiles. To our best knowledge this may be the first time an abelian sandpile problem is exactly solved using renormalization-type arguments. 
Detailed proofs are provided in \S\ref{sec:enumeration}.

That said, there are two exceptions, ``$e+2$'' and ``$4\frac{4}{9}^-$,'' where the block renormalization idea does not immediately apply. 
To tackle these two cases, we observe that in the identity element $e_n$, there is a unique shortest path connecting $o$ to each vertex in $\partial G_n$ along which every vertex carries $3$ chips, see Figure \ref{fig:4+2path}.
This path is the concatenation of the first $n$ iterations of the \textbf{Sierpinski arrowhead curve} \cite{Sierpinski}, and fills up half of $SG$.
Toppling once at $o$ triggers a chain of topplings down the path, and results in each vertex in $\partial G_n$ receiving an extra chip, which proves ``$e+2$.''
We then analyze the landscape of ``traps'' resulting from ``$e+2$'' along the path, and show that it requires adding at least $4\cdot 3^{n-2}$ chips at $o$ to deliver extra chips to $\partial G_n$, thereby establishing ``$4 \frac{4}{9}^-$.''
See \S\ref{sec:enumeration2} for details.

\begin{remark}[A take-home message]
In essence, we have just described two ways in which \emph{fractals} manifest themselves in the ``sandpile on a \emph{fractal}'': via self-similar tilings (block renormalization) and via the Sierpinski arrowhead curve.
It will be helpful to examine Figures \ref{fig:tailpattern} and \ref{fig:4+2path} while reading the proofs in \S\ref{sec:enumeration} and \S\ref{sec:enumeration2}.
We suspect the self-similar tiling idea should be applicable to sandpile growth on other fractal graphs.
But the appearance of the Sierpinski arrowhead curve in the identity element $e_n$, and the role it plays in the proofs below, seems to be unique to $SG$.

As a historical aside, Sierpinski introduced the arrowhead curve  \cite{Sierpinski} as a model of space-filling curves (which differ from those constructed by Peano \cite{Peano} and Hilbert \cite{Hilbert}, respectively) prior to his eponymous gasket.
It is a surprising coincidence that we rediscovered his arrowhead curve through the abelian sandpile identity element on his gasket.
\end{remark}

\renewcommand*{\arraystretch}{2}
\begin{figure}
\rotatebox{90}{
\begin{tabular}{>{\centering\arraybackslash}m{0.12in}|>{\centering\arraybackslash}m{1.1in} >{\centering\arraybackslash}m{1.1in}}
$m$ & $n=3$ & $n=4$ \\ \hline
\rotatebox{90}{$4\cdot 3^n$}
& \includegraphics[width=0.18\textwidth]{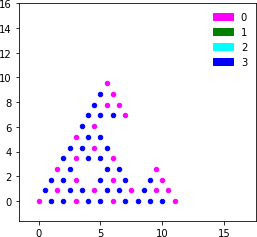}
& \includegraphics[width=0.18\textwidth]{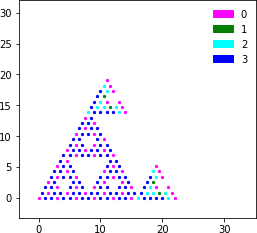}\\
\rotatebox{90}{$(4\cdot 3^n)+2$}
& \includegraphics[width=0.18\textwidth]{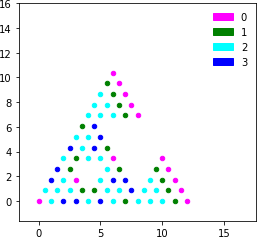} 
& \includegraphics[width=0.18\textwidth]{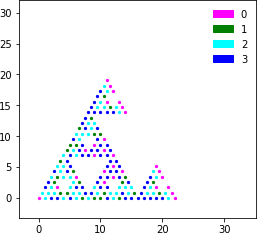}\\
\rotatebox{90}{$(4\frac{4}{9}\cdot 3^n)-2$ }
& \includegraphics[width=0.18\textwidth]{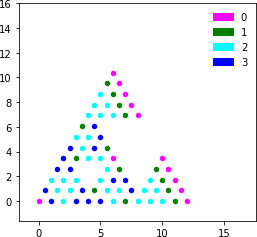}
& \includegraphics[width=0.18\textwidth]{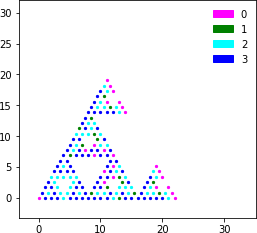}
\end{tabular}
\hspace{10pt}
\begin{tabular}{>{\centering\arraybackslash}m{0.12in}|>{\centering\arraybackslash}m{1.1in} >{\centering\arraybackslash}m{1.1in}}
$m$ & $n=3$ & $n=4$ \\ \hline
\rotatebox{90}{$4\frac{4}{9}\cdot 3^n$} 
& \includegraphics[width=0.18\textwidth]{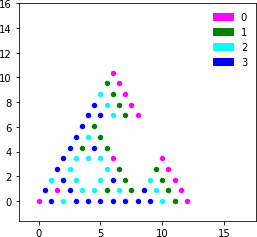} 
& \includegraphics[width=0.18\textwidth]{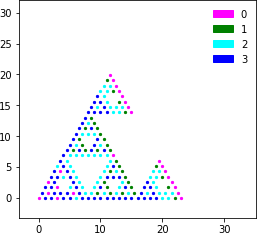}\\
\rotatebox{90}{$(4\frac{2}{3}\cdot 3^n)-2$} 
& \includegraphics[width=0.18\textwidth]{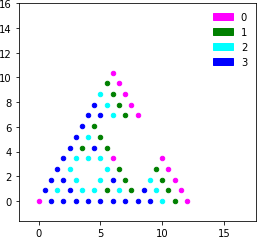} 
& \includegraphics[width=0.18\textwidth]{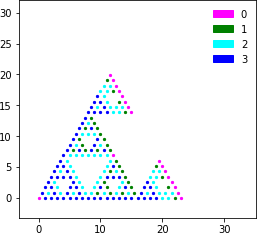}\\
\rotatebox{90}{$4\frac{2}{3}\cdot 3^n$} 
& \includegraphics[width=0.18\textwidth]{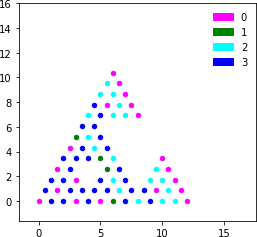} 
& \includegraphics[width=0.18\textwidth]{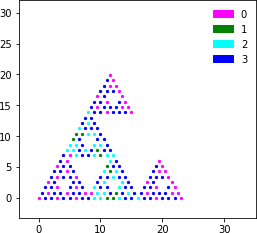}
\end{tabular}
\hspace{10pt}
\begin{tabular}{>{\centering\arraybackslash}m{0.12in}|>{\centering\arraybackslash}m{1.1in} >{\centering\arraybackslash}m{1.1in}}
$m$ & $n=3$ & $n=4$ \\ \hline
\rotatebox{90}{$(5\frac{1}{3}\cdot 3^n)-2$} 
& \includegraphics[width=0.18\textwidth]{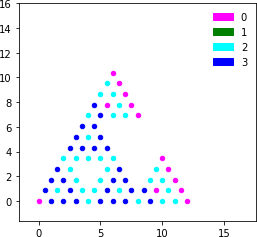} 
& \includegraphics[width=0.18\textwidth]{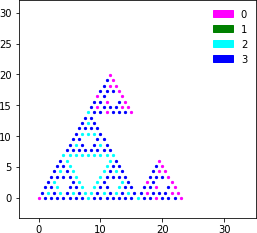}\\
\rotatebox{90}{$5\frac{1}{3}\cdot 3^n$} 
& \includegraphics[width=0.18\textwidth]{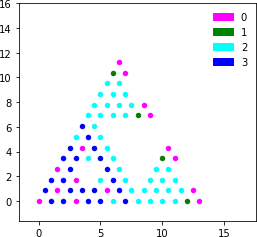} 
& \includegraphics[width=0.18\textwidth]{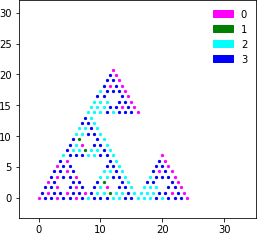}\\
\rotatebox{90}{$6\cdot 3^n-2$} 
& \includegraphics[width=0.18\textwidth]{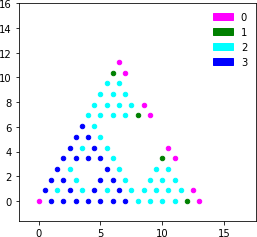} 
& \includegraphics[width=0.18\textwidth]{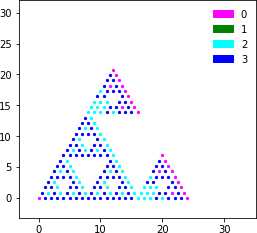}
\end{tabular}
}
\rotatebox{90}{
\centering
\begin{tabular}{c}
\includegraphics[trim=-8cm 0 0 -1cm, width=\textwidth]{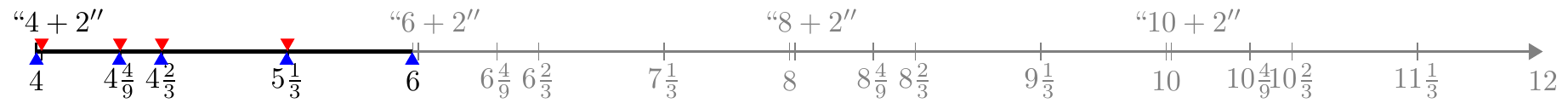}
\end{tabular}
}
\caption{The ``periodic table'' of single-source sandpile on $SG$: $(m\mathbbm{1}_o)^\circ$ at special values of $m$.
We only present the patterns for $m \in [4\cdot 3^n, 6\cdot 3^n)$.
Using the $(2\cdot 3^n)$ periodicity (Theorem \ref{thm:groupSG}, Item \eqref{item:tops}), we can infer patterns for $m\in [6\cdot 3^n, 12\cdot 3^n)$, as the timeline suggests. 
}
\label{fig:tailpattern}
\end{figure}

\begin{figure}
\centering
\begin{minipage}{0.45\textwidth}
\includegraphics[width=0.9\textwidth]{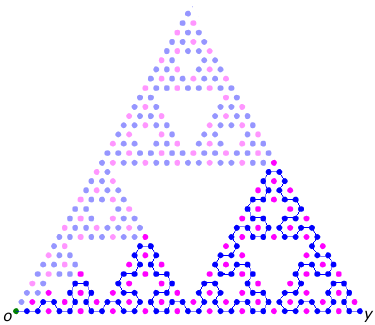}
\end{minipage}
\begin{minipage}{0.45\textwidth}
\begin{tabular}{cc}
\includegraphics[page=1,width=0.45\textwidth]{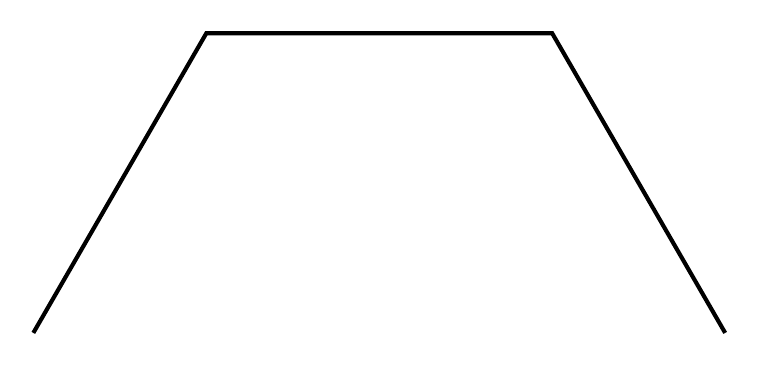} &
\includegraphics[page=2,width=0.45\textwidth]{arrowhead} \\
\includegraphics[page=3,width=0.45\textwidth]{arrowhead} &
\includegraphics[page=4,width=0.45\textwidth]{arrowhead} 
\end{tabular}
\end{minipage}
\caption{The configuration $e_5$ (left). Indicated in {\color{blue} blue} is the space-filling curve connecting $o$ to the sink vertex $y\in\partial G_5$ along which every vertex carries $3$ chips.
It is formed by the concatenation of the first 4 iterations of the Sierpinski arrowhead curve (right).}
\label{fig:4+2path}
\end{figure}

\renewcommand*{\arraystretch}{1.2}

From Theorem \ref{thm:tail}, we obtain a family of recursive formulae for the cluster radii, which completely explains the results shown in Figure \ref{fig:ASMScaling}, and is used in conjunction with the renewal theorem to prove Theorem \ref{thm:ASM}, Part \eqref{item:renewal}.
This is the content of our final Theorem \ref{thm:radialcycle}, which is understood mnemonically using Figure \ref{fig:radialcycle}.

\begin{theorem}
\label{thm:radialcycle}
The following recursions hold for $a$ and $b$ belonging to the respective intervals:
\vspace{5pt}

\begin{minipage}[t]{0.35\textwidth}
\begin{center}
$
\boxed{r_{a\cdot 3^n} = 2^n + r_{b\cdot 3^{n-2}}} \quad (n\geq 4)
$

\vspace{5pt}
\begin{tabular}{cc}
a & b \\ \hline
$\left[4, 4\frac{4}{9}\right) $ & $\left[4\frac{4}{9}, 4\frac{2}{3}\right)$ \\
$\left[4\frac{4}{9}, 4\frac{2}{3}\right) $ & $\left[5\frac{1}{3}, 6\right)$ \\
$\left[4\frac{2}{3}, 5\frac{1}{3}\right) $ & $\left[6, 6\frac{4}{9}\right)$ \\
$\left[5\frac{1}{3}, 6\right) $ & $ \left[8\frac{2}{3}, 9\frac{1}{3}\right) $
\end{tabular}
\end{center}
\end{minipage}
\begin{minipage}[t]{0.6\textwidth}
\begin{center}
$
\boxed{r_{a\cdot 3^n} = 2^n + r_{b\cdot 3^{n-1}}} \quad (n\geq 3)
$

\vspace{5pt}
\begin{tabular}{cc}
a & b \\ \hline
$\left[6, 6\frac{4}{9}\right) $ & $\left[4\frac{4}{9}, 4\frac{2}{3}\right)$ \\
$\left[6\frac{4}{9}, 7\frac{1}{3}\right) $ & $\left[4\frac{2}{3}, 5\frac{1}{3}\right)$ \\
$\left[7\frac{1}{3}, 8\right) $ & $\left[5\frac{1}{3}, 6\right)$ \\
$\left[8, 8\frac{2}{3}\right) $ & $ \left[7\frac{1}{3}, 8\right) $
\end{tabular}
\hspace{10pt}
\vspace{5pt}
\begin{tabular}{cc}
a & b \\ \hline
$\left[8\frac{2}{3}, 9\frac{1}{3}\right)$ & $\left[8, 8\frac{2}{3}\right)$ \\
$\left[9\frac{1}{3}, 10\right)$ & $\left[8\frac{2}{3}, 9\frac{1}{3}\right)$ \\
$\left[10, 12\right)$ & $\left[10, 12\right)$\\
& \\
\end{tabular}
\end{center}
\end{minipage}

In particular, for $n\geq 4$, the restriction of the function $r: [0,\infty)\to [0,\infty)$ to $[4\cdot 3^n, 4\cdot 3^{n+1})$ is piecewise constant on each of the following intervals:
\begin{align*}
&\left[4\cdot 3^n,~ 4\frac{4}{9}\cdot 3^n\right),~ 
\left[4\frac{4}{9}\cdot 3^n, ~4\frac{2}{3}\cdot 3^n\right),~
\left[4\frac{2}{3}\cdot 3^n, ~5\frac{1}{3}\cdot 3^n\right),~
\left[5\frac{1}{3}\cdot 3^n, ~6\cdot 3^n\right),\\
&\left[6\cdot 3^n,~ 6\frac{4}{9} \cdot 3^n\right),~
\left[6\frac{4}{9}\cdot 3^n,~ 7\frac{1}{3}\cdot 3^n\right),~
\left[7\frac{1}{3}\cdot 3^n, ~8\cdot 3^n\right),\\
&\left[8\cdot 3^n, ~8\frac{2}{3}\cdot 3^n\right),~
\left[8\frac{2}{3}\cdot 3^n, ~9\frac{1}{3}\cdot 3^n\right),~
\left[9\frac{1}{3}\cdot 3^n, ~10\cdot 3^n\right),\\
&\left[10\cdot 3^n, ~4\cdot 3^{n+1}\right).
\end{align*}
\end{theorem}

\begin{figure}
\centering
\includegraphics[width=0.5\textwidth]{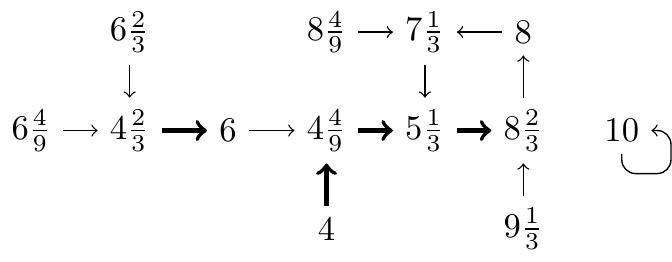}
\caption{Mnemonic for Theorem \ref{thm:radialcycle}. $a\to b$ means $r_{a\cdot 3^n} = 2^n + r_{b\cdot 3^{n-1}}$, while $a \boldsymbol\to b$ (with a thick arrow) means $r_{a\cdot 3^n} = 2^n + r_{b\cdot 3^{n-2}}$.}
\label{fig:radialcycle}
\end{figure}

\begin{table}
\begin{center}
\begin{scriptsize}
\input{sandpiletable.tex}
\end{scriptsize}
\begin{small}
\vskip 13pt
\begin{tabular}{lm{1.05in}l}
Legend: $\quad (m \mathbbm{1}_o)^\circ = $
&
\includegraphics[height=0.15\textwidth]{tablefig}
&
$\subseteq G_{n+1}$; $~\#\{\text{chips in } \eta\} = m-2m'$.
\end{tabular}
\end{small}
\vskip 13pt
\caption{Periodic structure of the single-source sandpile cluster on $SG$.
See Legend for explanation of $m$, $m'$, and $n$.
We only list values of $m$ at which $m'$ changes.  
$\Delta r_m:= r_m -r_{m-1}$ is the jump in the cluster radius.
Special values of $m$ are: ${\bf 4\cdot 3^n}$, ${\color{red} 6\cdot 3^n}$, ${\color{teal} 8\cdot 3^n}$, ${\color{blue} 10\cdot 3^n}$.
}
\label{table:spectrum}
\end{center}
\end{table}


The rest of the paper is organized as follows. 
Theorem \ref{thm:RRA} is proved in \S\ref{sec:RRA}.
Theorems \ref{thm:ASM}, \ref{thm:groupSG}, \ref{thm:tail}, and \ref{thm:radialcycle} are proved in \S\ref{sec:ASM}.
In \S\ref{sec:IDLA}, we provide strong numerical evidence suggesting sublogarithmic fluctuations in the radius of the IDLA cluster, and investigate a potential central limit theorem (CLT). Possible extensions to nested fractals, as well as related open questions, are discussed in \S\ref{sec:open}.


\section{Rotor-router aggregation on $SG$: Proof of Theorem \ref{thm:RRA}} \label{sec:RRA}


Let us first summarize the overall strategy in the proof of Theorem \ref{thm:RRA}.
\begin{itemize}[wide]
\item Solving the rotor-router problem is equivalent to solving the \emph{least action principle} for the model (see Proposition \ref{prop:LAP}).
\item A fast exact simulation algorithm for rotor-router aggregation based on this least action principle was devised by Friedrich and Levine (see Algorithm \ref{alg:FL}). The algorithm involves an initial approximation step followed by two error-correcting steps, and produces the correct rotor-router configuration and odometer function.
\item It turns out that on $SG$, we have an excellent candidate for the initial approximation using the divisible sandpile odometer, which was solved by Huss and Sava-Huss \cite{HSH17}; see \S\ref{sec:DSodo} for key facts needed for our proof. By an inductive argument described in \S\ref{sec:RRodo}, we can carry out precise error corrections in the Friedrich-Levine algorithm to arrive at the true rotor-router odometer.
\end{itemize}


\subsection{Abelian stack model and the Friedrich-Levine algorithm}

In this subsection we describe a more general aggregation model called the \textbf{abelian stack model}, which includes rotor-router aggregation and IDLA as special cases. The description follows \cite{FL13} closely.

Fix a directed graph $G=(V,E)$ which is locally finite and strongly connected, that is, given any vertices $x,y \in V$ there are directed paths from $x$ to $y$ and from $y$ to $x$.
A directed edge is denoted $e=( {\tt s}(e), {\tt t}(e))$, where ${\tt s}(e)$ and ${\tt t}(e)$ are the source and target vertices, respectively.
A rotor configuration $\rho: V\to E$ is an assignment of an edge $\rho(x)\in E$ to the vertex $x\in V$, with ${\tt s}(\rho(x))=x$.

In the abelian stack model, we equip each vertex $x\in V$ with a \emph{stack} of rotors $\{\rho_k(x)\}_{k=0}^\infty$, where each $\rho_k(x)$ is an edge with source vertex $x$.
A finite number of indistinguishable chips are assigned to $V$ according to some initial configuration.
For each $x\in V$, the first chip to visit $x$ is absorbed there and does not move again.
Each subsequent chip arriving at $x$ first shifts the stack at $x$ to become $(\rho_{k+1}(x))_k$.
After shifting the stack, the chip moves from $x$ to the target of the new rotor on top, $y={\tt t}(\rho_1(x))$.
This procedure is called a \textbf{firing} of the vertex $x$.
It can be readily seen that the $k$th chip fired from $x$ travels along the edge $\rho_k(x)$.

Given a directed edge $e$ and a nonnegative integer $n$, define
\[
R_\rho(e,n) = \#\{1\leq k \leq n : \rho_k({\tt s}(e))=e\}
\]
to be the number of times $e$ occurs among the first $n$ rotors in the stack at the vertex ${\tt s}(e)$. 
In the case of rotor-router aggregation with simple periodic rotor mechanism, one can show that
\begin{align}
\label{eq:Rrotor}
R_\rho(e,n) = \left\lfloor \frac{n+d_{\rm out}({\tt s}(e)) - j(e)}{d_{\rm out}({\tt s}(e))} \right\rfloor,
\end{align}
where $j(e)$ is the minimum positive integer such that $\rho_j({\tt s}(e))=e$, and $d_{\rm out}(v)$ is the out-degree of vertex $v$.

\begin{definition}
\label{def:stackLap}
The \textbf{stack Laplacian} of a function $u:V\to\mathbb{N}$ is the function $\Delta_\rho u: V\to\mathbb{Z}$ given by
\[
\Delta_\rho u(x) = \sum_{{\tt t}(e)=x} R_\rho(e,u({\tt s}(e))) - u(x),
\]
where the sum runs over all edges $e$ with target vertex $x$.
\end{definition}

The role of the stack Laplacian is as follows. Starting from a chip configuration $\sigma_0$, we perform $u(x)$ firings at each vertex $x\in V$. It is direct to verify that the resulting configuration is $\sigma_0+\Delta_\rho u$. 
We denote by
\[
{\tt Top}_\rho(u)(x) = \rho_{u(x)}(x) \qquad (x\in V)
\]
the rotor configuration on the tops of the stacks after the firings.

The main question in the abelian stack model is to identify the odometer function $u_*$ which produces the final stable configuration $\sigma_*$ from $\sigma_0$.
The solution to this question is given by the following \emph{least action principle}.

\begin{proposition}[\cite{FL13}*{Theorem 1}]
\label{prop:LAP}
Let $G$ be a directed graph, $\rho$ a collection of rotor stacks on $G$, and $\sigma_0$ a chip configuration on $G$.
Fix $u_*: V\to\mathbb{N}$, and let $A_* = {\rm supp}(u_*)$.
Let $\sigma_* = \sigma_0+\Delta_\rho u_*$.
Suppose that
\begin{itemize}
\item $\sigma_* \leq 1$.
\item $A_*$ is finite.
\item $\sigma_*(x)=1$ for all $x\in A_*$.
\item ${\tt Top}_\rho(u_*)$ is acyclic on $A_*$.
\end{itemize}
Then there exists a finite complete legal firing sequence for $\sigma_0$, and its odometer function is $u_*$.
\end{proposition}

Based on Proposition \ref{prop:LAP}, Friedrich and Levine established a three-step algorithm which produces $u_*$ exactly \cite{FL13}, as described in Algorithm \ref{alg:FL} below.
Step 1 of the algorithm takes in any function $u_1: V\to\mathbb{N}_0$ and returns the resulting approximate configuration $\sigma_1$.
Following it are two error-correcting steps, annihilation and reverse cycle-popping, which correct the errors in the particle configuration and the rotor configuration, respectively.
Readers interested in the computational costs associated with this algorithm may consult \cite{FL13}*{\S4--\S6}. 

\begin{algorithm}[h]
\caption{Computing the abelian stack odometer \cite{FL13}}
\label{alg:FL}
	\SetKwInOut{Input}{Inputs}
	\SetKwInOut{Output}{Outputs}
	
	\Input{Initial chip configuration $\sigma_0$ and approximate odometer $u_1$}
	\Output{Final chip configuration $\sigma_2$ and exact odometer $u_3$}
	
	\underline{(Step 1) Approximation:} \Return $\sigma_1=\sigma_0+\Delta_\rho u_1$\;
	\underline{(Step 2) Annihilation:} Set $u_2=u_1$ and $\sigma_2=\sigma_1$\;
	\ForEach{$x\in V$}{
		\uIf{$\sigma_2(x)>1$}{
	 	call $x$ a \emph{hill}, fire it by moving one chip from $x$ to ${\tt t}({\tt Top}(u_2)(x))$ and incrementing $u_2(x)$ by $1$\;
	 	}
	 	\uElseIf{$\sigma_2(x)<0$, or $\sigma_2(x)=0$ and $u_2(x)>0$}{
	 	call $x$ a \emph{hole}, unfire it by moving one chip from ${\tt t}({\tt Top}(u_2)(x))$ to $x$ and decrementing $u_2(x)$ by $1$\;
		 }
	 }
\Return $\sigma_2$ and $u_2$\;
	 \underline{(Step 3) Reverse cycle-popping:} Set $u_3=u_2$ and $A_3={\rm supp}(u_3):=\{x\in V: u_3(x)>0\}$\;
	 \If{${\tt Top}(u_3)$ is not acyclic on $A_3$}{
	 pick a cycle and unfire each of its vertices once. (This does not change $\sigma_2$.) Update $u_3$ and $A_3$;
	 }
\Return $u_3$
\end{algorithm}

\subsection{Divisible sandpile odometer} \label{sec:DSodo}

While the Friedrich-Levine algorithm was devised for fast simulation purposes, it turns out to work well on $SG$ from the analytic point-of-view, when we choose the divisible sandpile odometer as the input function $u_1$.

Recall the description of the divisible sandpile problem in \S\ref{sec:DSShapeThm}.
Consider an initial configuration $\sigma_0: V(G)\to[0,\infty)$. 
We would like to stabilize it to a final configuration $\sigma_\infty$ where $\sigma_\infty \leq 1$ everywhere, and identify the corresponding odometer function $u_\infty: V(G)\to [0,\infty)$, where $u_\infty(x)$ represents the amount of mass emitted from $x$ during the stabilization.
The solution to this problem is given in variational form by \eqref{eq:DSvar}.
On a general state space it is difficult to solve \eqref{eq:DSvar}.
Instead we can use an alternative formulation which is analogous to that for the abelian stack model, Proposition \ref{prop:LAP} above.

\begin{proposition}
\label{prop:DSLAP}
Let $u_* : V(G)\to [0,\infty)$ be a function, $A_* = \{z\in V(G): u_*(z)>0\}$ and $\sigma_* = \sigma_0+\Delta u_*$.
Suppose that
\begin{itemize}
\item $\sigma_*(z)=1$ for all $z\in A_*$.
\item $A_*$ is finite.
\item $\sigma_*\leq 1$.
\end{itemize}
Then the divisible sandpile odometer $u_\infty=u_*$.
\end{proposition}
\begin{proof}
See \cite{HSH17}*{Lemma 3.10} for the proof, which was stated for $SG$, but works on any infinite, locally finite, connected graph supporting an irreducible random walk process with infinitesimal generator $\Delta$.
\end{proof}

Let us specialize to $SG$. 
Denote the closed ball and the sphere of radius $n$ centered at $o$ by $B_n = \{y\in SG: d(o,y)\leq n\}$ and $S_n=\{y\in SG: d(o,y)=n\}$. 
Given a subset $A\subset SG$, we define its inner boundary by $\partial_I A :=\{x\in A: \exists y\in A^c,~x\sim y\}$.
For each $n\geq 1$, set
\begin{align}
\bi{n} := |B_n|-\frac{1}{2} |\partial_I B_n| = |B_{n-1}| + |\partial_I B_{n-1}|.
\end{align}
For the proof of the latter equality see \cite{HSH17}*{Lemma 4.1}.
The point is that $\bi{n}$ counts the number of vertices in $B_n$ which also takes into account boundary corrections.

In \cite{HSH17} Huss and Sava-Huss used an inductive procedure and Proposition \ref{prop:DSLAP} to give an explicit characterization of the divisible sandpile odometer function starting from $\bi{n}\mathbbm{1}_o$.
We summarize their main result as follows:

\begin{lemma}
\label{lem:divisibleodometer}
The following hold for the divisible sandpile odometer $\uds{n}$ associated with the initial distribution $\bi{n} \mathbbm{1}_o$:
\begin{enumerate}
\item \label{DS1} ${\rm supp}(\uds{n}) = B_{n-1}$.
\item \label{DS2} $\displaystyle \Delta \uds{n}(z) = 
\left\{\begin{array}{ll}
1-\bi{n} \mathbbm{1}_o, &\text{if}~z\in B_n\setminus \partial_I B_n,\\
1/2, & \text{if}~z\in \partial_I B_n,\\
0, & \text{if}~z\notin B_n.
\end{array}
\right.
$
\item \label{DS3} $\uds{n}(y)=2$ for each $y\in \partial_I B_{n-1}$.
\end{enumerate}
\end{lemma}
\begin{proof}
Parts \eqref{DS1} and \eqref{DS2} are established in \cite{HSH17}*{Theorem 4.2}. 
For Part \eqref{DS3}, observe that $y\in \partial_I B_{n-1}$ is connected to two vertices $z_1, z_2 \in S_n$. Take $z_1$, which has 4 adjacent vertices $y$, $z_2$, $w_1$, $w_2$, where $w_1, w_2 \in S_{n+1}$. 
By Part \eqref{DS1}, $\uds{n}(z_1)=\uds{n}(z_2)=\uds{n}(w_1)=\uds{n}(w_2)=0$. 
By Part \eqref{DS2},
\[
\Delta \uds{n}(z_1)= \frac{1}{4}\left(\uds{n}(y) +\uds{n}(z_2) + \uds{n}(w_1) + \uds{n}(w_2)\right)- \uds{n}(z_1)= \frac{1}{2}.
\]
Infer that $\uds{n}(y)=2$.
\end{proof}

\begin{remark}
Taking together Items \eqref{DS1} and \eqref{DS2}, it is not difficult to check that $\uds{n}$ solves the Dirichlet boundary-value problem
\begin{align}
\label{eq:Dir}
\left\{
\begin{array}{ll}
\Delta \uds{n}=1-\bi{n}\mathbbm{1}_o & \text{on } B_{n-1},\\
\uds{n}=0 & \text{on } (B_{n-1})^c.
\end{array}
\right.
\end{align}
Equation \eqref{eq:Dir}, and more generally, Poisson's equation on ball subsets of $SG$, have already been solved by Strichartz \cite{StrLap}.
Technically speaking, the divisible sandpile problem is a \emph{free} boundary-value problem.
But often its solution can be found by first guessing the support of the odometer, and then solving the \emph{Dirichlet} boundary-value problem on the support set.
\end{remark}

\subsection{Rotor-router cluster and odometer} \label{sec:RRodo}

In the subsection we show that the divisible sandpile odometer $\uds{n}$ makes an excellent approximation of the rotor-router odometer under the same initial configuration $\bi{n}\mathbbm{1}_o$, in the sense that we can perform precise error corrections in the Friedrich-Levine algorithm.
Our error-correction proof uses induction on $n$ and consists of two acts: ``filling the bulk'' and ``pulling the marionette.''
Upon making the error corrections, we identify the support of the rotor-router cluster and odometer function, thereby proving Theorem \ref{thm:RRA}.

\begin{proposition}
\label{prop:DSapprox}
For any periodic simple rotor mechanism $\rho$,
\[
\Delta_\rho \uds{n}(x) \in 
\left\{\begin{array}{ll} 
\{0\}, &\text{if}~x\notin B_n, \\
\{0,1\} ,&\text{if}~x\in \partial_I B_n,\\
\{0,1,2\}, &\text{if}~x\in S_n \setminus \partial_I B_n.
\end{array}\right.
\]
\end{proposition}
\begin{proof}
If $x\notin B_n$, then there is no vertex $y\in B_{n-1}$ which is connected to $x$,
so Part \eqref{DS1} of Lemma \ref{lem:divisibleodometer} implies that $\Delta_\rho \uds{n}(x)=0$.
If $x\in \partial_I B_n$, then it is connected to 4 vertices $y$, $z$, $w_1$, and $w_2$, where $y\in \partial_I B_{n-1}$, $z\in S_n$, and $w_1, w_2\in S_{n+1}$.
By Part \eqref{DS1} of Lemma \ref{lem:divisibleodometer}, $\uds{n}(x)=\uds{n}(z)=\uds{n}(w_1)=\uds{n}(w_2)=0$, while by Part \eqref{DS3} we have $\uds{n}(y)=2$. Therefore
\[
\Delta_\rho \uds{n}(x) = R_\rho\left((y,x), \uds{n}(y)\right) = R_\rho\left((y,x),2\right) =\left\{\begin{array}{ll} 1, & \text{if } (y,x) \in \{\rho_1(y), \rho_2(y)\}, \\ 0, & \text{if } (y,x) \in \{\rho_3(y), \rho_4(y)\}.\end{array}\right.
\]
Finally, if $x\in S_n\setminus \partial_I B_n$, then it is connected to 4 vertices $y_1$, $y_2$, $x_1$, and $x_2$, where $y_1, y_2\in \partial_I B_{n-1}$ and $x_1, x_2 \in S_n$. Again by Lemma \ref{lem:divisibleodometer}, $\uds{n}(y_1)=\uds{n}(y_2)=2$ and $\uds{n}(x)=\uds{n}(x_1)=\uds{n}(x_2)=0$, so
\[
\Delta_\rho \uds{n}(x) = \sum_{i=1}^2 R_\rho\left((y_i,x), 2\right)  =
\left\{\begin{array}{ll} 
2, & \text{if } (y_1,x) \in \{\rho_1(y_1), \rho_2(y_1)\} \text{ and } (y_2, x) \in \{\rho_1(y_2), \rho_2(y_2)\},\\ 
1, & \text{if } (y_1,x) \in \{\rho_1(y_1), \rho_2(y_1)\} \text{ and } (y_2, x) \in \{\rho_3(y_2), \rho_4(y_2)\}, \\ 
1, & \text{if } (y_1,x) \in \{\rho_3(y_1), \rho_4(y_1)\} \text{ and } (y_2, x) \in \{\rho_1(y_2), \rho_2(y_2)\}, \\ 
0, & \text{if } (y_1,x) \in \{\rho_3(y_1), \rho_4(y_1)\} \text{ and } (y_2, x) \in \{\rho_3(y_2), \rho_4(y_2)\}.
\end{array}\right.
\]
\end{proof}

Let $\urr{n}$ and $\srr{n}$ denote, respectively, the rotor-router odometer and the final chip configuration associated with the initial configuration $\bi{n}\mathbbm{1}_o$.

\begin{proposition}
\label{prop:RRsupport}
$B_{n-2} \subset {\rm supp}(\urr{n}) \subset B_{n-1}$
and
$B_{n-1} \subset {\rm supp}(\srr{n}) \subset  B_n$.
\end{proposition}
\begin{proof}
We prove this by induction on $n$.
When $n=1$, $\bi{1}=2$, and the claim clearly holds.
Now assume the claim holds for $n$.
Then we have ${\rm supp}(\srr{n+1}) \supset B_{n-1}$, \emph{i.e.,} $B_{n-1}$ is fully occupied.
To complete the induction, we need to settle the remaining $\bi{n+1}-|B_{n-1}|$ chips, and show that they fill $S_n$ and do not overrun $B_{n+1}$.

We apply Algorithm \ref{alg:FL}, using the divisible sandpile odometer $\uds{n+1}$ as the approximate odometer in Step 1.
Let $u'_{n+1}$ and $\sigma'_{n+1}= \bi{n+1}\mathbbm{1}_o + \Delta_\rho \uds{n+1}$ be, respectively, the odometer and the chip configuration which serve as input to Step 2 of Algorithm \ref{alg:FL}.
By Lemma \ref{lem:divisibleodometer}, Part \eqref{DS1}, $u'_{n+1}(x)=0$ for all $x\notin B_n$, so no vertex in $(B_n)^c$ is a hole.
Moreover, by Proposition \ref{prop:DSapprox}, if $x\in \partial_I B_{n+1}$ (resp.\@ $x\in (B_{n+1})^c$), then $\sigma'_{n+1}(x) \in \{0,1\}$ (resp.\@ $\sigma'_{n+1}(x)=0$), so every vertex in $\partial_I B_{n+1} \cup B_{n+1}^c$ is neither a hill nor a hole.

According to the above rationale, we carry out Step 2 in two acts:

\emph{\underline{Act 1: Filling the bulk.}} Fire and unfire vertices in $B_{n-1}$ so as to place 1 chip at each vertex in $B_{n-1}$, in accordance with the induction hypothesis. 
This leaves
\[
\bi{n+1}-|B_{n-1}| = |B_n|+|\partial_I B_n| - |B_{n-1}| = |S_n| + |\partial_I  B_n|
\] 
chips in $B_{n+1} \setminus B_{n-1}$.
In particular, since $u'_{n+1}(x)>0$ for each $x\in S_n$, we will fire and unfire as many vertices in $S_n$ as needed until 1 chip is placed at each vertex in $S_n$. This is carried out in the next act.

\emph{\underline{Act 2: Pulling the marionette.}}
According to Proposition \ref{prop:DSapprox}, there exists (many) a rotor configuration $\rho_{\max}$ which places the \emph{maximal} number of chips on $S_{n+1}$ after Act 1, namely:
\[
\sigma'_{n+1}(x) =
 \left\{\begin{array}{ll}
 1,& \text{if } x\in \partial_I B_{n+1},\\
 2,& \text{if } x\in S_{n+1}\setminus \partial_I B_{n+1}.
 \end{array}
 \right.
\]
The ensuing analysis differs depending on whether $n$ is odd or even. See Figure \ref{fig:RRAproof}.

\begin{figure}
\centering
\begin{subfigure}[t]{0.5\textwidth}
	\centering
	\includegraphics[height=0.12\textheight]{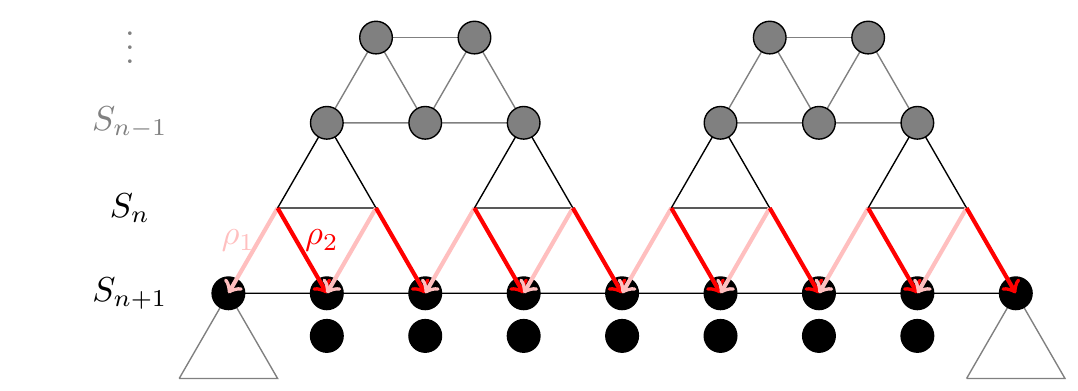}
	\caption{$n$ odd}
\end{subfigure}%
~
\begin{subfigure}[t]{0.5\textwidth}
	\centering
	\includegraphics[height=0.12\textheight]{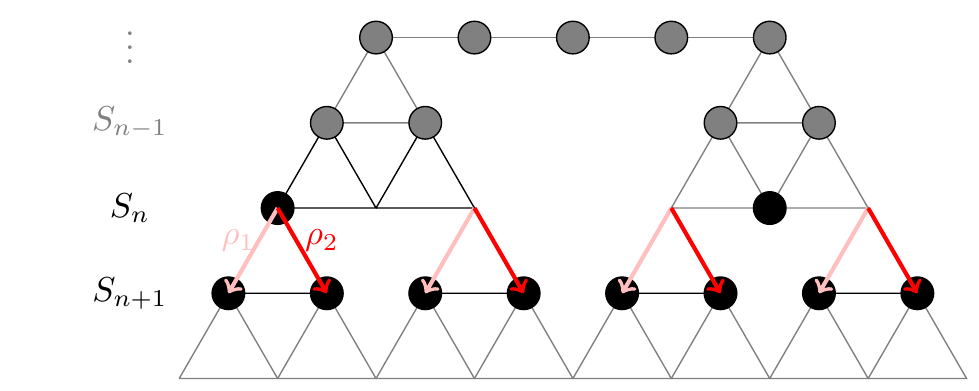}
	\caption{$n$ even}
\end{subfigure}
\caption{The setup for the ``marionette'' act in the proof of Proposition \ref{prop:RRsupport}.}
\label{fig:RRAproof}
\end{figure}

When \underline{$n$ is odd}, there are $2 |S_{n+1}| - |\partial_I B_{n+1}|$ chips on $S_{n+1}$. Moreover, since $u'_{n+1}(x)=2$, every vertex $x\in S_n$ carries a rotor which is targeted towards some vertex in $S_{n+1}$.
Since $|S_n|=|\partial_I B_n|$, deduce that $S_n$ carries $\left(|S_n|+|\partial_I B_n|\right)-\left(2|S_{n+1}| - |\partial_I B_{n+1}|\right)=2|S_n| - 2|S_{n+1}| + |\partial_I B_{n+1}|=0$ chips.
Therefore we unfire every vertex $x\in S_n$ to pull one chip from its successor vertex ${\tt t}(\rho_2(x))$ in $S_{n+1}$, and 
then fire the remaining hills on $S_{n+1}$.

When \underline{$n$ is even}, observe that $S_{n+1}=\partial_I B_{n+1}$ and $|\partial_I B_n|=\frac{1}{2} |S_{n+1}|$. 
Given that every vertex on $S_{n+1}$ carries 1 chip, deduce that $S_n$ carries $\left(|S_n|+|\partial_I B_n|\right) - |S_{n+1}|=|S_n|-|\partial_I B_n|$ chips, that is, there are at least $|\partial_I B_n|$ holes on $S_n$.
Therefore we fire and unfire as many vertices in $B_n$ as necessary until every vertex in $S_n$ carries $1$ chip. In the process $|\partial_I B_n|$ chips will be pulled from $S_{n+1}$, leaving $|\partial_I B_n|$ chips on $S_{n+1}$.

In either case, we arrive at a configuration $\sigma''_{n+1}$ with 
\begin{align}
\label{eq:sn1}
B_n \subset {\rm supp}(\sigma''_{n+1}) \subset B_{n+1}.
\end{align}

Using another rotor configuration $\rho$ places no more chips on $S_{n+1}$, and thus requires no more pulls from $S_{n+1}$, than $\rho_{\rm max}$.
Consequently, the resulting configuration $\sigma''_{n+1}$ will satisfy \eqref{eq:sn1}.
If $u''_{n+1}$ denotes the corresponding odometer function, we deduce that $B_{n-1}\subset {\rm supp}(u''_{n+1}) \subset B_n$.
(The inner bound ${\rm supp}(\sigma''_{n+1})\supset B_n$ implies that all vertices in $B_{n-1}$ have fired.)
This completes Step 2.

Finally, Step 3 (reverse cycle-popping) involves unfirings only and does not alter the chip configuration. So the final rotor-router cluster $\srr{n+1}$ is identical to $\sigma''_{n+1}$ and satisfies \eqref{eq:sn1}, while the support of the odometer cannot increase. In other words, $B_n \subset {\rm supp}(\srr{n+1})\subset B_{n+1}$ and $B_{n-1} \subset {\rm supp}(\urr{n+1})\subset B_n$.
\end{proof}

\begin{proof}[Proof of Theorem \ref{thm:RRA}]
For each $m\in \mathbb{N}$, let $n_m =\max\{k\geq 0: \bi{k} \leq m\}$.
Then $\bi{n_m}\leq m < \bi{n_m+1}$, which implies that
${\rm supp}(\urr{n_m}) \subset \mathcal{R}(m) \subset {\rm supp}(\urr{n_m+1})$
and
${\rm supp}(\srr{n_m}) \subset \sigma(m) \subset {\rm supp}(\srr{n_m+1})$.
Now apply Proposition \ref{prop:RRsupport} to deduce the theorem.
\end{proof}


\section{Abelian sandpile growth on $SG$: Proofs of Theorems \ref{thm:ASM}, \ref{thm:groupSG}, \ref{thm:tail}, and \ref{thm:radialcycle}} \label{sec:ASM}

The proofs of the various theorems proceed as follows: Theorem \ref{thm:ASM}, Part \eqref{item:ball} $\to$ Theorem \ref{thm:groupSG} $\to$ Theorem \ref{thm:tail}
$\to$ Theorem \ref{thm:radialcycle} $\to$ Theorem \ref{thm:ASM}, Part \eqref{item:renewal}.

Recall that $G_n$ is the level-$n$ Sierpinski gasket pre-fractal graph, which has three corner vertices $o$, $x$, $y$.
Set $\partial G_n=\{x,y\}$.
The graph $G_n^{(s)}$ (resp.\@ $G_n^{(o)}$) is the sinked version of $G_n$ with $\partial G_n$ (resp.\@ $o$) identified as sink.
The sandpile group of $G_n^{(s)}$ (resp.\@ $G_n^{(o)}$) is denoted $\mathcal{R}_n^{(s)}$ (resp.\@ $\mathcal{R}_n^{(o)}$).

\subsection{Sandpile cluster is an exact ball}

In this subsection we prove Theorem \ref{thm:ASM}, Part (\ref{item:ball}).
Given a sandpile configuration $\eta$ on $G$ and a subgraph $G'\subset G$, we denote the restriction of $\eta$ to $V(G')$ by $\eta|_{G'}$.

\begin{lemma}
\label{lem:topplecut}
Let $\eta$ be an unstable configuration on $\bigcup_{n\geq 1} G_n$ with the property that $\left.\eta\right|_{G_n^{(s)}} \in \mathcal{R}_n^{(s)}$ for some $n\in \mathbb{N}$.
Suppose we stabilize $\eta$ at all vertices in $G_n^{(s)}$ and obtain a configuration $\eta^o$ (which may be unstable outside $G_n^{(s)}$), and in doing so, each of the two vertices in $\partial G_n$ is toppled $k$ times.
Then every vertex in $G_n^{(s)}$ is also toppled $k$ times, and $\left.\eta^\circ\right|_{G_n^{(s)}}= \left.\eta\right|_{G_n^{(s)}}$.
\end{lemma}
\begin{proof}
On $\eta$ we topple at each of the two vertices in $\partial G_n$ once, and produce the configuration $\eta' = \eta + \sum_{x\in \partial G_n} \Delta'(x,\cdot)$.
In the process every vertex $y$ that is connected to $\partial G_n$ receives an extra chip.
So $\left.\eta'\right|_{G_n^{(s)}} = \left.(\eta+\sum_{y\sim s} \mathbbm{1}_y)\right|_{G_n^{(s)}} = \left.\eta\right|_{G_n^{(s)}} + \sum_{\substack{y\in G_n^{(s)}\\y\sim s}} \mathbbm{1}_y$.
By Lemma \ref{lem:mit}, we can stabilize $\left.\eta'\right|_{G_n^{(s)}}$ by toppling at every vertex in $G_n^{(s)}$ exactly once, and return the original configuration $\left.\eta\right|_{G_n^{(s)}}$.
This process can be repeated as many times as needed.
\end{proof}

Let $A(m)$ and $S(m)$ denote, respectively, the firing set and the receiving set of $(m\mathbbm{1}_o)^\circ$.
The following result is fundamental to the solution of the sandpile growth problem on $SG$.

\begin{proposition}
\label{prop:ball}
For each $m\geq 12$, there exists a unique $(n,m')\in \mathbb{N}^2$ with $m'<m/2$ such that
\begin{center}
\begin{tabular}{lm{1.05in}l}
$(m \mathbbm{1}_o)^\circ = $
&
\includegraphics[height=0.15\textwidth]{tablefig}
&
$\subseteq G_{n+1}$.
\end{tabular}
\end{center}
Moreover:
\begin{enumerate}
\item \label{item:AS} For each $m\in \mathbb{N}$, $B_o(r_m-1) \subset A(m) \subset S(m) = B_o(r_m)$.
\item \label{item:radial} For each $m\geq 12$, $r_m = 2^n + r_{m'-2}$.
\end{enumerate}
\end{proposition}

\begin{proof}
It is direct to check that for each $m<12$, $S(m) \subset G_1$, $S(m)=B_o(r_m)$ for some $r_m \in \{0,1,2\}$, and $A(m)=B_o(r_m-1)$, with the understanding that $B_o(-1)=\emptyset$.

If $m\geq 12$, we obtain $(m\mathbbm{1}_o)^\circ$ according to the following algorithm.
First topple and stabilize at every vertex of $G_n^{(s)}$, but pause any excess chips on $\partial G_n$.
This produces a recurrent configuration $\eta\in \mathcal{R}_n^{(s)}$ in $G_n^{(s)}$.
By the axial symmetry of $G_n$, each of the two cut points in $\partial G_n$ carries the same number of chips $m'$.
If $m'< 4$ (degree of the cut point), we are done.
Otherwise, we topple on $\partial G_n$, but with each toppling we also topple once at every vertex in $G_n^{(s)}$.
By Lemma \ref{lem:topplecut}, this leaves $\left.\eta\right|_{G_n^{(s)}}$ invariant while additional chips are fired into $G_{n+1} \setminus G_n$.
Continue stabilizing at every vertex in $G_{n+1}^{(s)}$ and, if necessary, pause any excess chips on $\partial G_{n+1}$.
If the resulting configuration is stable, we are done.
Otherwise, continue the above process into $G_{n+2}$.
This algorithm proves the claimed diagram.
The condition $m'<m/2$ follows from Propositions \ref{prop:maxtopple} and \ref{prop:merge} below.

Let us make two observations.
First, since $m$ is finite, the algorithm terminates.
Second, with each simultaneous toppling on $G_n$, each $x\in \partial G_n$ loses $4$ chips to its neighboring vertices, and receives $2$ chips back from $y\sim x$, $y\in G_n^{(s)}$, thereby losing a net number of $2$ chips while keeping $\left.\eta\right|_{G_n^{(s)}}$ intact. 
As a result, upon further toppling and stabilizing on $G_{n+1}\setminus G_n$ (plus any additional simultaneous toppling on $G_n$), the number of chips at $x$ decrements in steps of $2$, until $2$ (resp.\@ $3$) chips remain if $m'$ is even (resp.\@ odd). 
This process generates a copy of $((m'-2)\mathbbm{1}_o)^\circ$ (resp.\@ $((m'-3)\mathbbm{1}_o)^\circ$) in each connected component of $G_{n+1}\setminus G_n$.

To prove Part (\ref{item:AS}), we claim that for every $m\in \mathbb{N}$ such that $S(m) \subset G_n$, there exists $r_m \in \mathbb{N}_0$ such that 
$B_o(r_m-1)\subset A(m) \subset S(m) = B_o(r_m)$. 
When $n=1$, this claim holds by virtue of the first paragraph of the proof.
Now suppose the claim holds at level $n$.
Then for every $m\in \mathbb{N}$ such that $G_n \subsetneq S(m) \subset G_{n+1}$, it follows from the previous paragraph and the induction hypothesis that $S(m)=B_o(r_m)$ for some $r_m \in (2^n, 2^{n+1}]$, and that every vertex in $B_o(r_m-1)$ has fired, \emph{i.e.,} $A(m) \supset B_o(r_m-1)$.

To prove Part (\ref{item:radial}), we deduce from the diagram and Part (\ref{item:AS}) that $r_m = 2^n + r_{m'-2}$ (resp.\@ $r_m=2^n +r_{m'-3}$) if $m'$ is even (resp.\@ odd).
Recall (from \S\ref{sec:sandpileintro}) that when $m$ is even, $(m\mathbbm{1}_o)^\circ$ and $((m+1)\mathbbm{1}_o)^\circ$ differs only at the origin $o$. This implies that $r_{m'-3}=r_{m'-2}$ when $m'$ is odd.
\end{proof}

\begin{remark}
Both containments in $B_o(r_m-1) \subset A(m) \subset B_o(r_m)$ are strict in general.
For example, both $A(12)$ and $A(14)$ are equal to $G_1 \setminus \partial G_1$ (neither cut point on $\partial G_1$ topples), which strictly contains $B_o(1)$ and is strictly contained in $B_o(2)$.
\end{remark}

\subsection{Toppling identities \& the identity element of the sandpile group} \label{sec:toppling}

In this subsection we prove Theorem \ref{thm:groupSG}.
Let's begin with the toppling identities.
For concreteness, in the proof below we choose $e_n$ (\emph{cf.\@} Definition \ref{def:en}) to be the ``background'' recurrent configuration, and study the  toppling patterns over $e_n$.

\begin{lemma}
\label{lem:stab}
For each $n\in\mathbb{N}$, the following stabilizations hold. (The number at each sink vertex indicates the number of chips absorbed there.) 
\begin{align}
\label{2top}
\begin{array}{lm{1in}lm{1in}}
\text{With two sink vertices at the bottom:} \quad &
\includegraphics[width=0.15\textwidth]{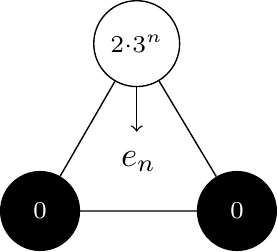} &
\longrightarrow &
\includegraphics[width=0.15\textwidth]{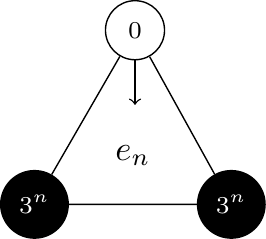} 
\end{array} \\
\label{1top}
\begin{array}{lm{1in}lm{1in}}
\text{With one sink vertex at the bottom-left:} \quad &
\includegraphics[width=0.15\textwidth]{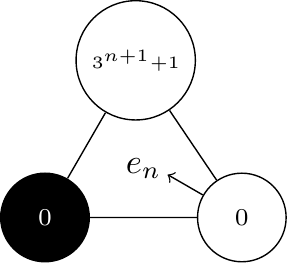} &
\longrightarrow &
\includegraphics[width=0.15\textwidth]{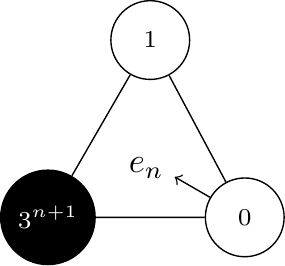}
\end{array}
\end{align}
\end{lemma}

\begin{proof}
We prove both results by induction on $n$.
The $n=1$ case is a direct calculation left for the reader. 
Now suppose both \eqref{2top} and \eqref{1top} hold at level $n$. 
To verify the claim at level $n+1$, recall how $e_{n+1}$ is constructed per Definition \ref{def:en}. We then apply the induction hypothesis to each level-$n$ cell as many times as needed, using Lemma \ref{lem:topplecut}. After each toppling operation over a single (or a pair of adjacent) level-$n$ cells, pause the chips at the cut (or corner) vertices. Then based on the number of the chips available, we carry out further operations until every vertex, except the sinks, carries a nonnegative number of chips fewer than its degree.
For \eqref{2top} the induction step is straightforward:
\begin{center}
\begin{tabular}{m{1.55in}lm{1.55in}lm{1.55in}}
\includegraphics[width=0.25\textwidth]{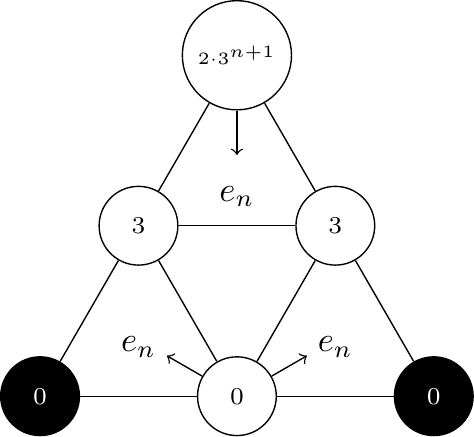}
&
$\longrightarrow$
&
\includegraphics[width=0.25\textwidth]{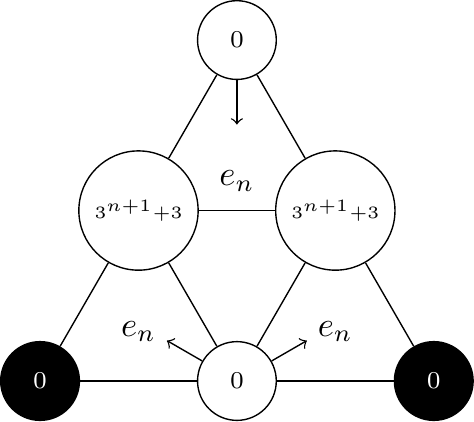}
&
$\longrightarrow$
&
\includegraphics[width=0.25\textwidth]{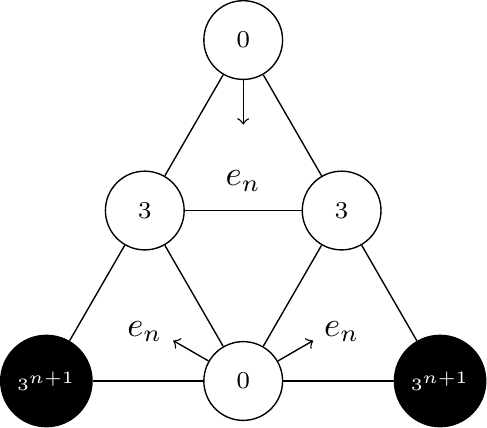}
\end{tabular}
\end{center}
For \eqref{1top} the induction step is described in Figure \ref{fig:sideind}.
\begin{figure}
\begin{center}
\begin{tabular}{lm{1.6in}lm{1.6in}lm{1.6in}}
&
\includegraphics[width=0.25\textwidth]{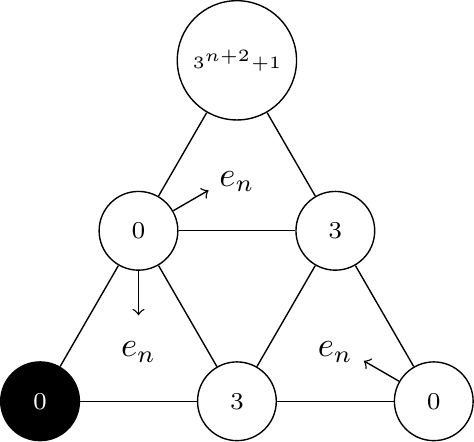}
&
$\longrightarrow$
&
\includegraphics[width=0.25\textwidth]{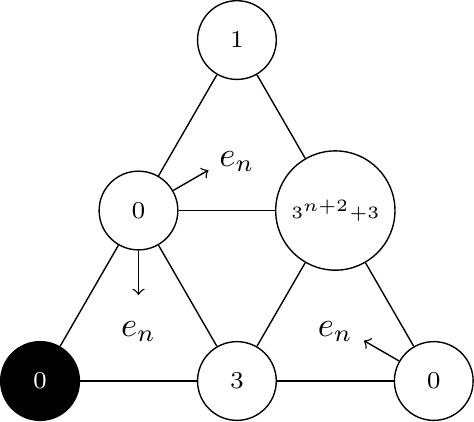}
&
$\longrightarrow$
&
\includegraphics[width=0.25\textwidth]{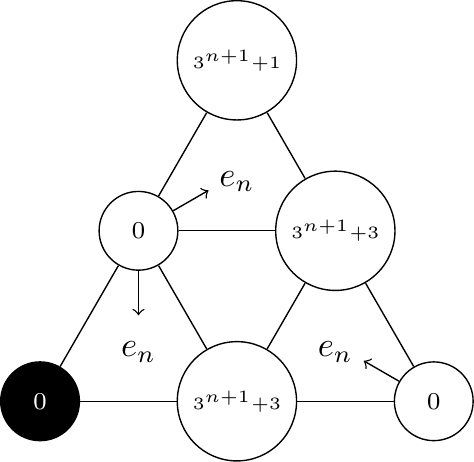}
\\
$\overset{*}{\longrightarrow}$
&
\includegraphics[width=0.25\textwidth]{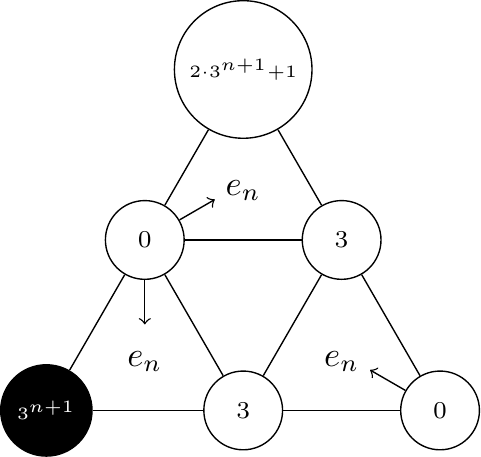}
&
$\longrightarrow$
&
\includegraphics[width=0.25\textwidth]{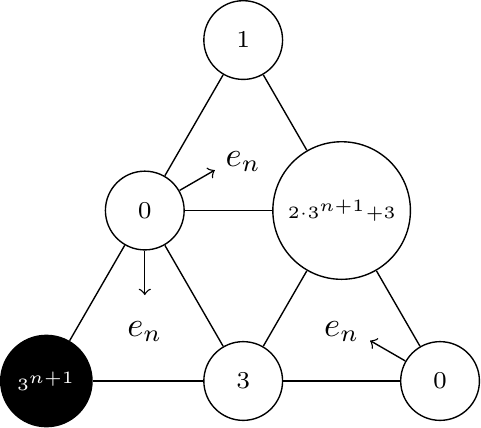}
&
$\longrightarrow$
&
\includegraphics[width=0.25\textwidth]{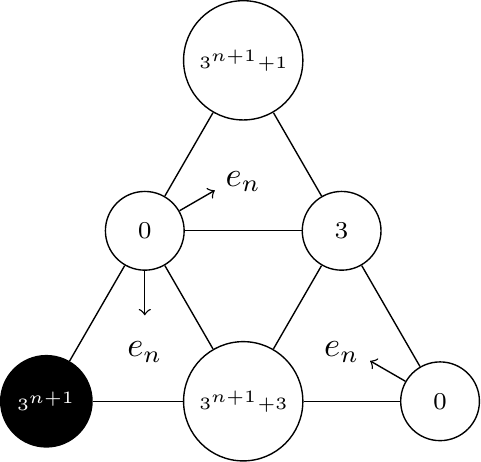}
\\
$\longrightarrow$
&
\includegraphics[width=0.25\textwidth]{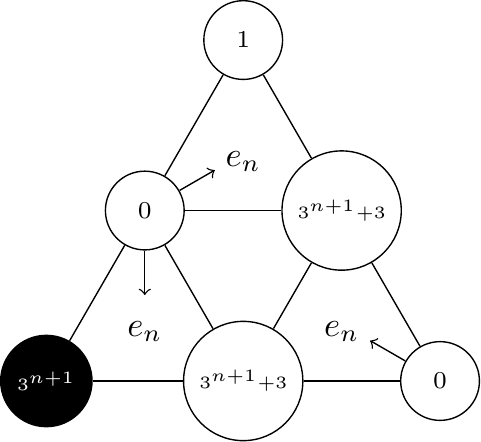}
&
$\overset{*}{\longrightarrow}$
&
\includegraphics[width=0.25\textwidth]{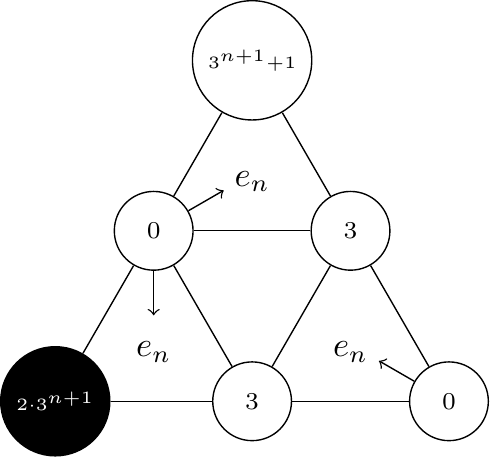}
&
$\longrightarrow$
&
\includegraphics[width=0.25\textwidth]{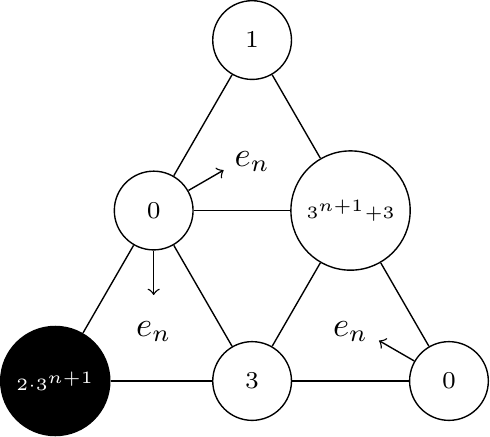}
\\
$\longrightarrow$
&
\includegraphics[width=0.25\textwidth]{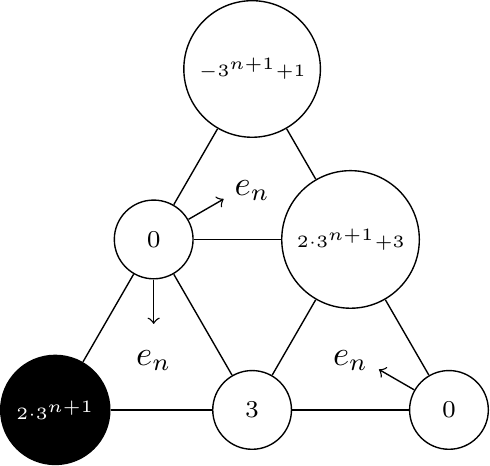}
&
$\longrightarrow$
&
\includegraphics[width=0.25\textwidth]{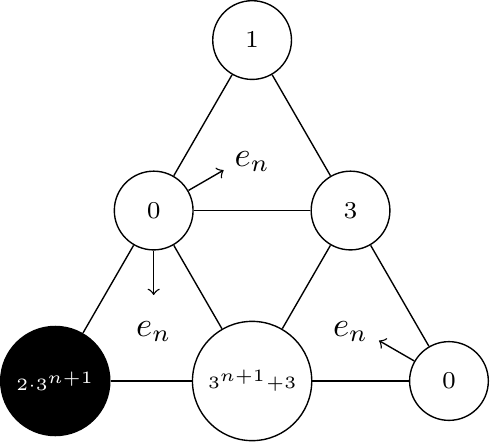}
&
$\longrightarrow$
&
\includegraphics[width=0.25\textwidth]{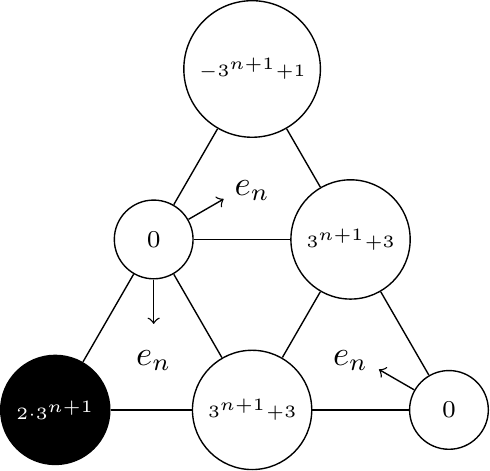}
\\
$\overset{*}{\longrightarrow}$
&
\includegraphics[width=0.25\textwidth]{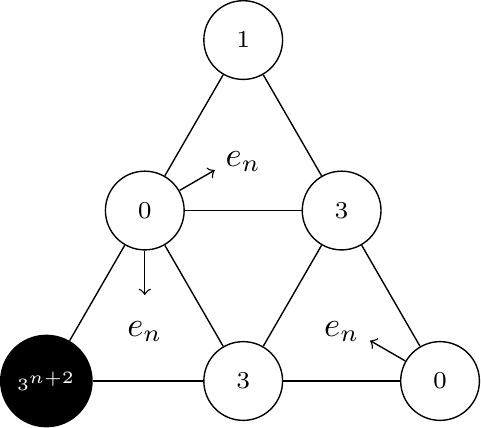}
\end{tabular}
\end{center}
\caption{The induction step in the stabilization \eqref{1top} of Lemma \ref{lem:stab}. Arrows with an asterisk $*$ indicate an application of Lemma \ref{lem:topplecut}.}
\label{fig:sideind}
\end{figure}
\end{proof}

\begin{remark}
\label{rem:en}
For the configuration $e_n$, we note that
\begin{tabular}{m{0.8in}}
\includegraphics[width=0.135\textwidth]{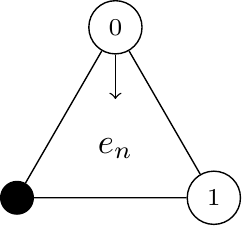}
\end{tabular}
is recurrent, while
\begin{tabular}{m{0.8in}}
\includegraphics[width=0.135\textwidth]{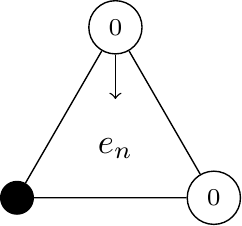}
\end{tabular}
is not recurrent, as can be checked using the burning test and by induction on $n$. 
\end{remark}

Next we establish the identity elements.

\begin{proposition}
\label{prop:id2s}
The identity element $e_n$ of $(\mathcal{R}_n^{(s)},\oplus)$ is 
\begin{tabular}{m{0.8in}l}
\includegraphics[width=0.135\textwidth]{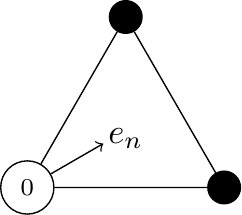}&.
\end{tabular}
\end{proposition}
\begin{proof}
We prove by induction on $n$ that $(2e_n)^\circ = e_n$, and upon stabilization each sink vertex in $\partial G_n$ absorbs $\frac{3}{2}(3^n-1)$ chips.
(This cardinality makes sense since $3(3^n-1)$ is the total number of chips in $e_n$.)

The base case $n=1$ is a straightforward computation. 
For the induction step, assume that $e_n$ is the identity element of $\mathcal{R}_n^{(s)}$, and that in the stabilization of $2 e_n$, each sink vertex receives $\frac{3}{2}(3^n-1)$ chips. Let $e_{n+1}$ be constructed according to Definition \ref{def:en}. Then we stabilize $2e_{n+1}$ using the induction hypothesis, followed by an application of \eqref{1top} in Lemma \ref{lem:stab}, as depicted in Figure \ref{fig:idind}.
\begin{figure}
\begin{center}
\begin{tabular}{m{1.7in}lm{1.7in}lm{1.7in}}
\includegraphics[height=0.25\textwidth]{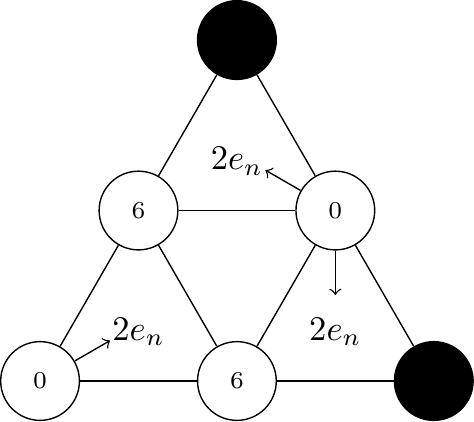}
&
$\longrightarrow$
&
\includegraphics[height=0.25\textwidth]{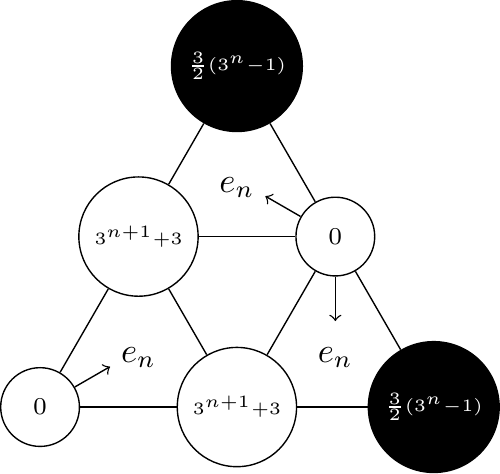}
&
$\longrightarrow$
&
\includegraphics[height=0.25\textwidth]{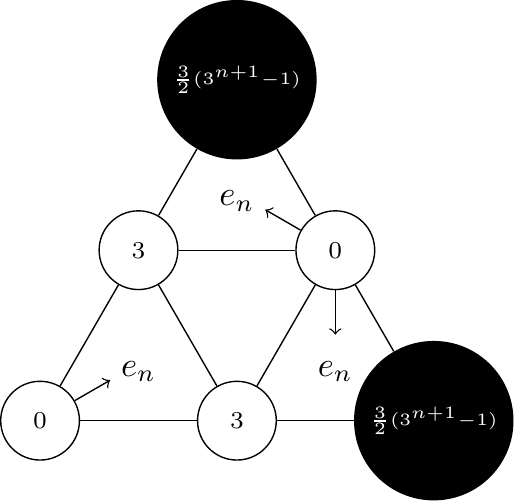}
\end{tabular}
\end{center}
\caption{The induction step in the proof of $(2e_n)^\circ = e_n$.}
\label{fig:idind}
\end{figure}
This proves that $e_{n+1}$ is the identity element of $\mathcal{R}_{n+1}^{(s)}$.
\end{proof}

Using Proposition \ref{prop:id2s} and the toppling identity \eqref{2top} we can derive a number of useful consequences.

\begin{proposition}
\label{prop:recper}
For every $\eta\in \mathcal{R}_n^{(s)}$,
\begin{align}
\label{eq:rec2top}
\begin{array}{m{1in}lm{.95in}l}
\includegraphics[height=0.15\textwidth]{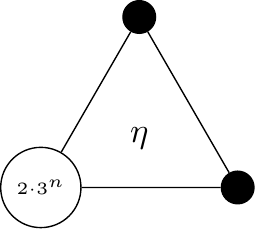} &
\longrightarrow &
\includegraphics[height=0.15\textwidth]{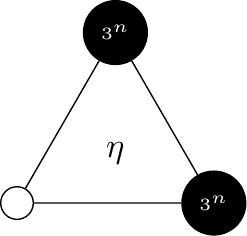} &
.
\end{array}
\end{align}
\end{proposition}
\begin{proof}
The stabilization \eqref{2top} in Lemma \ref{lem:stab} says that $(2\cdot 3^n) \mathbbm{1}_o\oplus e_n= e_n$, with each sink vertex receiving $3^n$ chips upon stabilization. 
By the abelian property, $(2\cdot 3^n)\mathbbm{1}_o\oplus \eta = \left((2\cdot 3^n)\mathbbm{1}_o \oplus e_n\right) \oplus \eta = e_n\oplus \eta=\eta$, and each sink vertex still receives $3^n$ chips upon stabilization.
\end{proof}

We can also reverse the process \eqref{eq:rec2top}, which leads to the following stabilization:
\begin{proposition}
\label{prop:rev1top}
For every $\eta\in \mathcal{R}_n^{(o)}$,
\begin{align}
\label{eq:rev1top}
\begin{array}{m{1in}lm{.95in}l}
\includegraphics[height=0.15\textwidth]{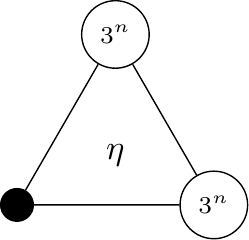} &
\longrightarrow &
\includegraphics[height=0.15\textwidth]{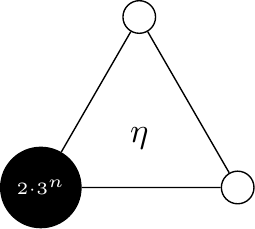} &
.
\end{array}
\end{align}
\end{proposition}

We now have all the tools to prove
\begin{proposition}
\label{prop:id1s}
The identity element $e_n^{(o)}$ of $(\mathcal{R}_n^{(o)}, \oplus)$ is
\begin{tabular}{m{0.8in}l}
\includegraphics[width=0.135\textwidth]{Mno}&.
\end{tabular}
\end{proposition}
\begin{proof}
We prove by induction on $n$ that $(2e_n^{(o)})^\circ = e_n^{(o)}$, and upon stabilization the sink $o$ receives $4\cdot 3^n-2$ chips. 
When $n=1$ the calculation is straightforward. 
Suppose the result holds on level $n$.
The induction step for level $n+1$ proceeds as follows. (From this point onwards, ``IH'' denotes an application of the induction hypothesis.)
\begin{align*}
\begin{array}{rm{1.6in}lm{1.6in}l}
&
\includegraphics[width=0.25\textwidth]{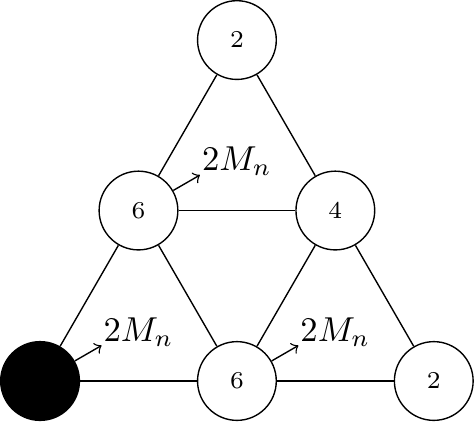}
&
\overset{\text{IH}}{\longrightarrow}
&
\includegraphics[width=0.25\textwidth]{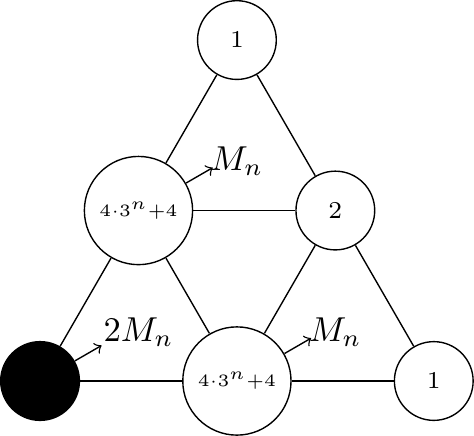}
&
\\
\overset{\eqref{eq:rev1top}}{\longrightarrow}
&
\includegraphics[width=0.25\textwidth]{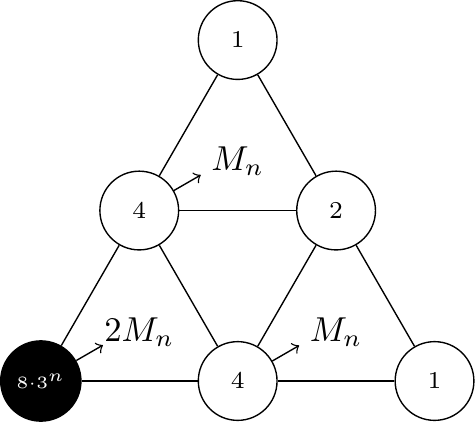}
&
\overset{\text{IH}}{\longrightarrow}
&
\includegraphics[width=0.25\textwidth]{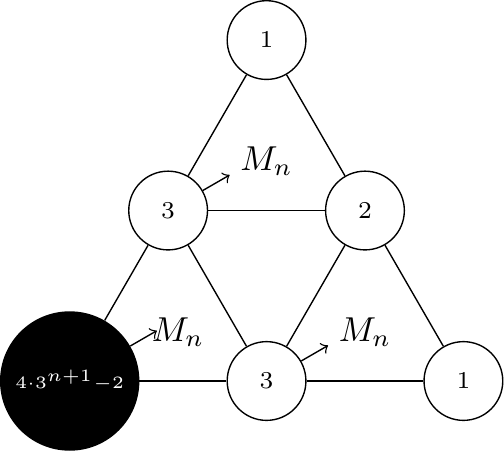}
&.
\end{array}
\end{align*}
\end{proof}

\begin{proof}[Proof of Theorem \ref{thm:groupSG}]
We already proved Items \eqref{item:identity}, \eqref{item:tops} and the first identity
\begin{align}
\label{eq:btop}
\eta\oplus 3^n(\mathbbm{1}_x+\mathbbm{1}_y) = \eta
\end{align}
in Item \eqref{item:topo}.
It remains to establish that for every $\eta\in\mathcal{R}_n^{(o)}$, 
\begin{align}
\label{eq:xtop} \eta\oplus 3^{n+1}\mathbbm{1}_x &= \eta,\\
\label{eq:ytop} \eta\oplus 3^{n+1}\mathbbm{1}_y &= \eta.
\end{align}
To prove \eqref{eq:xtop} we combine \eqref{1top} with the abelian property and proceed \emph{\`a la} the proof of Proposition \ref{prop:recper}.
The identity \eqref{eq:ytop} is equivalent to \eqref{eq:xtop} once we reflect the configuration across the axis of symmetry.
\end{proof}

\begin{remark}[Nested structure of the identity elements]
We pause to make an observation which applies to other nested self-similar fractal graphs: that the sequence of identity elements $(e_n)_n$ is nested in that for every $n\in\mathbb{N}$, $\left. e_{n+1}\right|_{G_n^{(s)}} = e_n$.
Indeed, if $e_{n+1}$ is the identity element of $\mathcal{R}_{n+1}^{(s)}$, then we can stabilize $2e_{n+1}$ by first stabilizing everywhere in $G_n^{(s)}$ to produce $\left(2\left.e_{n+1}\right|_{G_n^{(s)}}\right)^\circ$ and pausing excess chips on $\partial G_n$. 
Then we fire chips into $G_{n+1}^{(s)} \setminus G_n^{(s)}$ and stabilize every vertex therein, and by Lemma \ref{lem:topplecut}, this leaves the configuration in $G_n^{(s)}$ invariant. 
At the end of the stabilization, we recover $e_{n+1}$, and thus on $G_n^{(s)}$ we have $\left(2\left.e_{n+1}\right|_{G_n^{(s)}}\right)^\circ= \left.e_{n+1}\right|_{G_n^{(s)}}$, \emph{i.e.,} $\left. e_{n+1}\right|_{G_n^{(s)}} = e_n$.
By the same argument, the nested property also holds for the sequence $(e_n^{(o)})_n$.
\end{remark}

\subsection{A reflection and a rotation lemma} \label{sec:reflection}

In this subsection we establish two stabilization lemmas on $G_n^{(o)}$, taking advantage of the axial symmetry inherent in the graph.
These lemmas play a crucial role in the proof of radial jumps in \S\ref{sec:explosion} and \S\ref{sec:enumeration}.

\begin{lemma}[Reflection across the axis of symmetry]
\label{lem:reflection}
Let $\eta \in \mathcal{R}_n^{(o)}$ be such that $\eta= e_n^{(o)} \oplus \alpha\mathbbm{1}_x  \oplus \beta\mathbbm{1}_y$ for some $\alpha, \beta \in \mathbb{N}_0$.
Let $k_x, k_y \in \mathbb{N}_o$ solve the system of equations
\begin{align}
\label{eq:kxky}
\left\{
\begin{array}{rl}
\alpha+k_x &= \beta+p_0\cdot 3^n + p_1 \cdot 3^{n+1}\\
\beta+ k_y &=\alpha+p_0 \cdot 3^n + p_2 \cdot 3^{n+1}
\end{array}
\right.
\end{align}
for some $p_0, p_1, p_2\in\mathbb{Z}$.
Then
\begin{align}
\label{eq:refstab}
\begin{array}{m{1.05in}lm{1.05in}l}
\includegraphics[width=0.17\textwidth]{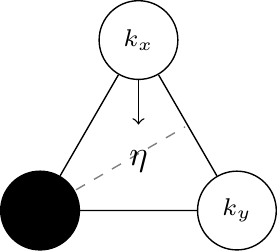} &
\equiv &
\includegraphics[width=0.17\textwidth]{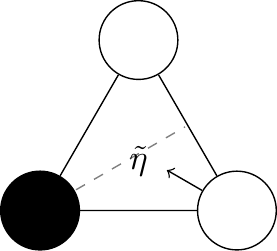} &
,
\end{array}
\end{align}
where $\tilde\eta$ is the reflection of $\eta$ across the axis of symmetry.
\end{lemma}

\begin{proof}
By the axial symmetry, the reflection of $\eta$ satisfies $\tilde\eta = e_n^{(o)} \oplus \beta \mathbbm{1}_x \oplus \alpha\mathbbm{1}_y$.
We then observe that \eqref{eq:refstab} is implied by the algebraic identity
\[
e_n^{(o)} \oplus (\alpha+k_x)\mathbbm{1}_x \oplus (\beta+k_y)\mathbbm{1}_y = e_n^{(o)} \oplus \beta \mathbbm{1}_x \oplus \alpha \mathbbm{1}_y.
\]
This explains \eqref{eq:kxky} in the special case $p_0=p_1=p_2=0$. 
For the general case of \eqref{eq:kxky}, we apply the equivalence under the toppling identities \eqref{eq:btop}, \eqref{eq:xtop}, and \eqref{eq:ytop}.
\end{proof}

\begin{lemma}[$120^\circ$-rotation of $M_n$]
\label{lem:120turn}
The following stabilizations hold:
\begin{align}
\label{eq:Mnturn}
\begin{array}{m{0.9in}lm{1in}l}
\includegraphics[width=0.155\textwidth]{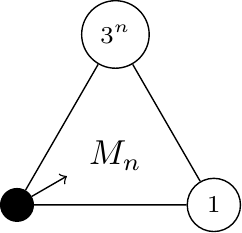} &
\longrightarrow &
\includegraphics[width=0.17\textwidth]{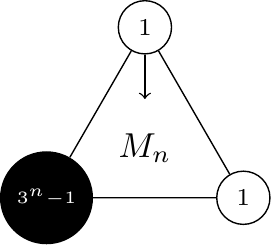} & ,
\end{array}
\\
\label{eq:Mnturn2}
\begin{array}{m{0.9in}lm{1in}l}
\includegraphics[width=0.155\textwidth]{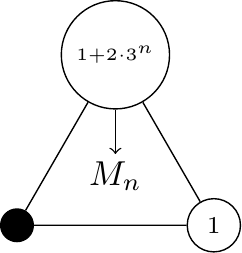} &
\longrightarrow &
\includegraphics[width=0.17\textwidth]{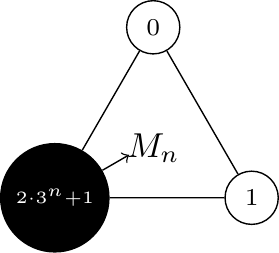} &.
\end{array}
\end{align}
\end{lemma}

\begin{proof}
We prove \eqref{eq:Mnturn} by induction on $n$.
The $n=1$ case is verified readily. 
Suppose the stabilization holds at level $n$.
At level $n+1$ the stabilization proceeds as follows:
\begin{center}
\begin{tabular}{lm{1.6in}lm{1.6in}lm{1.6in}}
&
\includegraphics[width=0.25\textwidth]{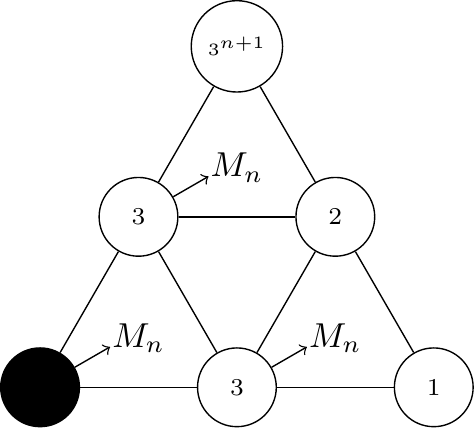}
&
$\overset{\eqref{eq:rec2top}}{\longrightarrow}$
&
\includegraphics[width=0.25\textwidth]{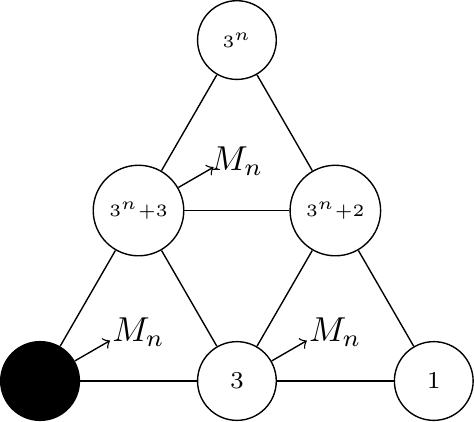}
&
$\overset{\text{IH}}{\longrightarrow}$
&
\includegraphics[width=0.25\textwidth]{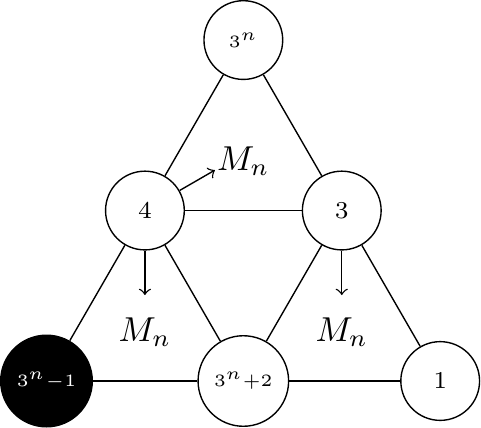}
\\
$\overset{\text{IH}}{\longrightarrow}$
&
\includegraphics[width=0.25\textwidth]{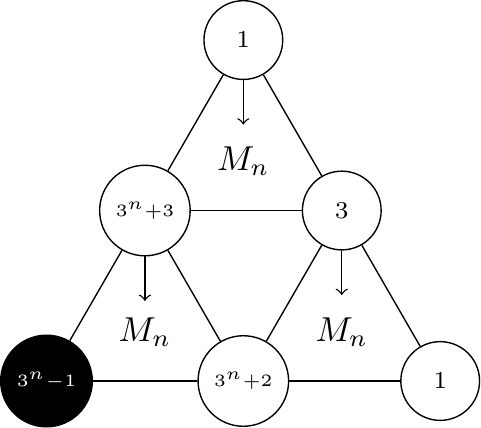}
&
$\overset{\eqref{eq:rev1top}}{\longrightarrow}$
&
\includegraphics[width=0.25\textwidth]{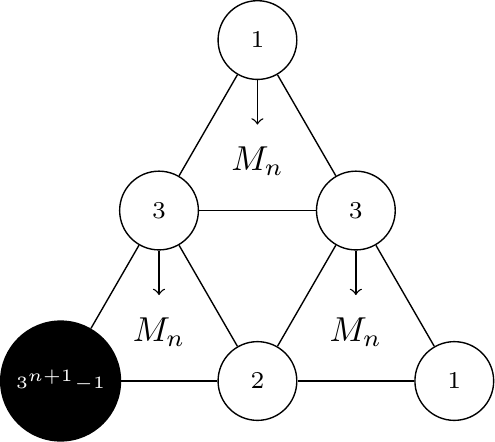}&.
\end{tabular}
\end{center}

Now we can prove \eqref{eq:Mnturn2} using \eqref{eq:Mnturn} and the toppling identity \eqref{eq:xtop}:
\begin{align*}
\begin{array}{m{.8in}lm{.8in}lm{.8in}l}
\includegraphics[width=0.15\textwidth]{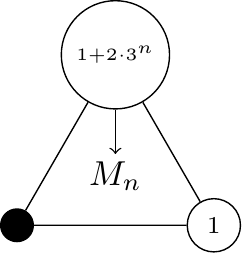} &
\overset{\eqref{eq:Mnturn}}{=} &
\includegraphics[width=0.15\textwidth]{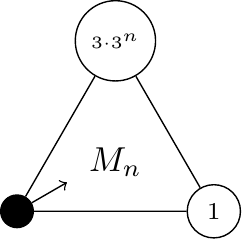} &
\overset{\eqref{eq:xtop}}{=} &
\includegraphics[width=0.15\textwidth]{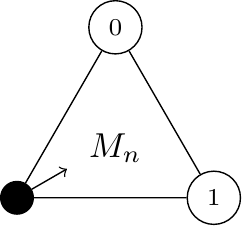} &
.
\end{array}
\end{align*}
The number of chips that the sink receives can be inferred readily.
\end{proof}

\subsection{Explosions in sandpile growth} \label{sec:explosion}

In this subsection we prove the existence of explosions in the growing cluster, \emph{i.e.,} a radial jump of size $>1$. 
To be precise, the main explosion, occurring at mass $4\cdot 3^n$, happens as the configuration transitions from $M_n$ (see Definition \ref{def:Mn}) to $e_n$ (see Definition \ref{def:en}).

\begin{proposition}
\label{prop:maxtopple}
For each $n\in \mathbb{N}$,
\begin{center}
\begin{tabular}{lm{.8in}l}
$((4\cdot 3^n-2) \mathbbm{1}_o)^\circ = $
&
\includegraphics[height=0.13\textwidth]{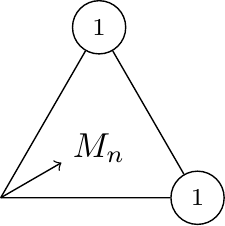}
&.
\end{tabular}
\end{center}
It follows that $r_{4\cdot 3^n-2} = 2^n$.
\end{proposition}
\begin{proof}
It is direct to verify the identity for $n=1$. 
Suppose the identity holds at level $n$.
Note that $4\cdot 3^{n+1}-2 = (4\cdot 3^n - 2) + 4(2\cdot 3^n)$.
Based on this we carry out the stabilization as follows:
\begin{center}
\begin{tabular}{rM{.9in}l M{.9in}l M{1.4in} l}
$(4\cdot 3^{n+1}-2)\mathbbm{1}_o \overset{\text{IH}}{\longrightarrow} $ &
\includegraphics[width=.9in]{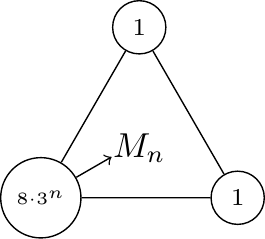} &
$\overset{\eqref{eq:rec2top}}{\longrightarrow}$ &
\includegraphics[width=.9in]{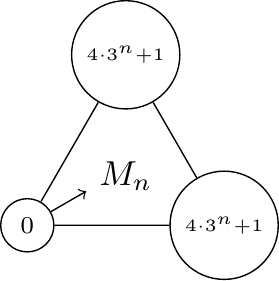} &
$\overset{\substack{\text{Lem.\@ \ref{lem:topplecut}} \\ + \text{IH}}}{\longrightarrow}$ &
\includegraphics[width=1.4in]{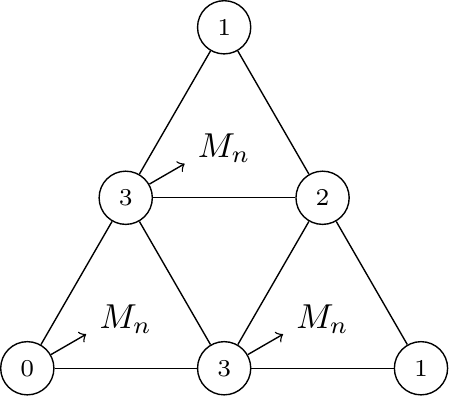} &
.
\end{tabular}
\end{center}

In the last step, we use Lemma \ref{lem:topplecut} to topple the $4\cdot 3^n +1$ chips at the cut point in $\partial G_n$. The number of chips decrements in steps of $2$ until $3$ chips remain, producing a facsimile of $((4\cdot 3^n-2)\mathbbm{1}_o)^\circ$ that emanates from the cut point. By the induction hypothesis this configuration is $M_n$. The two copies of $M_n$ merge at the midpoint between the two sinks of $G_{n+1}^{(s)}$, resulting in $1+1$ chips at that midpoint. According to Definition \ref{def:Mn}, this is the stable configuration $M_{n+1}$ with 1 chip at each sink of $G_{n+1}^{(s)}$.

Since $M_n$ has full support in $G_n$, it follows that $r_{4\cdot 3^n-2} = 2^n$. 
\end{proof}

\begin{proposition}
\label{prop:merge}
For each $n\in \mathbb{N}$,
\begin{align}
\label{eq:Mtoe}
\begin{array}{lm{0.8in}lm{1.1in}l}
(4\cdot 3^n \mathbbm{1}_o)^\circ = 
&
\includegraphics[height=0.13\textwidth]{Mnresult}
&
\oplus ~2\mathbbm{1}_o =
&
\includegraphics[height=0.16\textwidth]{enresult}
&,
\end{array}
\end{align}
where $b_n=|V(G_{n-1})|= \frac{3}{2}(3^{n-1}+1)$. It follows that $r_{4\cdot 3^n} = 2^n + r_{b_n-2}$.
\end{proposition}
\begin{proof}
When $n=1$ the identity is verified directly.
Suppose the identity holds at level $n$.
Based on the induction hypothesis we have
\[
\begin{array}{m{0.9in}lm{1.05in}l}
\includegraphics[width=0.15\textwidth]{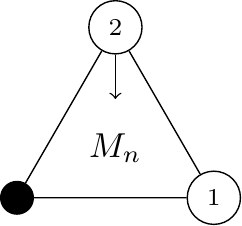} &
\longrightarrow &
\includegraphics[width=0.18\textwidth]{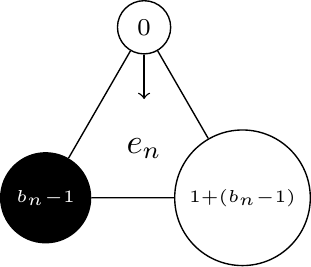} &
.
\end{array}
\]
Combine this with the $M_n$ rotation identity \eqref{eq:Mnturn}, infer that
\begin{align}
\label{eq:eneno}
\begin{array}{m{0.9in}l}
\includegraphics[width=0.15\textwidth]{enrec} &
~=~e_n^{(o)} \oplus 3^n\mathbbm{1}_x \oplus (1-b_n)\mathbbm{1}_y
~\overset{\eqref{eq:btop}}{=}~ e_n^{(o)} \oplus (2\cdot 3^n)\mathbbm{1}_x \oplus (b_n-2)\mathbbm{1}_y.
\end{array}
\end{align}
(Recall Remark \ref{rem:en}.) We will use \eqref{eq:eneno} in the final reflection step below. 

Now we indicate the stabilization steps.
The first step is based on the count $4\cdot 3^{n+1} = 4\cdot 3^n + 4(2\cdot 3^n)$ and the induction hypothesis:
\begin{center}
\begin{tabular}{rM{1.4in}l M{1.4in}l }
$(4\cdot 3^{n+1})\mathbbm{1}_o \overset{\text{IH}}{\longrightarrow} $ &
\includegraphics[width=.9in]{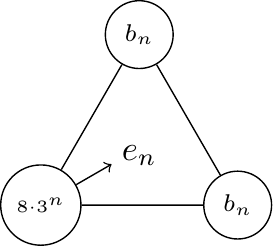} &
$\overset{\eqref{eq:rec2top}}{\longrightarrow}$ &
\includegraphics[width=.9in]{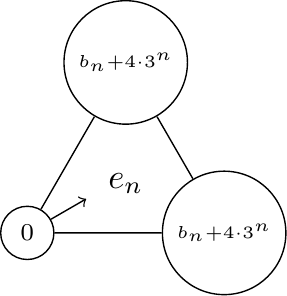} &\\
$\overset{\text{IH}}{\longrightarrow}$ &
\includegraphics[width=1.4in]{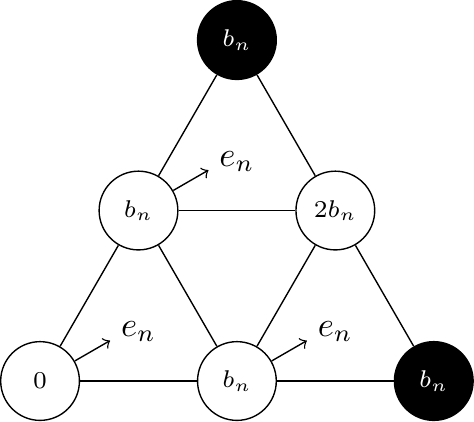} &
$\longrightarrow$ &
\includegraphics[width=1.4in]{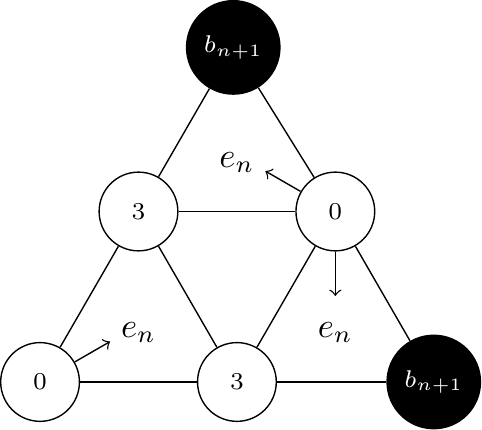} &
.
\end{tabular}
\end{center}

The final reflection step requires justification.
Observe that we need to send the remaining excess chips to the sinks $\partial G_{n+1}$.
Since $b_n$ is odd (resp.\@ even) when $n$ is odd (resp.\@ odd), we separately prove that
\begin{align}
\label{eq:enturn}
\begin{array}{lm{1.05in}lm{1.05in}l}
n \text{ even:} &
\includegraphics[width=0.17\textwidth]{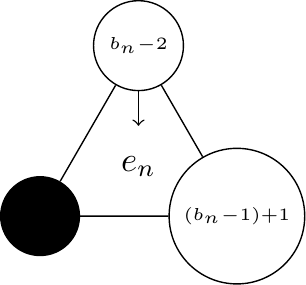} &
\longrightarrow &
\includegraphics[width=0.17\textwidth]{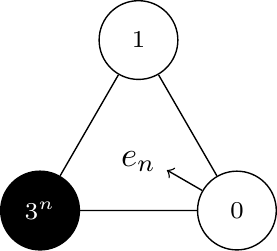} &
;
\\
n \text{ odd:} &
\includegraphics[width=0.17\textwidth]{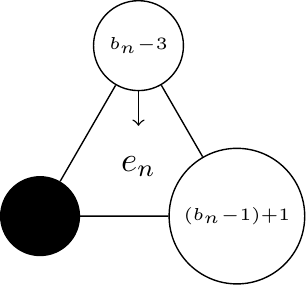} &
\longrightarrow &
\includegraphics[width=0.17\textwidth]{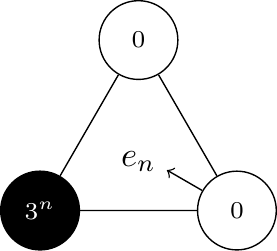} &
.
\end{array}
\end{align}

Recalling \eqref{eq:eneno}, we apply the reflection Lemma \ref{lem:reflection} with $\alpha=2\cdot 3^n$ and $\beta=b_n-2$, and verify that \eqref{eq:kxky} is satisfied with $k_x=b_n-2$, $k_y=b_n-1$, $p_0=-1$, $p_1=1$ and $p_2=0$. This proves the case when $n$ is even.
For the case when $n$ is odd, the argument is the same except that $1$ chip is removed from $x$.
The number of chips $3^n$ received by the sink can be inferred directly.

In either case, once we glue back the resulting configuration, there will be $3$ chips on each junction vertex in $\partial G_n$.
This leads to the claimed final configuration shown above, with each sink vertex in $\partial G_{n+1}$ receiving $b_n+3^n=b_{n+1}$ chips.

Finally, topple on $\partial G_{n+1}$, making sure that every topple on $\partial G_{n+1}$ also triggers a topple at every vertex in $G_{n+1}^{(s)}$ according to Lemma \ref{lem:topplecut}. 
\end{proof}

We can now establish the existence of explosion at mass $4\cdot 3^n$.

\begin{corollary}
\label{cor:maxradius}
For $n\geq 2$, $\max\{m\in \mathbb{N}: r_m = 2^n\} = 4\cdot 3^n-1$. Moreover, $r_{4\cdot 3^n} - r_{4\cdot 3^n-1} \geq 2$ for each $n\geq 3$.
\end{corollary}
\begin{proof}
The first statement is a consequence of Propositions \ref{prop:maxtopple} and \ref{prop:merge}, together with the fact that $b_n\geq 4$ for $n\geq 2$. 
For the second statement, note that Propositions \ref{prop:merge} and \ref{prop:ball} imply that $r_{4\cdot 3^n} = 2^n + r_{b_n-2}$.
When $n\geq 3$, $r_{b_n-2} \geq 2$, so $r_{4\cdot 3^n}-r_{4\cdot 3^n-1} \geq (2^n+2)-2^n =2$.
\end{proof}

\begin{corollary}[Periodicity of sandpile patterns]
\label{cor:pattper}
Let $m\geq 4\cdot 3^n$.
Then for each $k \in \mathbb{N}$ and $k\leq n$, $(m\mathbbm{1}_o)^\circ |_{G_k^{(s)}} = \left((m+ 2\cdot 3^k)\mathbbm{1}_o\right)^\circ |_{G_k^{(s)}}$.
\end{corollary}
\begin{proof}
The condition $m\geq 4\cdot 3^n$ is to ensure that $\left.(m\mathbbm{1}_o)^\circ\right|_{G_n^{(s)}}$ is a recurrent configuration (but see Remark \ref{rem:firsttimerec}). The claim then follows from Proposition \ref{prop:recper}.
\end{proof}

Having identified the main explosions, we can now identify the secondary explosions occurring at mass $6\cdot 3^n$, $8\cdot 3^n$, and $10\cdot 3^n$.

\begin{proposition}
\label{prop:periodicity}
For each $n\in \mathbb{N}$ and $p\in \{0,1,2,3\}$, 
\[
\begin{array}{lm{1in}lm{1in}l}
\left(((4+2p) 3^n-2) \mathbbm{1}_o\right)^\circ = 
&
\includegraphics[height=0.15\textwidth]{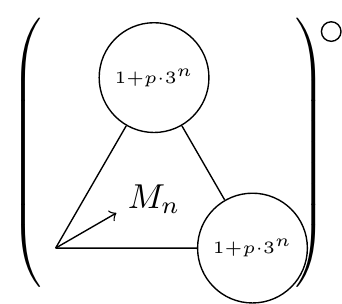}
&
;~
\left((4+2p) 3^n \mathbbm{1}_o\right)^\circ = 
&
\includegraphics[height=0.15\textwidth]{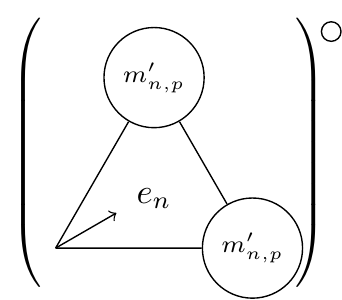}
&
,
\end{array}
\]
where $m'_{n,p} = b_n + p\cdot 3^n = \left(p+\frac{1}{2}\right)3^n + \frac{3}{2}$. It follows that 
\[
r_{(4+2p)3^n-2} = 2^n + r_{p\cdot 3^n-1}
\quad\text{and}\quad
r_{(4+2p) 3^n} = 2^n + r_{(p+\frac{1}{2}) 3^n-\frac{1}{2}}.
\]
\end{proposition}
\begin{proof}
The diagrams follow directly from Propositions 
\ref{prop:maxtopple}, \ref{prop:merge} and  \ref{prop:recper}. 
The radial identities are a consequence of Proposition \ref{prop:ball}, Item \eqref{item:radial}.
\end{proof}

We can use Proposition \ref{prop:periodicity} and induction to obtain an easy consequence.

\begin{corollary}
\label{cor:constantrad}
For every $m\in [10\cdot 3^{n-1}, 4\cdot 3^n)$:
\begin{enumerate}
\item \label{item:rm1} $r_m = 2^{n}$.
\item \label{item:m1} $(m\mathbbm{1}_o)^\circ$ has $1$ chip on each vertex in the boundary set $\partial G_n$.
\end{enumerate}
\end{corollary}
\begin{proof}
To prove Part \eqref{item:rm1}, we note by Corollary \ref{cor:maxradius} that $r_m \leq 2^{n}$. To show that $r_m = 2^{n}$, it is enough to show that $r_{10\cdot 3^{n-1}} = 2^{n}$ by induction on $n$. When $n=1$, $r_{10\cdot 3^0} = r_{10} = 2$. Suppose $r_{10\cdot 3^{n-1}} = 2^{n}$. Then using Proposition \ref{prop:periodicity} we have
\[
r_{10\cdot 3^{n}} = 2^{n} + r_{(3+\frac{1}{2})3^{n+1}-\frac{1}{2}} \geq 2^{n} + r_{10\cdot 3^{n}} = 2^{n} + 2^{n} = 2^{n+1}.
\]
Combined with $r_{10\cdot 3^{n}} \leq 2^{n+1}$ this implies $r_{10\cdot 3^{n}}=2^{n+1}$.

To prove Part \eqref{item:m1}, note that since the receiving set $\left((10\cdot 3^{n-1})\mathbbm{1}_o\right)^\circ$ is an exact ball and has full support in $G_n$, each vertex in $\partial G_n$ must receive at least $1$ chip. 
Meanwhile, Proposition \ref{prop:maxtopple} states that $\left((4\cdot 3^{n}-1)\mathbbm{1}_o\right)^\circ$  also has full support in $G_n$ and carries 1 chip on each vertex in $\partial G_n$. The claim follows from the fact that the number of chips on $\partial G_n$ increases monotonically with $m$.
\end{proof}

\subsection{Enumeration of radial jumps (I): analysis of sandpile tiles} \label{sec:enumeration}

In the next two subsections we complete the proof of Theorem \ref{thm:tail}.
We present the proofs in the following order:
\begin{itemize}
\item ``$4\frac{2}{3}^-$'' and ``$5\frac{1}{3}^-$,'' followed by ``$4\frac{2}{3}$.''
\item ``$4\frac{4}{9}$'' and ``$5\frac{1}{3}$,'' which are proved in tandem.
\item ``$e+2$,'' followed by ``$4\frac{4}{9}^-$.''
\end{itemize}
The first two items are based on tiling ideas and are proved in this subsection.
The last item is proved in the next subsection.

To establish ``$4\frac{2}{3}^-$'', we introduce another family of sandpile tiles.

\begin{lemma}
\label{lem:Mn3n}
The following stabilization holds:
\begin{align}
\label{eq:Mn3n}
\begin{array}{m{1in}lm{.95in}l}
\includegraphics[height=0.15\textwidth]{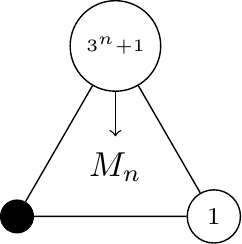} &
\longrightarrow &
\includegraphics[height=0.15\textwidth]{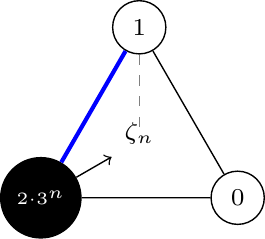} &
,
\end{array}
\end{align}
where 
\begin{tabular}{m{0.8in}lm{0.8in}l}
 \includegraphics[width=0.135\textwidth]{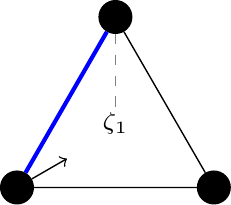}
&
$=$
&
 \includegraphics[width=0.135\textwidth]{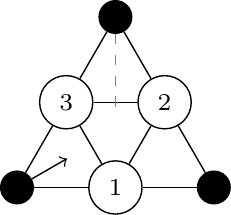}
 &,
\end{tabular}
and for $n\geq 1$, $\zeta_{n+1}$ is constructed by gluing three copies of $\zeta_n$ according to the rule

\begin{center}
\begin{tabular}{m{1.3in}lm{1.3in}l}
\includegraphics[height=0.2\textwidth]{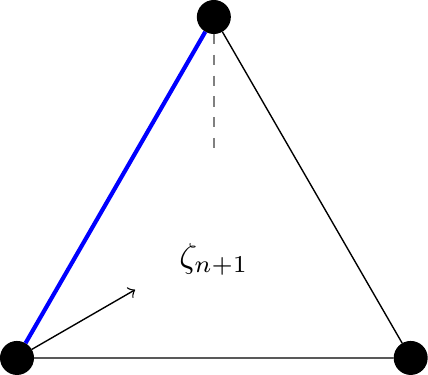}
&
$=$
&
\includegraphics[height=0.2\textwidth]{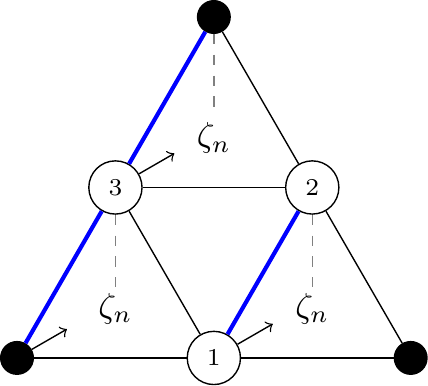}
&
.
\end{tabular}
\end{center}
(For emphasis: Every vertex along the {\color{blue}\underline{\textbf{blue line}}} carries $3$ chips.)
\end{lemma}
\begin{proof}
The $n=1$ case is verified directly.
Assume that the claim holds at level $n$.
At level $n+1$ we have the following series of stabilizations:
\begin{center}
\begin{tabular}{lm{1.6in}lm{1.6in}lm{1.6in}}
&
\includegraphics[width=0.25\textwidth]{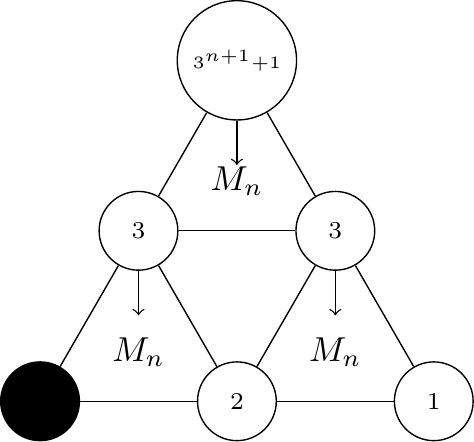}
&
$\overset{\eqref{eq:rec2top}}{\longrightarrow}$
&
\includegraphics[width=0.25\textwidth]{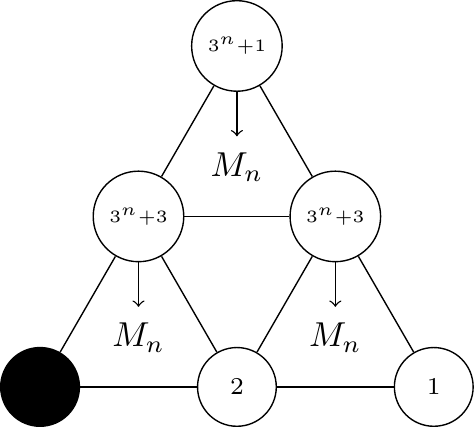}
&
$\overset{\text{IH}}{\longrightarrow}$
&
\includegraphics[width=0.25\textwidth]{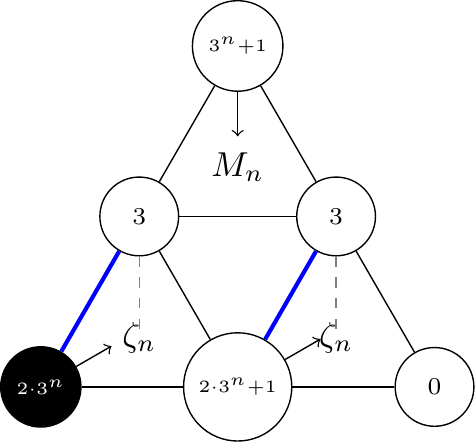}
\\
$\overset{\text{IH}}{\longrightarrow}$
&
\includegraphics[width=0.25\textwidth]{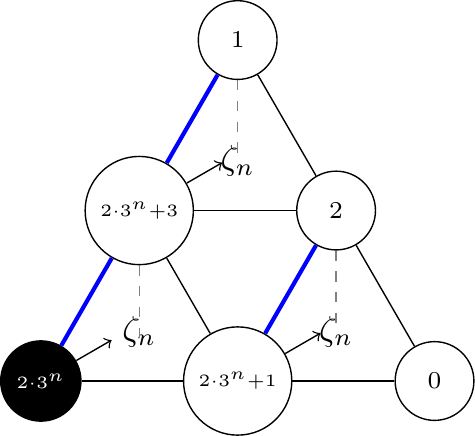}
&
$\overset{\eqref{eq:rev1top}}{\longrightarrow}$
&
\includegraphics[width=0.25\textwidth]{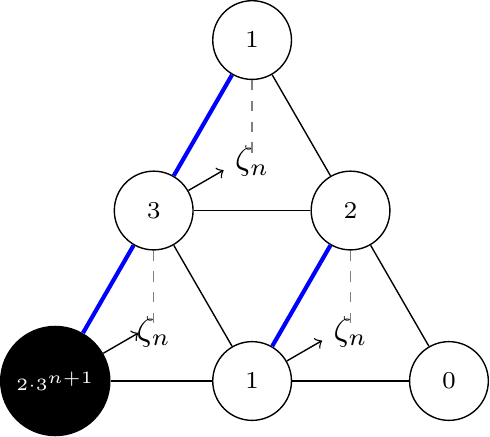}
&.
\end{tabular}
\end{center}
\end{proof}

\begin{proposition}[``$4\frac{2}{3}^-$'']
\label{prop:423-}
\[
\begin{array}{lm{1in}l}
\left((4\frac{2}{3}\cdot 3^n-2)\mathbbm{1}_o\right)^\circ = 
&
\includegraphics[height=0.15\textwidth]{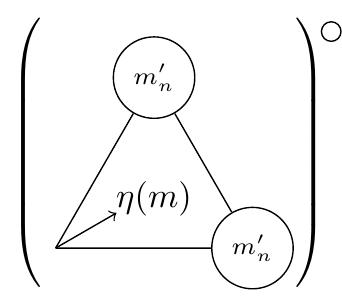}
&
, \text{ where } m_n'= 2\cdot 3^{n-1}+1.
\end{array}
\]
\end{proposition}
\begin{proof}
We start with the configuration $\left((4\cdot 3^n-2)\mathbbm{1}_o\right)^\circ$ and add to it $\frac{2}{3}\cdot 3^n = 2\cdot 3^{n-1}$ chips at $o$.
The resulting stabilization is
\begin{align*}
\begin{array}{m{1.6in}lm{1.6in}lm{1.6in}l}
\includegraphics[width=0.25\textwidth]{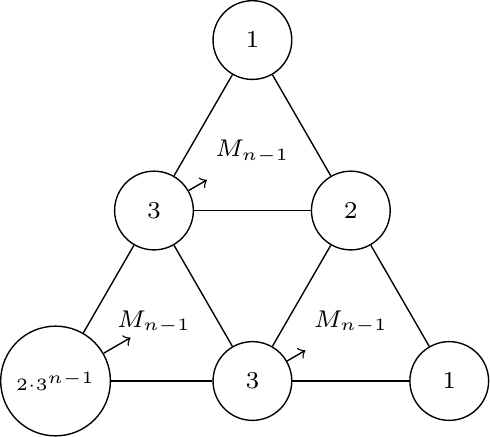}
&
\overset{\eqref{eq:rec2top}}{\longrightarrow}
&
\includegraphics[width=0.25\textwidth]{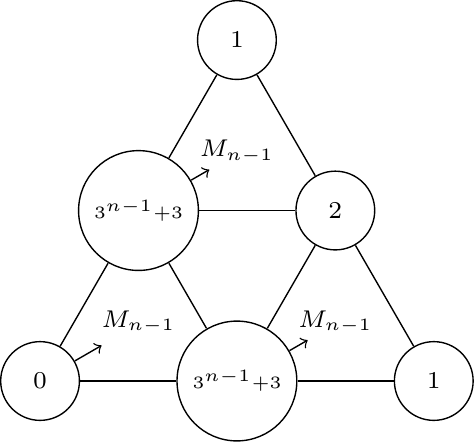}
&
\overset{\eqref{eq:Mn3n}}{\longrightarrow}
&
\includegraphics[width=0.25\textwidth]{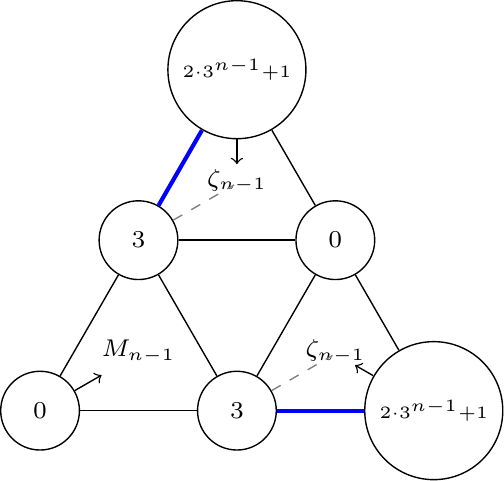}
&
.
\end{array}
\end{align*}
\end{proof}

\begin{proposition}[``$5\frac{1}{3}^-$'']
\label{prop:513-}
\[
\begin{array}{lm{1.75in}l}
\left((5\frac{1}{3}\cdot 3^n-2) \mathbbm{1}_o\right)^\circ = 
&
\includegraphics[height=0.22\textwidth]{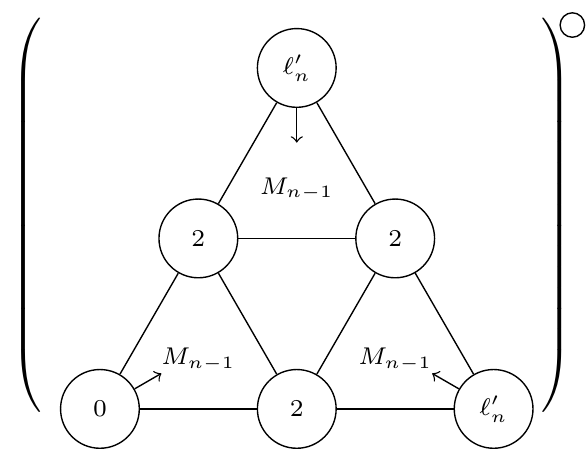}
&
, ~\text{where $\ell_n'= 2\cdot 3^{n-1}+2$.}
\end{array}
\]
\end{proposition}
\begin{proof}
We start with the configuration $\left((4\cdot 3^n-2)\mathbbm{1}_o\right)^\circ$, established in Proposition \ref{prop:maxtopple}, and add to it $1\frac{1}{3}\cdot 3^n = 4\cdot 3^{n-1}$ chips at $o$. The stabilization then proceeds as follows:
\begin{align}
\label{eq:Mstab}
\begin{array}{m{1.6in}lm{1.6in}lm{1.6in}l}
\includegraphics[width=0.25\textwidth]{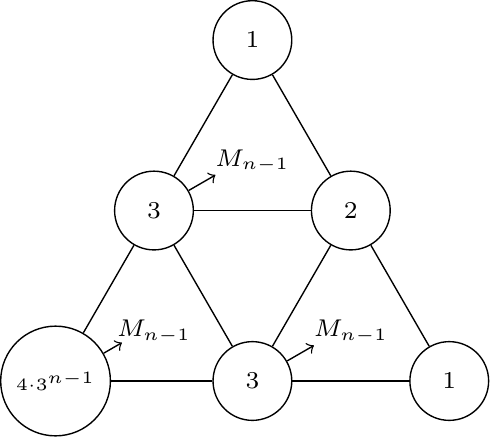}
&
\overset{\eqref{eq:rec2top}}{\longrightarrow}
&
\includegraphics[width=0.25\textwidth]{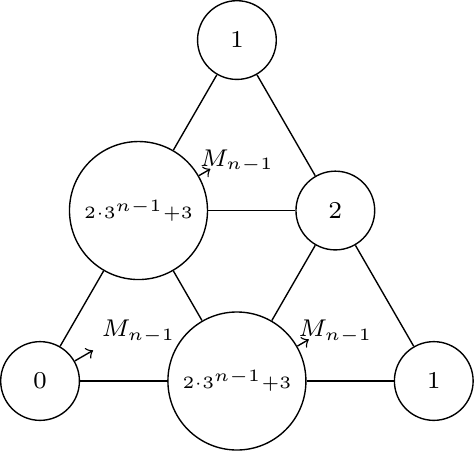}
&
\overset{\eqref{eq:Mnturn2}}{\longrightarrow}
&
\includegraphics[width=0.25\textwidth]{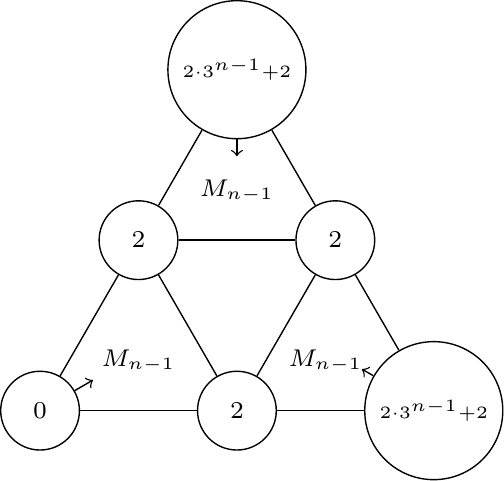}
&.
\end{array}
\end{align}
\end{proof}

It will be helpful to introduce a shorthand for the tile
\[
\begin{array}{m{1.2in}lm{1.2in}l}
\includegraphics[width=0.2\textwidth]{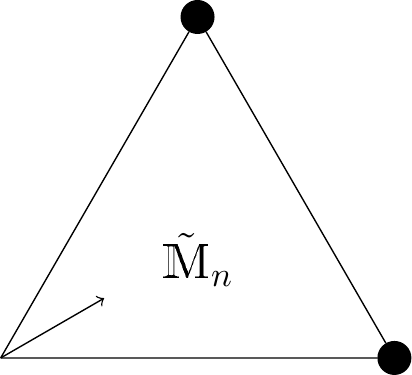}
&
:=
&
\includegraphics[width=0.2\textwidth]{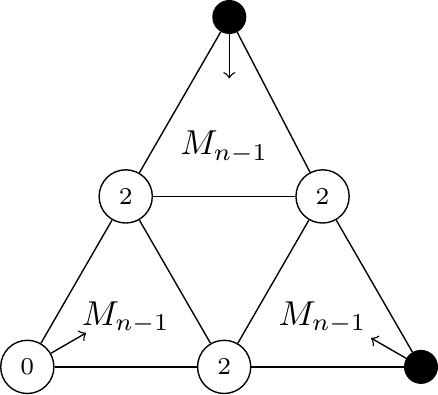}
&
.
\end{array}
\]

\begin{proposition}[``$4\frac{2}{3}$'']
\label{prop:423}
\[
\begin{array}{lm{1in}l}
\left((4\frac{2}{3}\cdot 3^n) \mathbbm{1}_o\right)^\circ = 
&
\includegraphics[height=0.15\textwidth]{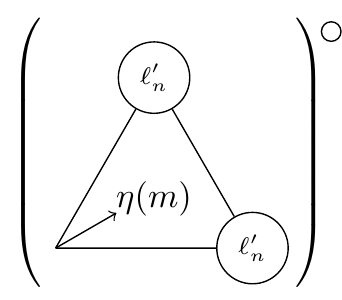}
&
, \text{ where }  \ell_n'= 2\cdot 3^{n-1}+2.
\end{array}
\]
\end{proposition}
\begin{proof}
Using Propositions \ref{prop:423-} and \ref{prop:513-}, we can infer that the number of chips received at the sink is at least $2\cdot 3^{n-1}+1$ and at most $2\cdot 3^{n-1}+2$.
To see that it is the latter count, we add $2$ chips at $o$ to the diagram for $\left((4\frac{2}{3}\cdot 3^n-2)\mathbbm{1}_o\right)^\circ$:
\[
\begin{array}{lm{1.6in}lm{1.6in}l}
\left(4\frac{2}{3}\cdot 3^n\right)\mathbbm{1}_o \longrightarrow
&
\includegraphics[width=0.25\textwidth]{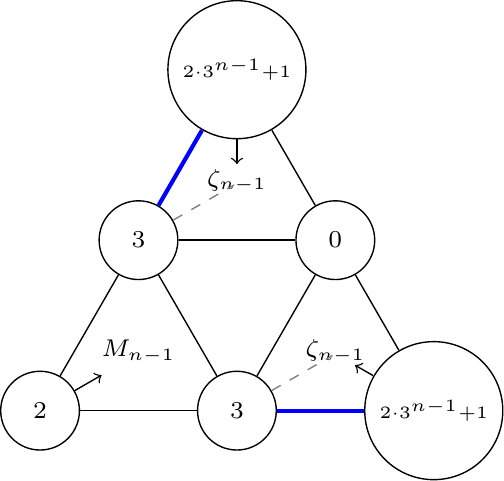}
&
\overset{\eqref{eq:Mtoe}}{\longrightarrow}
&
\includegraphics[width=0.25\textwidth]{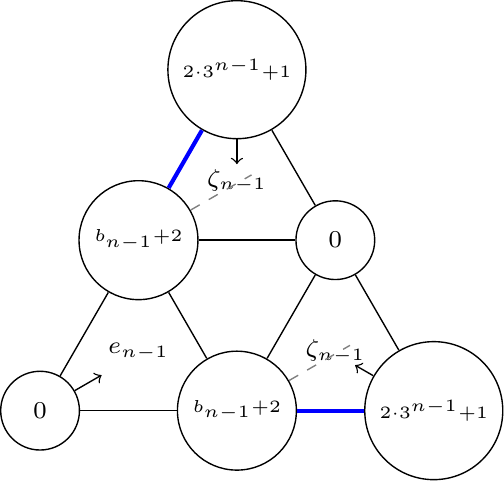}
&
.
\end{array}
\]
We then topple both cut vertices in $\partial G_{n-1}$.
Thanks to the \underline{{\color{blue} \textbf{line of 3's}}} that connects $\partial G_{n-1}$ to the sink $\partial G_n$, this triggers a ``chain reaction'' of topplings down the line and delivers extra chips to $\partial G_n$.
\end{proof}

Next we prove ``$4\frac{4}{9}$'' and $5\frac{1}{3}$'' in tandem.
\begin{proposition}[``$4\frac{4}{9}$'']
\label{prop:449}
\[
\begin{array}{lm{1in}l}
\left((4\frac{4}{9}\cdot 3^n)\mathbbm{1}_o\right)^\circ = 
&
\includegraphics[height=0.15\textwidth]{423-}
&
, \text{ where }  m_n'= 2\cdot 3^{n-1}+1.
\end{array}
\]
\end{proposition}

\begin{proposition}[``$5\frac{1}{3}$'']
\label{prop:513}
\[
\begin{array}{lm{1in}l}
\left(\left(5\frac{1}{3}\cdot 3^n\right) \mathbbm{1}_o\right)^\circ = 
&
\includegraphics[height=0.15\textwidth]{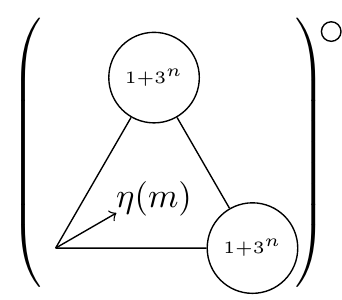}
&
.
\end{array}
\]
\end{proposition}

We first claim that the following stabilization diagrams hold.
For ``$4\frac{4}{9}$'' we start with the configuration $\left((4\cdot 3^n-2)\mathbbm{1}_o\right)^\circ$ and add to it $\frac{4}{9}\cdot 3^n +2 = 4\cdot 3^{n-2}+2$ chips at $o$. This gives
\begin{align*}
\begin{array}{rm{1.6in}lm{1.6in}l}
(4\frac{4}{9}\cdot 3^n)\mathbbm{1}_o \longrightarrow
&
\includegraphics[width=0.25\textwidth]{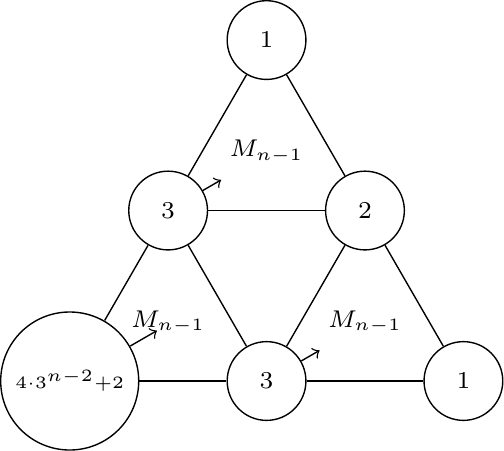}
&
\overset{\eqref{eq:rec2top}}{\longrightarrow}
&
\includegraphics[width=0.25\textwidth]{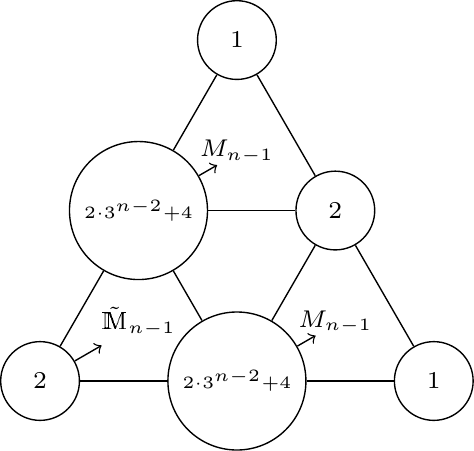}
&
\\
\overset{(\tilde{\mathbbm{M}}_{n-1})}{\longrightarrow}
&
\includegraphics[width=0.25\textwidth]{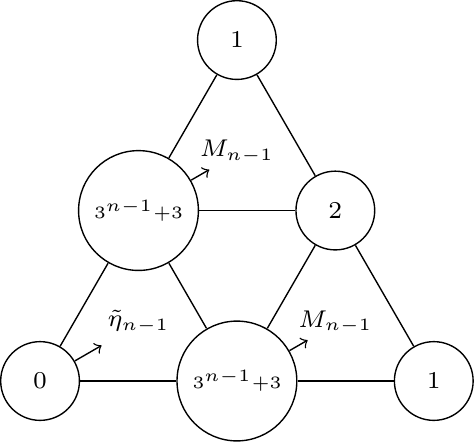}
&
\overset{\eqref{eq:Mn3n}}{\longrightarrow}
&
\includegraphics[width=0.25\textwidth]{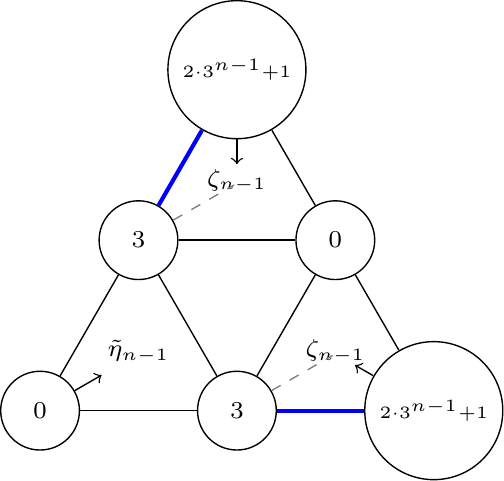}
&,
\end{array}
\end{align*}
where $\tilde\eta_n \in \mathcal{R}_n^{(s)}$.
Let's represent this stabilization by the shorthand
\begin{align}
\label{eq:Mn+4+2}
\begin{array}{m{1in}lm{1in}l}
\includegraphics[width=0.17\textwidth]{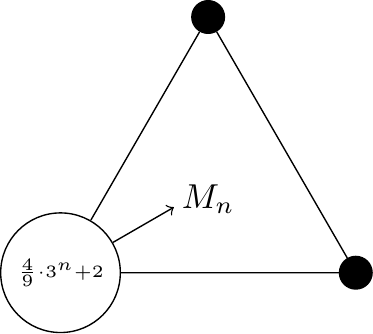}
&
\overset{(\mathcal{Q}_n)}{\longrightarrow}
&
\includegraphics[width=0.17\textwidth]{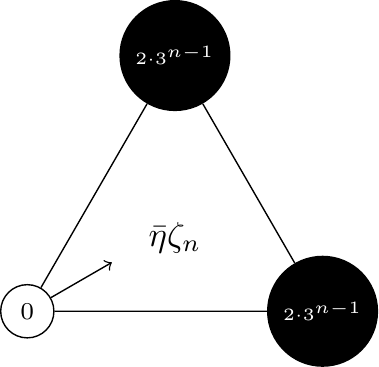}
&.
\end{array}
\end{align} 

For ``$5\frac{1}{3}$'' we start with the configuration $\left((5\frac{1}{3}\cdot 3^n-2)\mathbbm{1}_o\right)^\circ$, proved in Proposition \ref{prop:513-}, and add to it $2$ chips at $o$.
This entire stabilization process is denoted $(\tilde{\mathbbm{M}}_n)$:
\begin{align*}
\begin{array}{m{1.6in}lm{1.6in}lm{1.6in}l}
\includegraphics[width=0.25\textwidth]{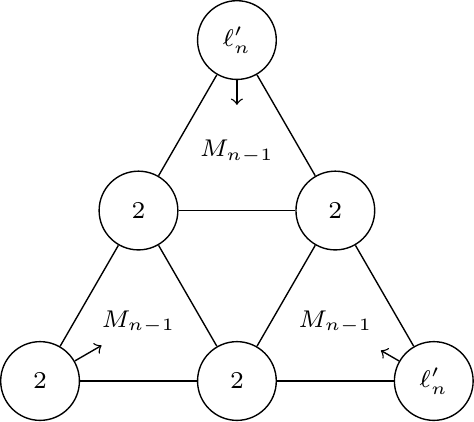}
&
\overset{\eqref{eq:rec2top}}{\longrightarrow}
&
\includegraphics[width=0.25\textwidth]{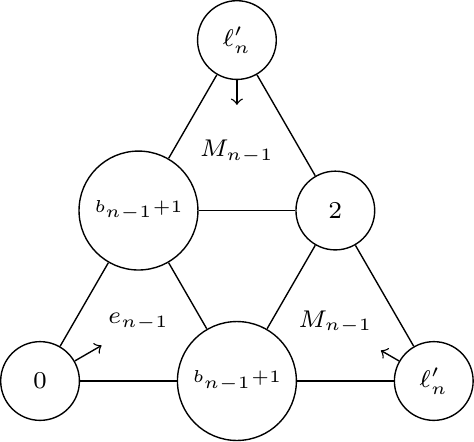}
&
\overset{(Mb_{n-1})}{\longrightarrow}
&
\includegraphics[width=0.25\textwidth]{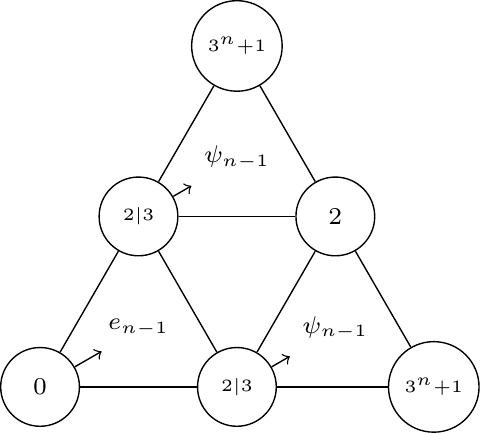}
&
,
\end{array}
\end{align*}
where $\psi_n \in \mathcal{R}_n^{(o)}$, and the notation ``$a|a+1$'' means that the vertex carries either $a$ or $a+1$ chips depending on the odd/even parity of $n$ (and of $b_n$).
In short we have
\begin{align}
\label{eq:Mnrot+2}
\begin{array}{m{0.9in}lm{0.9in}l}
\includegraphics[width=0.15\textwidth]{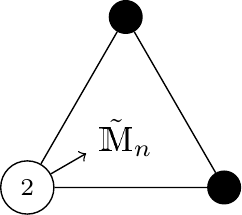}
&
\overset{(\tilde{\mathbbm{M}}_n)}{\longrightarrow}
&
\includegraphics[width=0.15\textwidth]{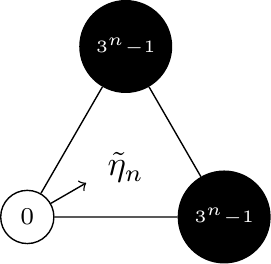}
&.
\end{array}
\end{align} 

Finally, we have the stabilization ($Mb_n$), which is established via induction on $n$:
\begin{align}
\label{eq:Mnbn}
\begin{array}{m{1in}lm{1in}l}
\includegraphics[height=0.15\textwidth]{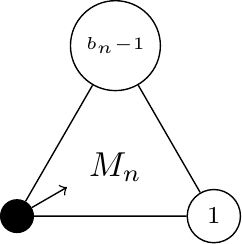} &
\overset{(Mb_n)}{\longrightarrow} &
\includegraphics[height=0.155\textwidth]{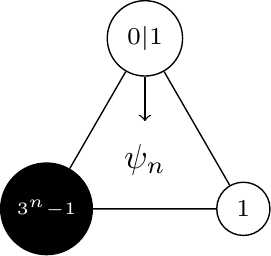} &
.
\end{array}
\end{align}
The $n=1$ case is verified directly.
Suppose \eqref{eq:Mnbn} holds at level $n$.
Then in conjunction with \eqref{eq:Mnturn} we deduce that for $k\in [b_n-1, 3^n]$,
\begin{align}
\label{eq:Mnbn3n}
\begin{array}{m{1in}lm{1in}l}
\includegraphics[height=0.15\textwidth]{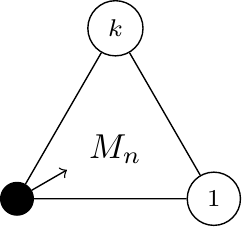} &
\longrightarrow &
\includegraphics[height=0.155\textwidth]{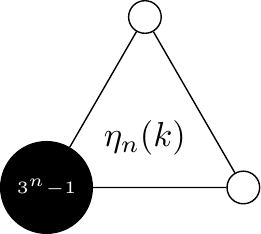}
&,
\end{array}
\end{align}
where $\eta_n(k) \in \mathcal{R}_n^{(o)}$.
The induction step then proceeds as follows, noting that $b_{n+1} = b_n +3^n$, and $b_n -2 =\frac{3}{2}(3^{n-1}+1)-2 \in \left[\frac{4}{9}\cdot 3^n, \frac{2}{3}\cdot 3^n \right)$:
\begin{center}
\begin{tabular}{lm{1.6in}lm{1.6in}lm{1.6in}}
&
\includegraphics[width=0.25\textwidth]{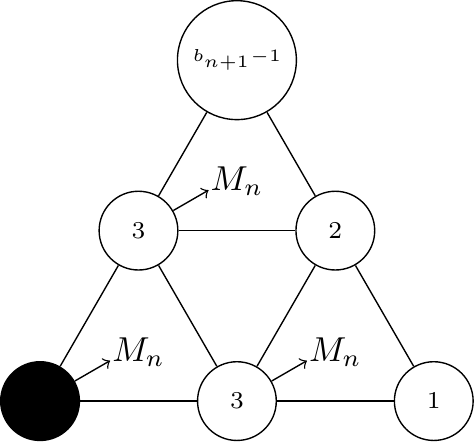}
&
$\overset{\eqref{eq:Mnturn}}{\longrightarrow}$
&
\includegraphics[width=0.25\textwidth]{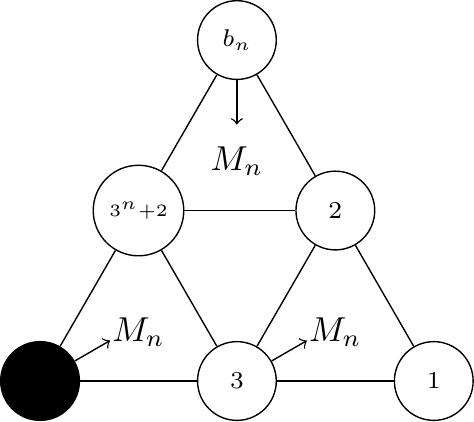}
&
$\overset{\substack{(\mathcal{Q}_n)+ \\ \text{Prop.\@ \ref{prop:423-}}}}{\longrightarrow}$
&
\includegraphics[width=0.25\textwidth]{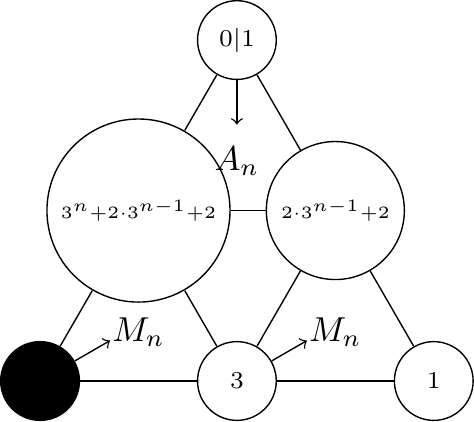}
\\
$\overset{\eqref{eq:Mnbn3n}}{\longrightarrow}$
&
\includegraphics[width=0.25\textwidth]{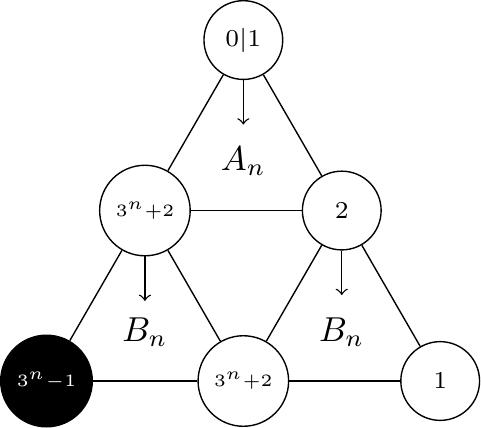}
&
$\overset{\eqref{eq:rev1top}}{\longrightarrow}$
&
\includegraphics[width=0.25\textwidth]{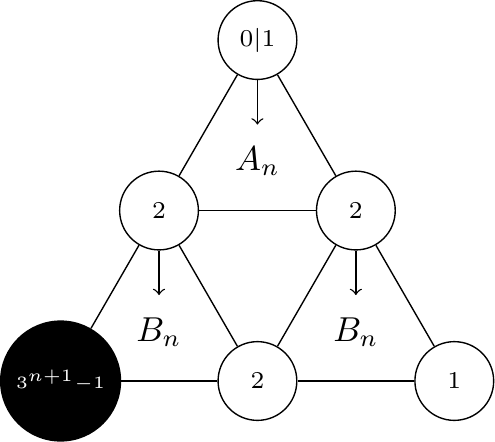}
&.
\end{tabular}
\end{center}
In the above, $A_n$ and $B_n$ are two recurrent configurations whose patterns are complicated to describe, according to our sandpile simulations.

\begin{proof}[Proof of Propositions \ref{prop:449} and \ref{prop:513}]
Looking at \eqref{eq:Mn+4+2}, \eqref{eq:Mnrot+2}, and \eqref{eq:Mnbn} and how they are established diagramatically, it suffices to prove the following trio of implications, where ``OC'' denotes ``other conditions'' that have already been established independently:
\begin{align*}
(\tilde{\mathbbm{M}}_n) + \text{OC} &\Longrightarrow (\mathcal{Q}_{n+1}), \\
(Mb_n) + \text{OC} &\Longrightarrow (\tilde{\mathbbm{M}}_{n+1}), \\
(Mb_n) + (\mathcal{Q}_n) + \text{OC} &\Longrightarrow (Mb_{n+1}).
\end{align*}
The proof by induction on $n$ is straightforward.
\end{proof}

\begin{remark}
\label{rem:firsttimerec}
We observe from our numerical computations that $\left.((10\cdot 3^{n-1})\mathbbm{1}_o)^\circ \right|_{G_n^{(s)}}$ is a recurrent configuration.
If this observation can be proved directly, then ``$5\frac{1}{3}$'' can be established without resort to diagrams, since $5\frac{1}{3} \cdot 3^n = 10\cdot 3^{n-1}+2\cdot 3^n$, and we can use Corollary \ref{cor:constantrad} and the $(2\cdot 3^n)$-periodicity to deduce that each sink in $\partial G_n$ receives $3^n+1$ chips.
\end{remark}

\subsection{Enumeration of radial jumps (II): analysis of traps along the space-filling curve} \label{sec:enumeration2}

It remains to prove ``$e+2$'' and ``$4\frac{4}{9}^-$.'' 
For these we take advantage of the structure of the identity element $e_n$, namely, the space-filling curve along which every vertex carries $3$ chips. See Figure \ref{fig:4+2path}.

\begin{proposition}[``$e+2$'']
\label{prop:4+2}
For every $n\in \mathbb{N}$, adding $2\mathbbm{1}_o$ to $e_n$ and then stabilizing results in 1 chip received by each sink in $\partial G_n$. Therefore 
\[
\begin{array}{lM{1in}l}
\left((4\cdot 3^n+2)\mathbbm{1}_o\right)^\circ = 
&
\includegraphics[height=0.15\textwidth]{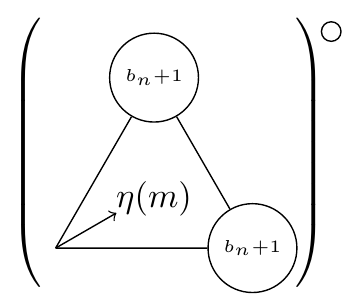}
&
.
\end{array}
\]
\end{proposition}
\begin{proof}
Given the configuration $e_n$, adding 2 chips at $o$ triggers a chain of topplings along each space-filling curve, thereby sending 1 chip to each sink vertex.
That no more chips can drop into the sink is implied by Proposition \ref{prop:449-} below.
\end{proof}

Our next task is to analyze the configuration $e_n \oplus 2\mathbbm{1}_o$ in $G_n^{(s)}$, and explain how traps appear along the space-filling curve.
As a warm-up exercise let us work through two toy examples.

First consider a ``line graph.''
Take a positive integer $N\geq 2$. 
Consider the graph whose vertex set is $(0,0) \cup \{1,\cdots, N-1\}\times \{-1,0,1\} \cup (N,0)$, and whose edge set consists of edges of the form $(x,0)\sim(x+1,0)$, $0\leq x\leq N-1$; $(x,1)\sim(x,0)$, $1\leq x\leq N-1$; and $(x,-1)\sim(x,0)$, $1\leq x\leq N-1$.
Designate $(0,0)$ as the origin, $(N,0)$ as the sink, and suppose on every vertex $(x,0)$ along the line carries the maximal number of chips (3), and every vertex elsewhere carries $0$ chips.
It is easy to verify that the $1$ chip at the origin is transmitted all the way to the sink without changing the configuration on and off the line:
\[
\begin{array}{M{0.4\textwidth}lM{0.4\textwidth}l}
\includegraphics[width=0.4\textwidth]{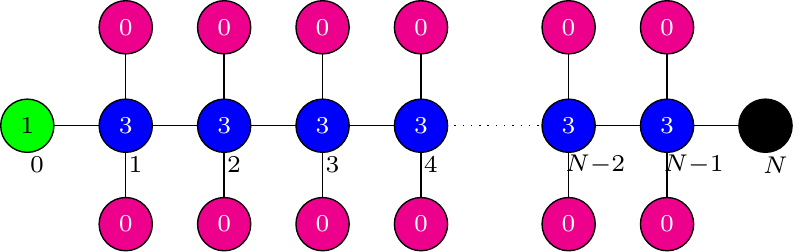}&
\longrightarrow&
\includegraphics[width=0.4\textwidth]{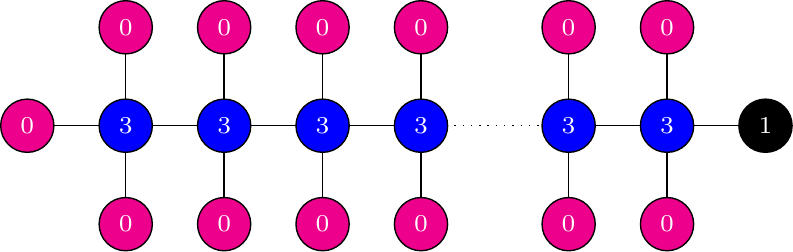}&
.
\end{array}
\]

Next we modify the above graph by identifying vertices off the line.
For simplicity, take $N=4$ and identify the vertices $(1,1)$ and $(3,1)$.
We then carry out the stabilization systematically by alternating stabilizations on and off the line.
\[
\begin{array}{rM{0.2\textwidth}lM{0.2\textwidth}lM{0.2\textwidth}l}
&
\includegraphics[width=0.2\textwidth]{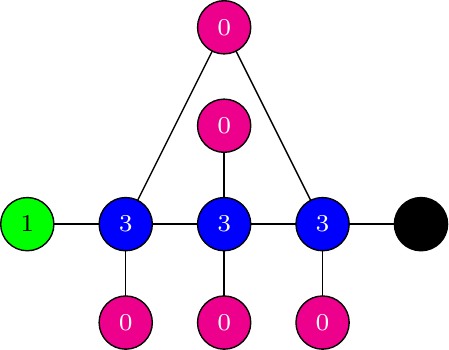}&
\longrightarrow &
\includegraphics[width=0.2\textwidth]{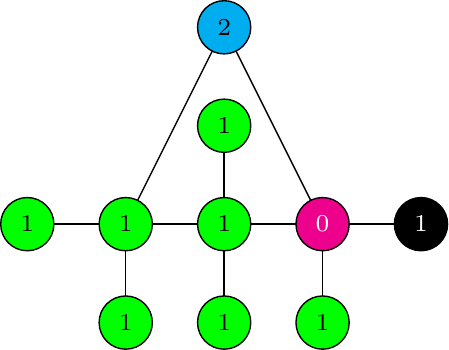}&
\longrightarrow &
\includegraphics[width=0.2\textwidth]{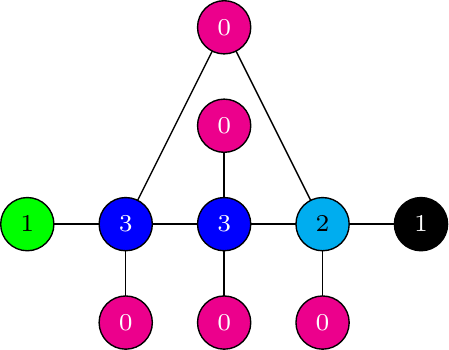}
\\
\longrightarrow &
\includegraphics[width=0.2\textwidth]{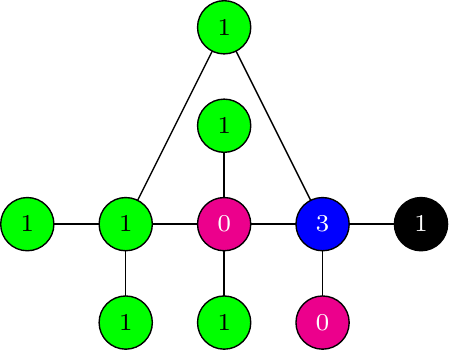}&
\longrightarrow &
\includegraphics[width=0.2\textwidth]{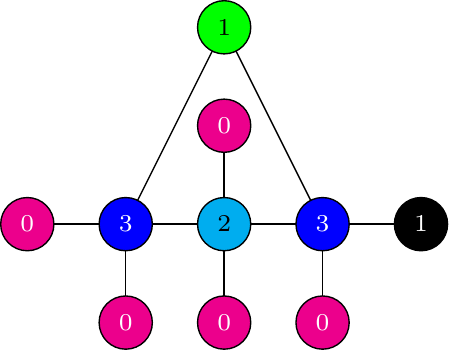}&
.
\end{array}
\]
Observe the appearance of the vertex ``$2$'' on the line: this forms a ``trap'' in the sense that the next chip to travel down the line will be stopped by it.
To ``re-open'' the line of communication requires additional chips at the origin.
Of course this does not preclude additional chip(s) from traveling down the ``sky hook,'' origin--3--1--3--sink.
At any rate, we emphasize that \emph{the identification of off-line vertices results in the formation of traps on the line.}

Having explained the trap mechanism in the simple setting, we now proceed to the actual analysis.
Parametrize one half of $G_n^{(s)}$ by the length of the space-filling curve.
Starting with 2 chips at the origin, we carry out the stabilization systematically.
Due to the identifications of the off-line vertices, this triggers a series of topplings backwards along the curve, resulting in the creation of a finite number of well-separated \emph{traps}, represented by blotches of 1's and 2's; see the left configuration in Figure \ref{fig:traps}.
Figure \ref{fig:trapdynamics} shows the precise mechanism of trap creation in the case $e_3$.

A key observation is that whenever the on-line propagation stops in the middle of a ``\emph{sky hook}''---where (far-away) vertices on-line are connected to the same off-line vertex---this results in the creation of a trap at the start of the hook. See the transition from (c) to (e), and from (g) to (i), in Figure \ref{fig:trapdynamics}.
This leaves the configuration in some portion of the space-filling curve unchanged.
Thus after stabilization is complete, the space-filling curve contains several traps, and connecting two adjacent traps is a \emph{corridor}: each vertex on the corridor carries $3$, and off it, $0$ chip.

\begin{figure}
\centering
\begin{tabular}{M{0.4\textwidth} M{0.4\textwidth} M{0.07\textwidth}}
\includegraphics[width=0.4\textwidth]{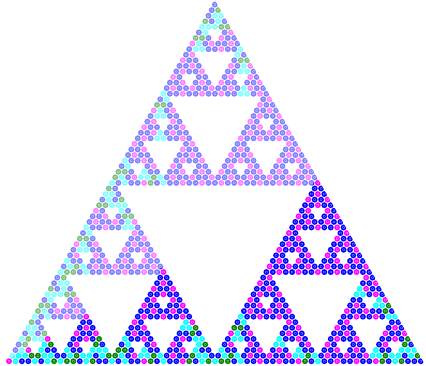}
&
\includegraphics[width=0.4\textwidth]{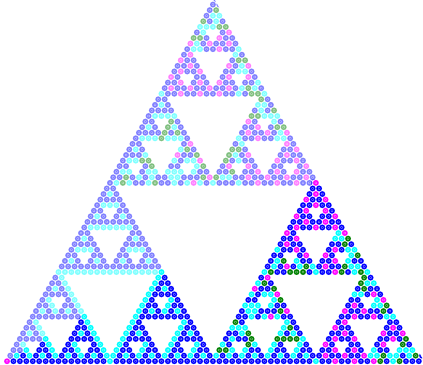}
&
\includegraphics[width=0.05\textwidth]{colorlegend}
\end{tabular}
\caption{The configuration $e_5 \oplus 2\mathbbm{1}_o$ (left) and $e_5 \oplus \left(4\cdot 3^{5-2}\right)\mathbbm{1}_o$ (right). Blotches of 1's and 2's are ``traps,'' which appear in well-defined locations.}
\label{fig:traps}
\end{figure}

\begin{figure}
\centering
\begin{small}
\begin{tabular}{M{0.2\textwidth} l M{0.55\textwidth} l}
\includegraphics[width=0.15\textwidth]{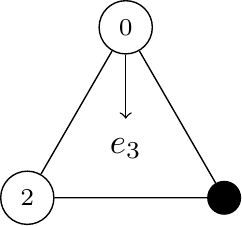} & $=$ &
\includegraphics[width=0.55\textwidth]{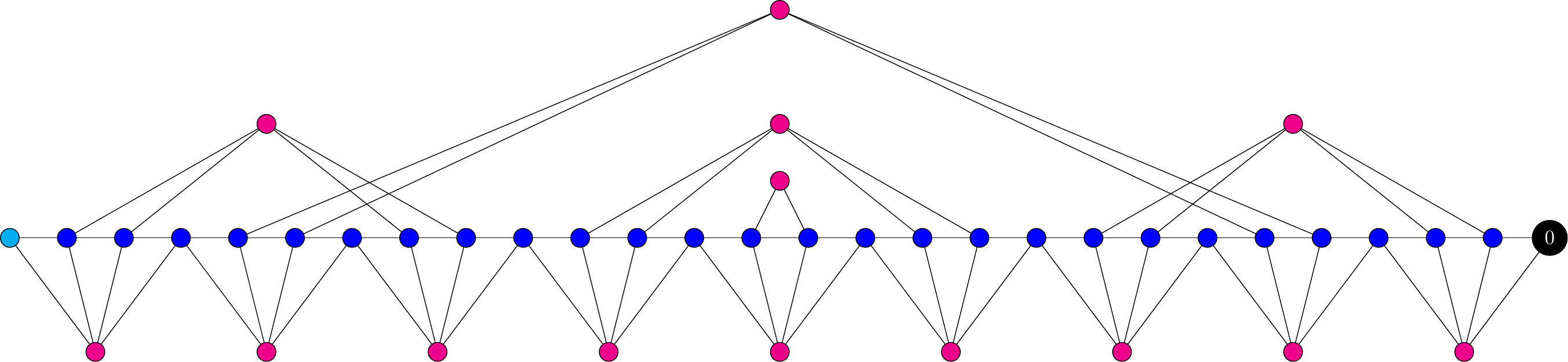} & (a) \\
&\multirow{9}{*}{
\tikz{\draw [-{Latex[length=5mm, width=2mm]}] (0,0) -- (0,-18);}
}&
\includegraphics[width=0.55\textwidth]{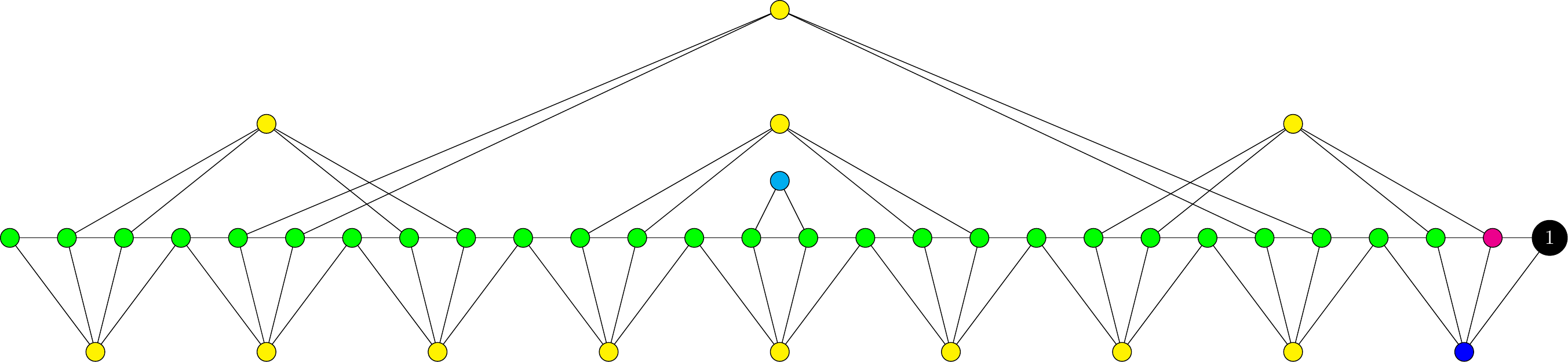} & (b) \\
&&\includegraphics[width=0.55\textwidth]{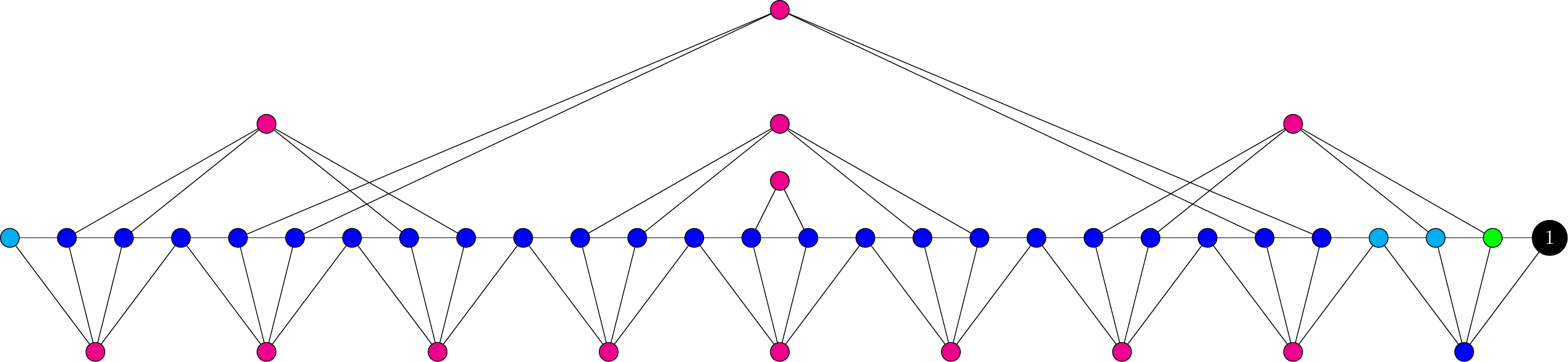} & (c) \\
&&\includegraphics[width=0.55\textwidth]{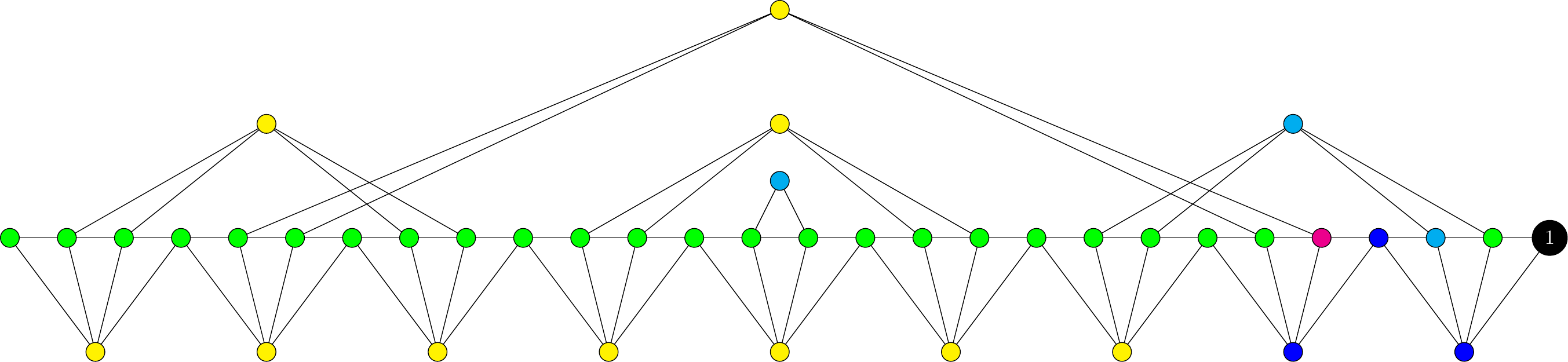} & (d)\\
&&\includegraphics[width=0.55\textwidth]{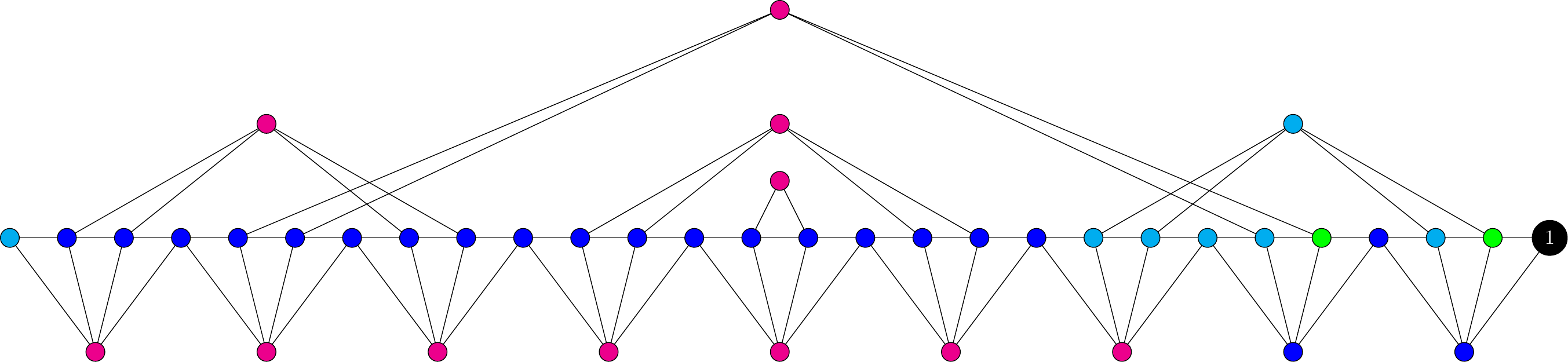} & (e) \\
&&\includegraphics[width=0.55\textwidth]{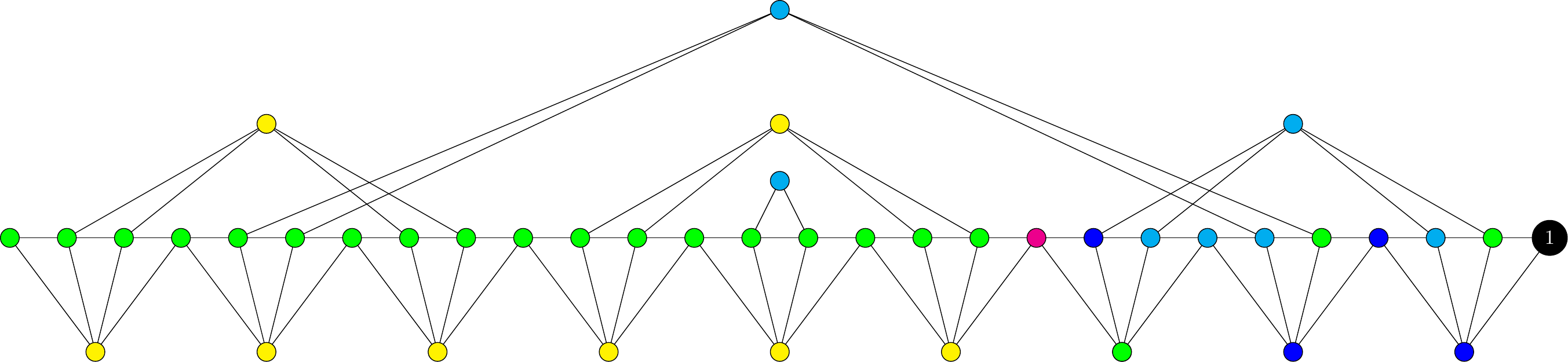} & (f) \\
&&\includegraphics[width=0.55\textwidth]{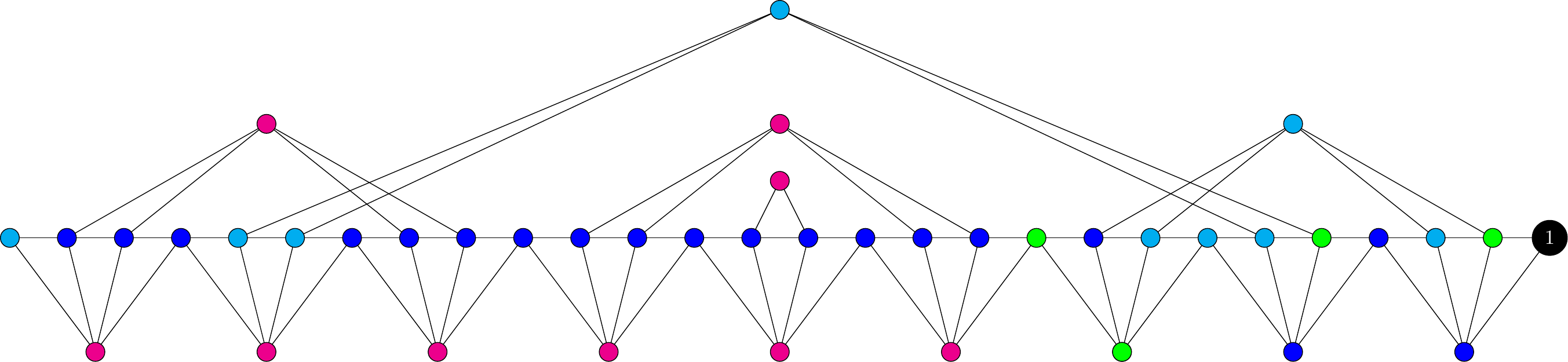} & (g)\\
&&\includegraphics[width=0.55\textwidth]{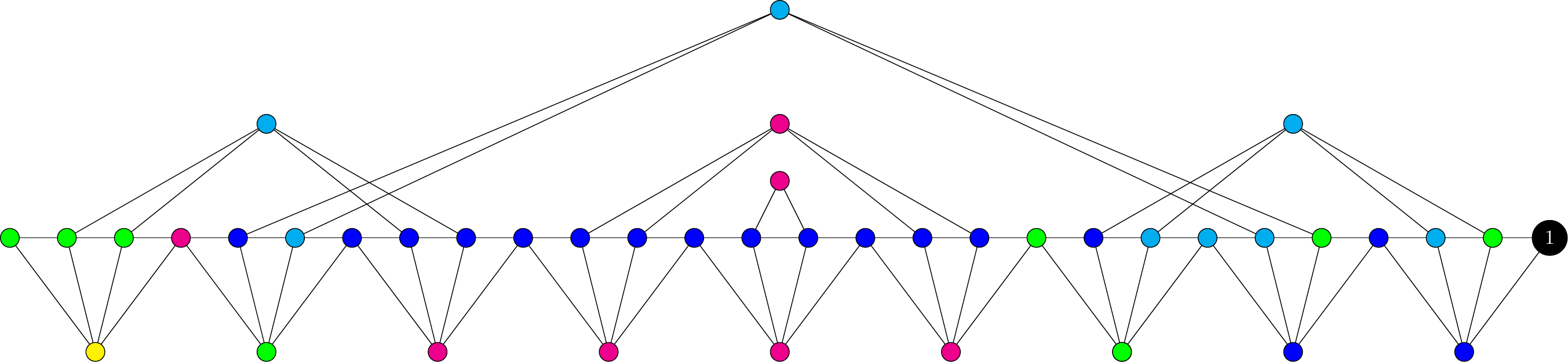} & (h)\\
&&\includegraphics[width=0.55\textwidth]{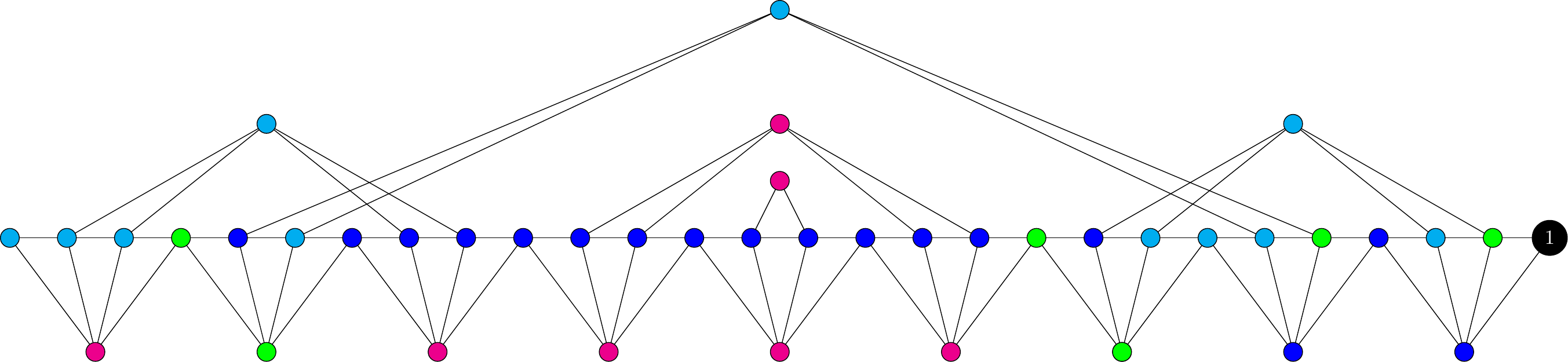} & (i) \\
&&\includegraphics[width=0.55\textwidth]{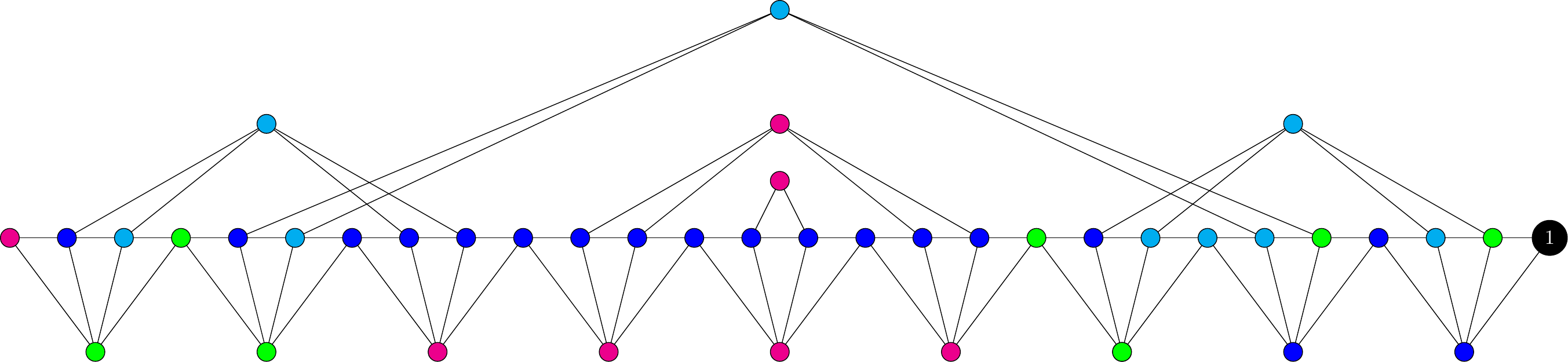} &(j)
\end{tabular}
\hspace{10pt}
\begin{tabular}{M{7pt} M{7pt}}
0 & \tikz{\node [draw, circle, fill=magenta] at (0,0) {};}\\
1 & \tikz{\node [draw, circle, fill=green] at (0,0) {};}\\
2 & \tikz{\node [draw, circle, fill=cyan] at (0,0) {};}\\
3 & \tikz{\node [draw, circle, fill=blue] at (0,0) {};}\\
4 & \tikz{\node [draw, circle, fill=yellow] at (0,0) {};}
\end{tabular}
\end{small}
\caption{Stabilization after adding 2 chips at the start of the space-filling curve in $e_3$. This is carried out systematically by alternating stabilizations on and off the space-filling curve. The final configuration consists of two traps and a corridor.}
\label{fig:trapdynamics}
\end{figure}

Given the landscape of traps, the goal is to add enough chips at $o$ so that they can overcome one trap after another.
Keep in mind, however, that whenever a chip lands on a corridor but is stopped by the next trap, the aforementioned argument implies that additional traps may be created along the corridor.

Now that we have laid out the key observations, it is time to establish the following comparison lemma, which says that over the configuration $e_n$, it is ``easier'' to send chips down the axial direction than the oblique direction.

\begin{lemma}
\label{lem:comparison}
If
\begin{tabular}{M{0.7in}}
\includegraphics[width=0.7in]{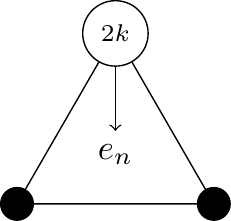}
\end{tabular}
does not result in more than 1 chip received by each sink, then neither does 
\begin{tabular}{M{0.75in}}
\includegraphics[width=0.75in]{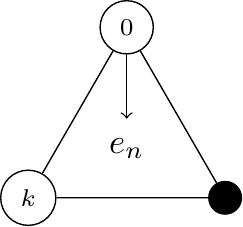}
\end{tabular}
.
\end{lemma}
\begin{proof}
First note that by the axial symmetry, it is enough to split
\begin{tabular}{M{0.7in}}
\includegraphics[width=0.7in]{enaxial}
\end{tabular}
in two halves, and consider whether
\begin{tabular}{M{0.7in}}
\includegraphics[width=0.7in]{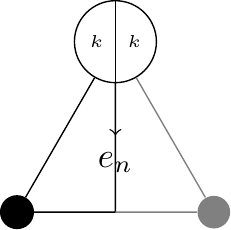}
\end{tabular}
results in more than $1$ chip received by the sink.
With this proviso we now parametrize the two configurations in the statement by the length of the space-filling curve.
When the origin topples once, the resulting landscape of traps (\tikz{\node [draw, scale=0.6,circle, fill=green, text=black] at (0,0) {\scriptsize{T}};}) is shown below. (In the right-hand diagrams, each tick mark represents distance $3^{n-3}$ in the graph metric along the space-filling curve. For clarity, traps, corridors, and sky hooks of length scale $<3^{n-3}$ are omitted from the diagrams.)

\begin{tabular}{M{0.12\textwidth} l P{0.82\textwidth}}
\includegraphics[width=0.11\textwidth]{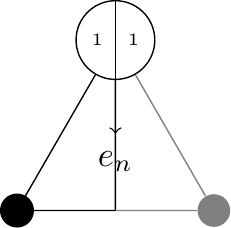}&
$\longrightarrow$ &
\includegraphics[width=0.42\textwidth]{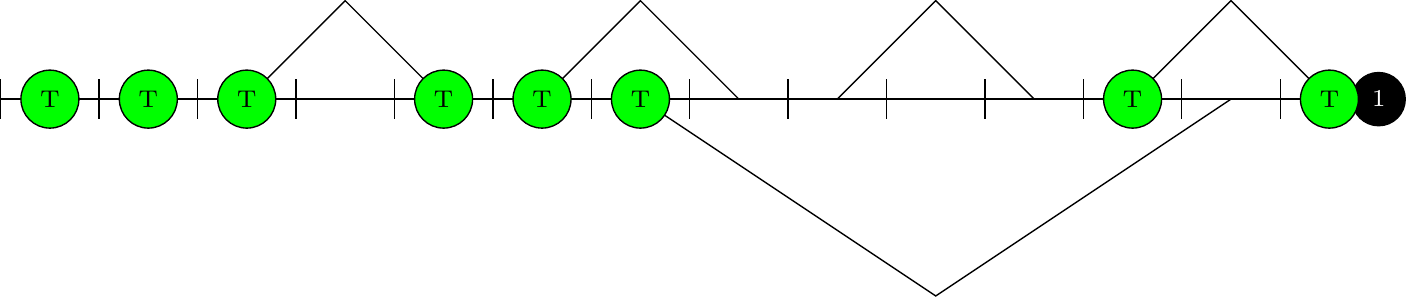}\\
\includegraphics[width=0.12\textwidth]{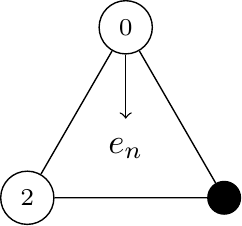}&
$\longrightarrow$ &
\includegraphics[width=0.81\textwidth]{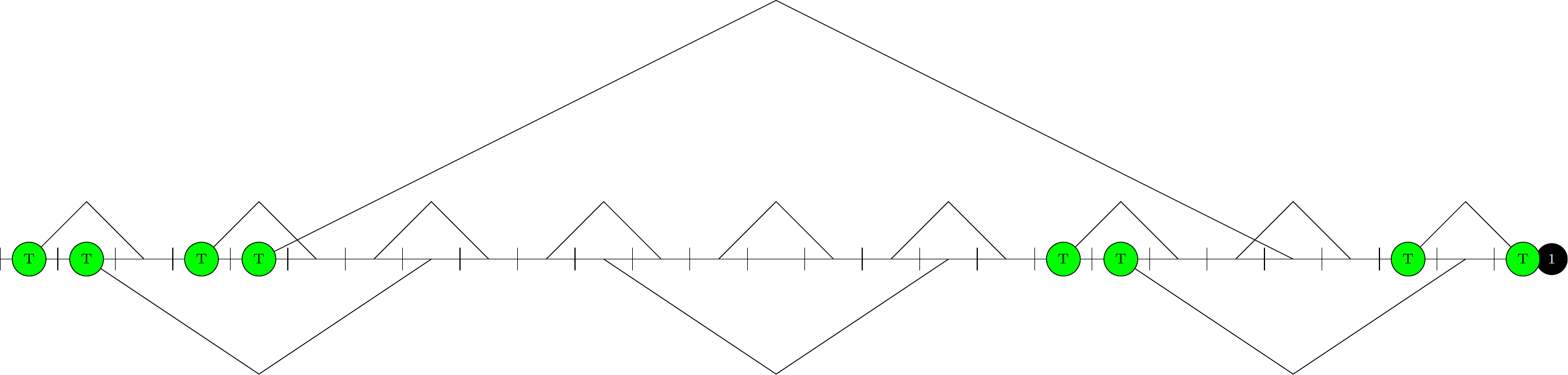}
\end{tabular}

Suppose 
\begin{tabular}{M{0.7in}}
\includegraphics[width=0.7in]{enaxialhalf}
\end{tabular}
does not result in more than $1$ chip received by each sink.
In the best-case, or most greedy, scenario, chips topple along the shortest path (marked in {\color{red} thick red line}) in such a way that all but the final trap has been overcome. Traps which are guaranteed to be overcome in this scenario are denoted by \tikz{\node [draw, scale=0.6,circle, fill=blue, text=white] at (0,0) {\scriptsize{T}};}.
Observe that there are 7 traps along the shortest path which are overcome.
\vspace{-10pt}
\begin{align}
\label{eq:axial2}
\includegraphics[width=0.42\textwidth]{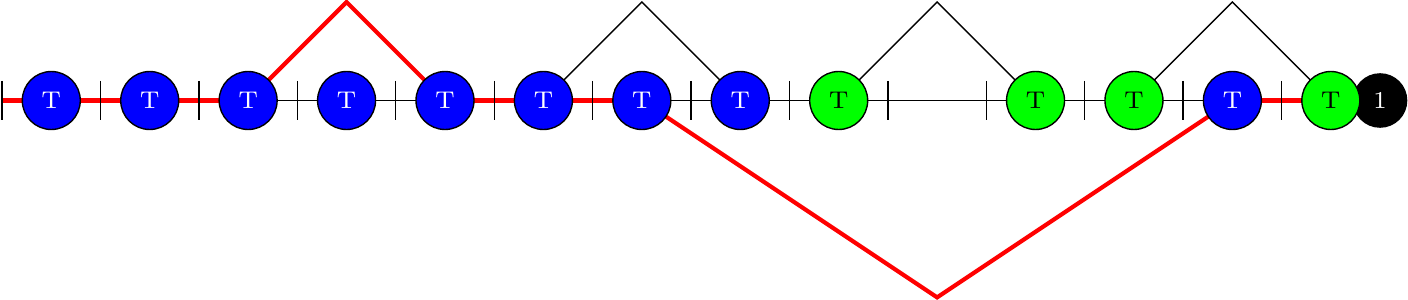}
\end{align}
As a result, when we stabilize 
\begin{tabular}{M{0.7in}}
\includegraphics[width=0.7in]{enoblique}
\end{tabular}
we also expect, in the best-case scenario, no more than $7$ traps overcome along the shortest path, as indicated in the diagram below:
\begin{align}
\label{eq:oblique2}
\includegraphics[width=0.81\textwidth]{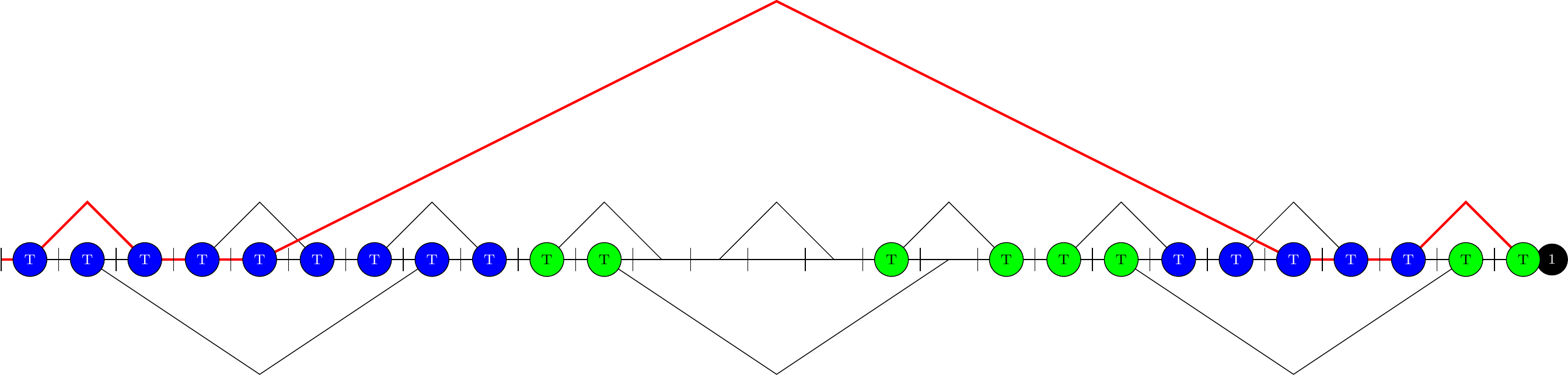}
\end{align}
In fact, since the graph in \eqref{eq:oblique2} has additional branching compared to the graph in \eqref{eq:axial2}, the diagram \eqref{eq:oblique2} overrepresents the number of topplings than the actual case. 
Nevertheless it does indicate that no more than 1 chip can drop into the sink.
\end{proof}

\begin{proposition}[``$4\frac{4}{9}^-$'']
\label{prop:449-}
\[
\begin{array}{lm{1in}l}
\left((4\frac{4}{9}\cdot 3^n -2)\mathbbm{1}_o\right)^\circ = 
&
\includegraphics[height=0.15\textwidth]{4+2}
&
.
\end{array}
\]
\end{proposition}
\begin{proof}
This is equivalent to the stabilization
\begin{align}
\label{eq:49hypo}
\begin{array}{m{1.1in}lm{1in}l}
\includegraphics[width=1.1in]{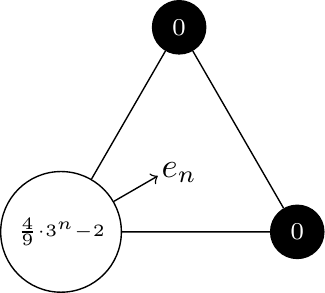} &
\longrightarrow &
\includegraphics[width=1in]{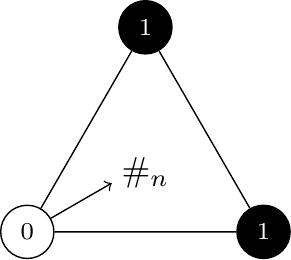} &
,
\end{array}
\end{align}
which we prove by induction on $n$.
The base case is verified directly.
Assume \eqref{eq:49hypo} holds at level $n-1$.
Observe that to obtain $\left((4\frac{4}{9}\cdot 3^n -2)\mathbbm{1}_o\right)^\circ$, we start with the configuration $\left((4\cdot 3^n-2)\mathbbm{1}_o\right)^\circ$ and add to it $\frac{4}{9}\cdot 3^n = 4\cdot 3^{n-2}$ chips at $o$.
\begin{align*}
\begin{array}{rm{1.6in}lm{1.6in}l}
&
\includegraphics[width=1.6in]{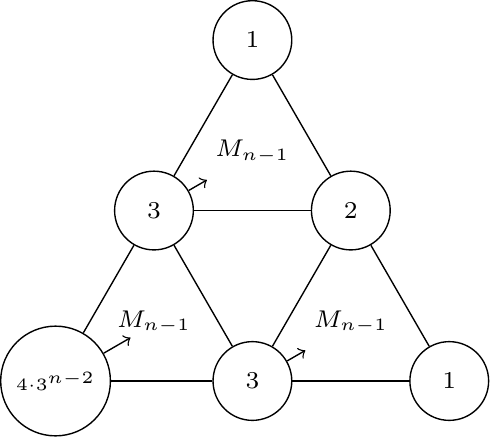} &
\overset{\eqref{eq:Mstab}}{\longrightarrow} &
\includegraphics[width=1.6in]{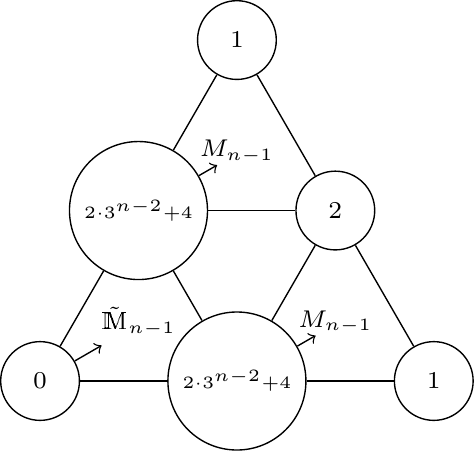}
& \\
\overset{\eqref{eq:Mtoe}}{\longrightarrow} &
\includegraphics[width=1.6in]{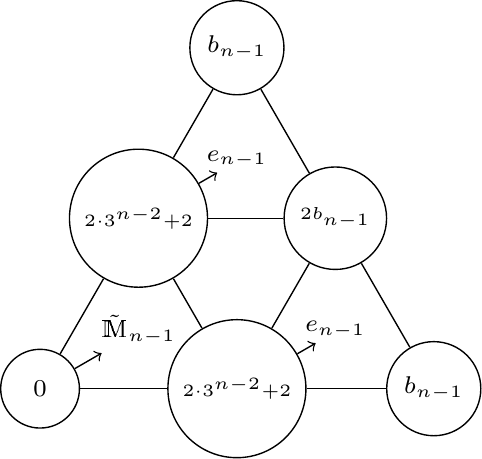} &
\overset{\eqref{eq:enturn}}{\longrightarrow} &
\includegraphics[width=1.6in]{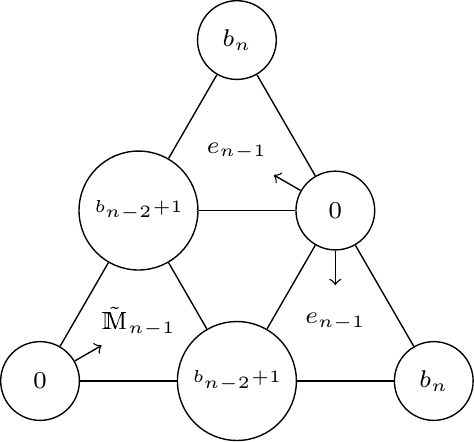} &
.
\end{array}
\end{align*}
In the next step we topple the $b_{n-2}+1$ chips at the cut point in $\partial G_n$.
Recall that in $e_{n-1}$ there is a path connecting $\partial G_{n-1}$ and $\partial G_n$ along which every vertex carries $3$ chips, \emph{cf.\@} Proposition \ref{prop:4+2}.
Thus the first topple will trigger a chain of topplings which sends 1 chip to the sink in $\partial G_n$.
We now claim that no additional chips can drop into $\partial G_n$, \emph{i.e.,}
\begin{align*}
\begin{array}{m{1in}lm{1in}l}
\includegraphics[width=1in]{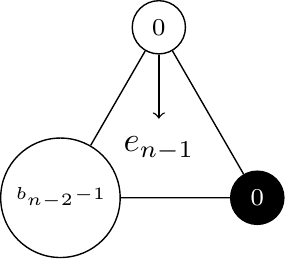} &
\longrightarrow &
\includegraphics[width=1in]{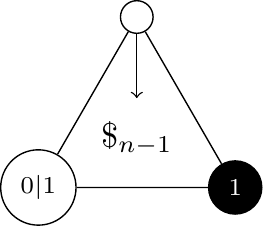} &
.
\end{array}
\end{align*}
Since $b_{n-2}-1 = \frac{3}{2}(3^{n-3}+1)-1 = \frac{1}{6}\cdot 3^{n-1} +\frac{1}{2}< \frac{1}{2}\left(\frac{4}{9}\cdot 3^{n-1} -2\right)$, the claim follows from the induction hypothesis \eqref{eq:49hypo} at level $n-1$ and the comparison Lemma \ref{lem:comparison}. We thus verify \eqref{eq:49hypo} at level $n$.
\end{proof}

\begin{proof}[Proof of Theorem \ref{thm:tail}]
So far we have proved the results for $p=0$. To obtain the result for $p\in \{1,2,3\}$, we just add to every existing diagram $2p\cdot 3^n$ chips at $o$. By \eqref{eq:rec2top}, stabilization results in adding $p\cdot 3^n$ chips to each sink in $\partial G_n$.
\end{proof}

\subsection{Recursive formula for the cluster radius}
\label{sec:exact}

\begin{proof}[Proof of Theorem \ref{thm:radialcycle}]
We combine the results in Theorem \ref{thm:tail} with Proposition \ref{prop:ball}, Item \eqref{item:radial}.
Specifically, given $m'-2$, we identify the numbers $c$ and $d$ such that the $r(x)$ is constant on $[c,d)$ and $m'-2 \in [c,d)$.
This then leads to the claimed radial recursions (see also Figure \ref{fig:radialcycle}), as well as the complete characterization of the radial jumps.
The reason that we specify the results for $n\geq 3$ (or $n\geq 4$) is due to jumps for $n\in \{1,2,3\}$ which do not follow the periodicity stated in Theorem \ref{thm:tail}, \emph{cf.\@} Table \ref{table:spectrum} and Figure \ref{fig:ASMScaling}.
\end{proof}

\subsection{Cluster growth asymptotically follows a power law modulated by log-periodic oscillations} \label{sec:osc}

The goal of this subsection is to prove Theorem \ref{thm:ASM}, Part (\ref{item:renewal}).
The proof employs the renewal theorem, which is widely used in the study of fractal geometry, see \emph{e.g.\@} \cites{FalconerTFG, KigamiLapidus, LV96, Lalley}.
The key input is

\begin{proposition}[Remainder estimate]
\label{prop:scalingest}
For all $m\in \mathbb{N}$,
$
|r_{3m}-2r_m| \leq 1
$.
\end{proposition}
\begin{proof}
Theorem \ref{thm:radialcycle} states that given $a\in [4,12)$, there exists $b\in [4,12)$ such that
\[
\text{either } ~r_{a\cdot 3^n} = 2^n + r_{b\cdot 3^{n-2}} \text{ for } n\geq 4, ~\text{ or } r_{a\cdot 3^n} = 2^n + r_{b\cdot 3^{n-1}} \text{ for  } n\geq 3.
\]
This implies that 
\begin{align}
\label{radiusinduction}
\text{either } ~r_{a\cdot 3^{n+1}} - 2 r_{a\cdot 3^n} = r_{b\cdot 3^{n-1}} - 2 r_{b\cdot 3^{n-2}}~\text{ for } n\geq 4, \text{ or } r_{a\cdot 3^{n+1}} - 2 r_{a\cdot 3^n} = r_{b\cdot 3^n} - 2 r_{b\cdot 3^{n-1}} \text{ for  } n\geq 3.
\end{align}
So it is enough to check that for every $b\in [1,4\cdot 3^3)$, $|r_{3b} - 2r_{b}| \leq 1$,  \emph{cf.\@} Table \ref{table:spectrum}, and then apply \eqref{radiusinduction} inductively.
\end{proof}

For our purposes the following version of the renewal theorem will suffice. Denote by $\mathcal{F}$ the space of Borel measurable functions $f: \mathbb{R}\to\mathbb{R}$ such that $\lim_{t\to-\infty} f(t)=0$ and such that $f$ is bounded on $(-\infty, a]$ for every $a\in \mathbb{R}$. A Borel measure $\mu$ is said to be \textbf{$\tau$-arithmetic} if $\tau>0$ is the largest number such that the support of $\mu$ is contained in the additive subgroup $\tau \mathbb{Z}$. If no such $\tau$ exists then we say that $\mu$ is \textbf{non-arithmetic}.

\begin{lemma}[Renewal theorem, \emph{cf.\@} \cite{FalconerTFG}*{Proposition 7.1 and Theorem 7.2}]
\label{lem:renewal}
Let $g: \mathbb{R}\to\mathbb{R}$, and $\mu$ be a Borel probability measure supported on $[0,\infty)$, 
Suppose:
\begin{enumerate}[label=(R\arabic*)]
\item \label{cond1} $\displaystyle \lambda:=\int_0^\infty\, t\,d\mu(t) < \infty$.
\item \label{cond2} $\displaystyle \int_0^\infty \,e^{-at}\,d\mu(t) <1$ for every $a>0$.
\item \label{cond3} $g$ has a discrete set of discontinuities, and there exist $c,\alpha>0$ such that $|g(t)| \leq ce^{-\alpha|t|}$ for all $t\in\mathbb{R}$.
\end{enumerate}
Then there is a unique $f \in \mathcal{F}$ which solves the \textbf{renewal equation}
\[
f(t) = g(t) + \int_0^\infty\, f(t-y)\,d\mu(y) \qquad (t\in\mathbb{R})
\]
and the solution is
\[
f(t) = \sum_{k=0}^\infty \left(g* \mu^{* k}\right)(t),
\]
where $(g *\mu)(t) = \int_0^\infty\, g(t-y)\,d\mu(y)$ denotes the convolution of $g$ and $\mu$, and $\mu^{*k}$ denotes the $k$-fold convolution of $\mu$.

Furthermore,
\begin{itemize}
\item If $\mu$ is \textbf{non-arithmetic}, then
\[
\lim_{t\to\infty} f(t) = \frac{1}{\lambda} \int_{-\infty}^\infty \, g(y)\,dy.
\]
\item If $\mu$ is \textbf{$\tau$-arithmetic}, then for all $y\in [0,\tau)$,
\[
\lim_{k\to\infty} f(k\tau+y) = \frac{\tau}{\lambda}\sum_{j=-\infty}^\infty g(j\tau+y).
\]
\end{itemize}
\end{lemma}

\begin{proof}[Proof of Theorem \ref{thm:ASM}, Part \eqref{item:renewal}]
We extend $m\mapsto r_m$ to a function $r: [0,\infty) \to \mathbb{N}$ via $r(x)= r_{\lfloor x\rfloor}$. 
It is easy to verify that Proposition \ref{prop:scalingest} extends to the function $r$. 
For the rest of the proof we designate the remainder function
\begin{align}
\label{eq:remainder}
R: [0,\infty) \to \{-1, 0, 1\}, \quad R(x) = r(x) - 2 r\left(\frac{x}{3}\right).
\end{align}
Since $r_1=0$, it follows that $R(x) =0$ whenever $x\leq 1$.

Making the change of variables $x=e^t$ in \eqref{eq:remainder} we obtain
\[
r(e^t) = 2r\left(\frac{e^t}{3}\right) + R(e^t) = 2r(e^{t-\log 3}) + R(e^t).
\]
Multiplying both sides by $e^{-t/d_H}$ yields the renewal equation
\begin{align}
\label{eq:re}
f(t) = f(t-\log 3)+ g(t) = \int_0^\infty\, f(t-y) \, \delta_{\log 3}(dy) + g(t),
\end{align}
where 
\begin{align}
f(t) := e^{-t/d_H}r(e^t) \quad\text{and}\quad g(t) := e^{-t/d_H} R(e^t).
\end{align}
It is clear that $f\in \mathcal{F}$.

To ensure that $f$ is the unique solution to \eqref{eq:re}, we verify Conditions \ref{cond1} through \ref{cond3} of Lemma \ref{lem:renewal}. Firstly $\mu=\delta_{\log 3}$, so $\lambda=\log 3$, verifying \ref{cond1}. \ref{cond2} is validated by noting that the relevant integral equals $e^{-a\log 3}= 3^{-a} <1$ for $a>0$. Last but not least, since $x\mapsto R(x)$ has only jump discontinuities at the positive integers, it is straightforward to deduce that the set of discontinuities of $g$ is discrete. In addition
$g(t)$ decays exponentially in $|t|$, \emph{i.e.,}
\[
|g(t)| \left\{\begin{array}{ll} \leq e^{-t/d_H}, & t>0, \\ =0, & t\leq 0. \end{array}\right.
\]
This verifies \ref{cond3}.

By Lemma \ref{lem:renewal}, $f$ is the unique solution to \eqref{eq:re}. As $\mu$ is a delta mass, we are in the $\tau$-arithmetic case with $\tau=\log 3$, which implies the limit statement
\[
\lim_{k\to\infty} f(k\log 3+y) = \sum_{j=-\infty}^\infty g(j\log 3+y) \quad \text{for all}~y\in [0, \log 3),
\]
which can be rewritten in terms of $r$ and $R$ as follows:
\begin{align}
\label{eq:renew}
\lim_{k\to\infty} (3^k e^y)^{-1/d_H} r(3^k e^y) =  \sum_{j=-\infty}^\infty (3^j e^y)^{-1/d_H} R(3^j e^y) \quad \text{for all}~y\in [0, \log 3).
\end{align}
To deduce \eqref{eq:asymprm} we will replace $3^k e^y$ by $x$.
This triggers a change of variables on the RHS of \eqref{eq:renew} and yields
\[
x^{-1/d_H} r(x) = \sum_{j=-\infty}^\infty \left(3^j \tilde{x}\right)^{-1/d_H} R\left(3^j \tilde{x}\right) + o(1) \quad\text{as } x\to\infty,
\]
where $\tilde{x}:=\tilde{x}(x) \in [1,3)$ is the unique number such that $\log \tilde{x} \equiv \log x \pmod{\log 3}$.
Note that since $R(x)=0$ for $x\leq 1$, the terms with negative values of $j$ do not contribute to the Fourier series. 
Hence we obtain \eqref{eq:asymprm} with
\begin{align}
\label{eq:G}
\mathcal{G}(\log x) = \mathcal{G}(\log\tilde{x})=\sum_{j=0}^\infty (3^j \tilde{x})^{-1/d_H}R(3^j \tilde{x}).
\end{align}

Since $R(x)$ has a finite number of jump discontinuities on $[0, 4\cdot 3^4)$, and it has a well-defined number of jumps on $[4\cdot 3^n, 4\cdot 3^{n+1})$, $n\geq 4$, at $a\cdot 3^n$, where $a\in \left\{4+2p, 4\frac{4}{9}+2p, 4\frac{2}{3}+2p, 5\frac{1}{3}+2p ~|~ p\in \{0,1,2,3\}\right\}$, it follows that $\mathcal{G}(\log x)$ has a finite number of jump discontinuities.

To obtain the identity \eqref{Gest}, we restrict to $x\in \left[\frac{10}{9}, \frac{4}{3}\right)$, in which case $\tilde{x}=x$.
By Corollary \ref{cor:constantrad} it follows that $R(3^j x) =0$ for all $j\geq 2$.
Meanwhile, $R(3^j x)=1$ when $j=1$, and $R(3^j x)=0$ when $j=0$. 
So $\mathcal{G}(\log x) = (3x)^{-1/d_H} = \frac{1}{2}x^{-1/d_H}$.

The proof of the global estimate \eqref{eq:globalest} is given in \S\ref{sec:globalest}.
\end{proof}

\begin{remark}
Computing the $\mathcal{G}$ function \eqref{eq:G} for $\tilde{x}\in [1,\frac{10}{9})$ requires knowing $R(3^j \tilde{x})$ for all $j$, which can be obtained using Theorem \ref{thm:radialcycle}, \emph{cf.\@} Figure \ref{fig:radialcycle}.
This is a relatively tedious exercise, so we opt for a more geometric approach in making the global estimate \eqref{eq:globalest}.
\end{remark}

\subsection{Geometric estimate of sandpile growth} \label{sec:globalest}

In this subsection we estimate the growth of the abelian sandpile cluster on $SG$, using only information about volumes of subsets of $SG$.
The techniques described below have been applied to other graphs \cites{BorgneRossin,ASMSGStr}, and appear to give close-to-optimal global lower bound on $SG$, \emph{cf.\@} Figure \ref{fig:ASMScaling}.
However these do not yield the necessary remainder estimate to produce the log-periodic oscillations shown in the previous subsection \S\ref{sec:osc}.

To be precise, we prove
\begin{proposition}
\label{prop:rmbounds}
For every $m\in \mathbb{N}$ we have
\[
\frac{(r_m)^{d_H}+1}{m} \geq \frac{2}{9}  \quad \text{and}
\quad \frac{(r_m-1)^{d_H}}{m} \leq \left(\frac{3}{4}\right)^{d_H}.
\]
\end{proposition}
The bounds \eqref{eq:globalest} follow from taking the limits $m\to\infty$ of the inequalities in this proposition.

Recall from Proposition \ref{prop:ball}, Part \eqref{item:ball} that $B_o(r_m-1) \subset A(m) \subset S(m)=B_o(r_m)$, where $A(m)$ and $S(m)$ are, respectively, the firing set and the receiving set corresponding to $(m\mathbbm{1}_o)^\circ$.

\begin{proof}[Proof of Proposition \ref{prop:rmbounds}, lower bound]
Since $S(m)=B_o(r_m)$, and the maximal stable configuration has $3$ chips everywhere (except at $o$ where a max of 1 chip is allowed, but WLOG we may assume that $m$ is even, so $o$ carries 0 chip), we have
\[
\frac{m}{3} \leq \#\left(\text{vertices in $B_o(r_m)$ occupied with a chip}\right) \leq |B_o(r_m)|.
\]
So it suffices to give a good upper bound for $|B_o(r_m)|$. 

Suppose $r_m =(1-\epsilon)2^k$ for some $k\in \mathbb{N}$ and $2^{-p} \leq \epsilon < 2^{-(p-1)}$, $p\in \{2,\cdots, k\}$. Observe that $B_o(r_m) \subset B_o((1-2^{-p})2^k)$, and that $B_o(2^k) \setminus B_o((1-2^{-p})2^k)$ is the union of $2^p$ copies of the graph $G_{k-p}$ excluding the head vertex. Since $|V(G_k)| = \frac{3}{2}(3^k+1)$, we deduce that
\begin{align*}
|B_o(r_m)| &\leq |B_o((1-2^{-p})2^k)| = \frac{3}{2}(3^k+1) - 2^p \left[\frac{3}{2} \left(3^{k-p}+1\right)-1\right]\\
&< \frac{3}{2} \left[3^k\left(1-\left(\frac{2}{3}\right)^{p} \right)+1\right] < \frac{3}{2} \left[(r_m)^{d_H} \frac{1-\epsilon^{d_H-1}}{(1-\epsilon)^{d_H}} +1\right].
\end{align*}
Since
\[
\sup_{\epsilon\in(0,\frac{1}{2})} \frac{1-\epsilon^{d_H-1}}{(1-\epsilon)^{d_H}} = 1,
\]
we conclude that 
\[
m\leq  3 |B_o(r_m)| < \frac{9}{2}((r_m)^{d_H}+1).
\]
\end{proof}

To prove the upper bound we invoke a result of Rossin \cite{Rossin}.

\begin{lemma}[\cite{Rossin}*{Lemme 15}]
\label{lem:Rossin}
Let $X$ be a connected subgraph of $G$. The minimum number of grains needed for every vertex of $X$ to topple at least once is equal to $|{\rm in}(X)| + |C_G(X)|$, where ${\rm in}(X)$ is the set of (internal) edges in $X$, $C_G(X) = \{(i,j) : i\in X,~j\in V\setminus X\}$, and $|\cdot|$ denotes the cardinality of the edge set.
\end{lemma}

\begin{proof}[Proof of Proposition \ref{prop:rmbounds}, upper bound]
Since $S(m)=B_o(r_m) \supset A(m)$, we can bound the mass needed to fill $B_o(r_m)$ from below by the mass needed to topple everywhere in $A(m)$. 

Suppose $r_m=(1-\epsilon)2^k+1$ for some $k\in \mathbb{N}$ and $2^{-(p+1)} \leq \epsilon < 2^{-p}$, $p\in \{1,2,\cdots, k\}$. Since $A(m) \supset B_o(r_m-1) \supset B_o((1-2^{-p})2^k)$, according to Lemma \ref{lem:Rossin}, we have
\begin{align}
\label{ineq:mr}
m  \geq |{\rm in}(A(m)) |+ |C_G(A(m)) | \geq |{\rm in}(B_o(r_m-1))| \geq |{\rm in}(B_o((1-2^{-p})2^k))|.
\end{align}
It is direct to check that $B_o(2^k) \setminus B_o((1-2^{-p})2^k))$ is the union of $2^p$ copies of the graph $G_{k-p}$. Since $G_k$ has $3^{k+1}$ edges, we deduce that the RHS of \eqref{ineq:mr} equals
\[
3^{k+1}-2^p 3^{k-p+1}  = 3 \left(\frac{r_m-1}{1-\epsilon}\right)^{d_H} \left(1-\left(\frac{2}{3}\right)^p\right) \geq 3 (r_m-1)^{d_H} \frac{1-(2/3)^p}{(1-2^{-(p+1)})^{d_H}}.
\]
Therefore
\begin{align}
\label{eq:rmupbd2}
\frac{(r_m-1)^{d_H}}{m} \leq \frac{1}{3} \frac{(1-2^{-(p+1)})^{d_H}}{1-(2/3)^p}.
\end{align}
It can be checked explicitly that the function $\displaystyle x\mapsto \frac{(1-2^{-(x+1)})^{d_H}}{1-(2/3)^x}$ is decreasing on $(0,\infty)$, and tends to $1$ as $x\to\infty$. For a uniform estimate we can take $p=1$ in \eqref{eq:rmupbd2} to obtain the claimed upper bound.
\end{proof}


\section{Fluctuations of the IDLA cluster on $SG$} \label{sec:IDLA}

\begin{figure}
\centering
\includegraphics[width=0.45\textwidth]{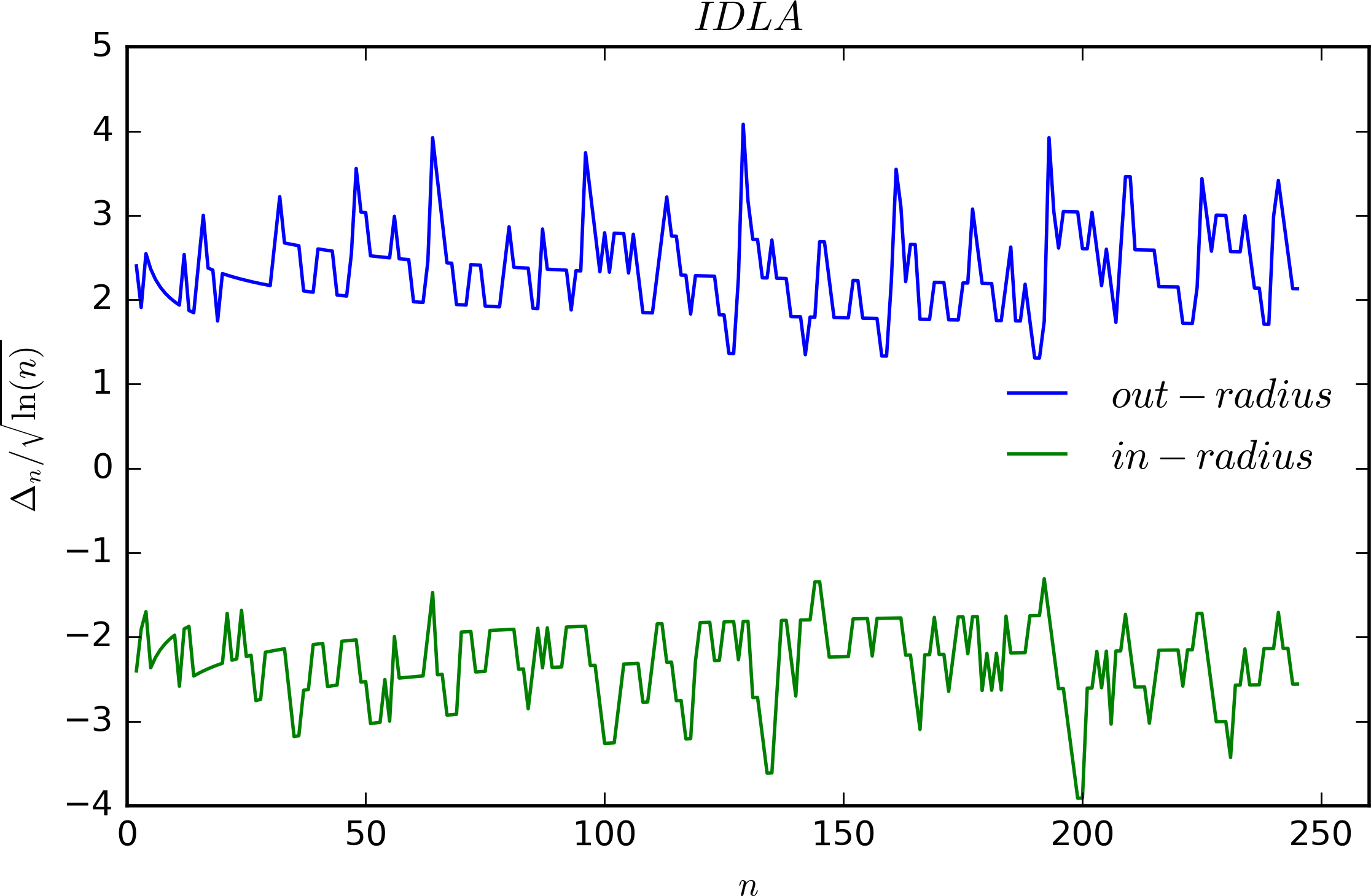}
\caption{Maximal out-radius and minimal in-radius vs.\@ the expcted radius for 1000 realizations of the IDLA on $SG$.}
\label{fig:IDLA}
\end{figure}

In this section we present numerical results concerning fluctuations of the IDLA cluster. Proposition \ref{prop:IDLAShapeThm} implies that the limit shape is a ball in the graph metric, without quantifying the order of the fluctuations about the limit shape. In the case of $\mathbb{Z}^d$, it is proved that the fluctuations are of order $\sqrt{\log n}$ when $d=2$, and of order $\log n$ when $d\geq 3$, \emph{cf.\@} \cites{AG1, AG2, JLS1}. 
The second-named author has written the Python program ``AutomataSG'' \cite{AutomataSG} and performed simulations of IDLA on $SG$, which strongly suggest the following

\begin{conjecture}
There exists $C>0$ such that
\[
B_o(n-C\sqrt{\log n}) \subset \mathcal{I}(|B_o(n)|) \subset B_o(n+C\sqrt{\log n})
\]
for all $n$, with probability $1$.
\end{conjecture}

To prove the $\sqrt{\log n}$ fluctuation one needs to show that growing a tentacle of length $\ell$ has a probability of order $\exp(-c\ell^2)$, which is the case on $\mathbb{Z}^d$, $d\geq 3$ \cites{AG1, AG2, JLS1}, and on the comb lattice \cite{AR16}.
On $SG$ we expect the proof techniques to go beyond what were used in \cite{IDLASG}, and involves careful potential theoretic analysis on $SG$.

Another salient feature we see numerically are the log-periodic oscillations in the rescaled radial fluctuations (by $\sqrt{\log n}$). These are present in both the sample maximum (Figure \ref{fig:IDLA}) and in the sample mean, rescaled by the expected radius $\sqrt{\log n}$, with the latter being more pronounced. 
We also simulated rotor-router aggregation where the rotor mechanisms are identical and fixed for all vertices, but the initial rotor configuration is randomized by choosing, independently for every vertex, each of the $4$ (or $2$ at $o$) possible rotor directions with equal probability.
The resulting \emph{unscaled} radial fluctuations also exhibits log-periodic oscillations (Figure \ref{fig:avgstat}).

\begin{openquestion}
Prove the existence of log-periodic oscillations shown in Figures \ref{fig:IDLA} and \ref{fig:avgstat}.
Even better, characterize the almost-sure properties of the log-periodic oscillation (\emph{e.g.\@} it contains a dense set of jump discontinuities).
\end{openquestion}

\begin{figure}
\centering
\includegraphics[width=0.45\textwidth]{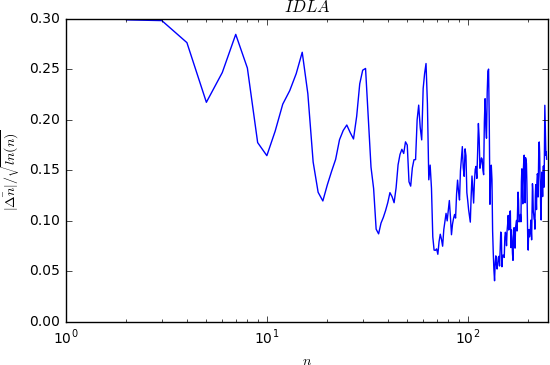}
\includegraphics[width=0.47\textwidth]{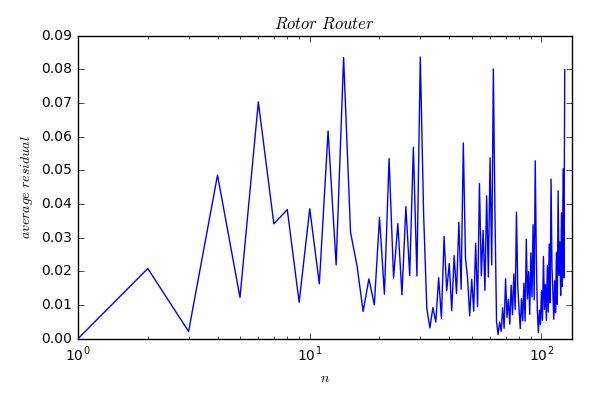}
\caption{Sample average of the absolute value of the radial fluctuations about the expected radius, for 1000 realizations of the IDLA (left) and of rotor-router aggregation (right) on $SG$.}
\label{fig:avgstat}
\end{figure}

We also expect a corresponding central limit theorem for the rescaled space-time fluctuations, though here we only investigate it numerically.
Following \cite{JLS2} we Poissonize the IDLA process.
Let $\{N(t): t\geq 0\}$ be a rate-$1$ simple Poisson process, and consider the continuous-time IDLA process $\{\mathcal{I}(N(t)): t\geq 0\}$ started from $o$.
We pick a radius $r=\epsilon 2^k$ for $k\in\mathbb{N}$ and $\epsilon \in 1-2^{-p}$, $p\in \mathbb{N}$; observe that the sphere $S_o(r)=\{x\in SG: d(o,x)=r\}$ is an interval when $p=1$, and is a union of disjoint intervals which is a pre-fractal approximation of the Cantor set when $p\geq 2$.
By Proposition \ref{prop:IDLAShapeThm}, $\mathcal{I}(N(B_o(r)))$ is close to $B_o(r)$, so we consider radial fluctuations of the former about the sphere $S_o(r)$ by introducing the function
$h: S_o(r) \to \mathbb{Z}$, $h(x) = A(x)-r$, where $A(x)$ is the height of the ``tentacle'' in the cluster $\mathcal{I}(N(B_o(r)))$ measured transversally from the sphere $S_o(r)$.
(Our function $h$ is related to the ``lateness function'' introduced by \cite{JLS2}.)
By the FKG inequality \cite{FKG}, the covariance of $h(x)$ and $h(y)$ is always nonnegative.
In Figure \ref{fig:IDLAcov} we present simulations of the covariance of the rescaled $h$ function (by $\sqrt{\log n}$).
From the data it appears that the covariance is higher when $x$ and $y$ are taken to lie in the same connected component of $SG\cap S_o(r)$, and furthermore, local maxima of the covariance occur when $x$ and $y$ are both cut points.
This may be explained by the fact that tentacles are rooted from the cut points.

\begin{figure}
\centering
\bgroup
\def\arraystretch{1.5}%
\begin{tabular}{cccc}
& $\epsilon=\frac{3}{4}$\hspace{0.2in} & $\epsilon=\frac{7}{8}$\hspace{0.2in} & $\epsilon=1$\hspace{0.2in} \\ 
\rotatebox{90}{\hspace{0.5in} $k=5$} &
\includegraphics[width=1.6in]{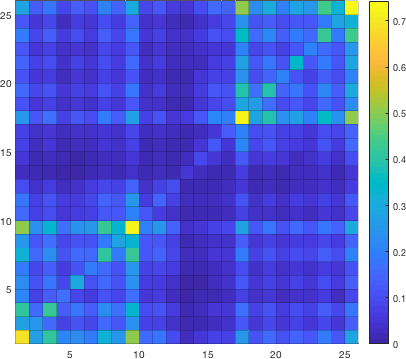} &
\includegraphics[width=1.6in]{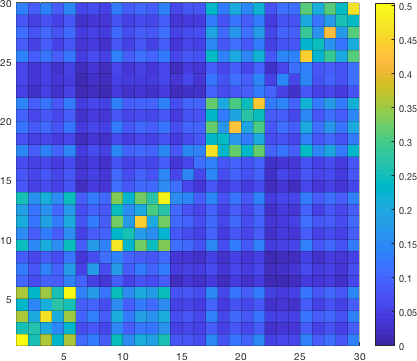} &
\includegraphics[width=1.6in]{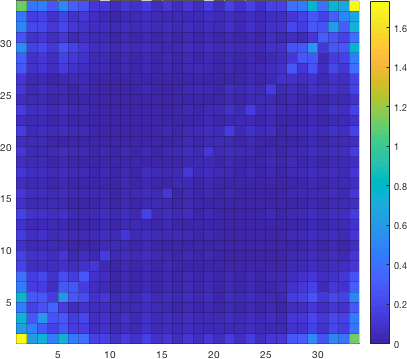} \\
\rotatebox{90}{\hspace{0.5in} $k=6$} &
\includegraphics[width=1.6in]{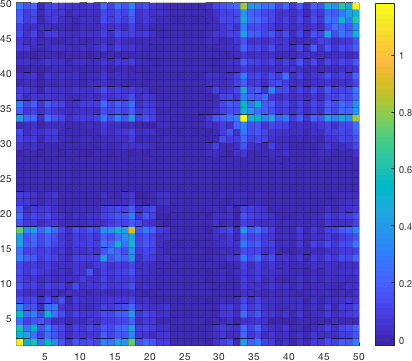} &
\includegraphics[width=1.6in]{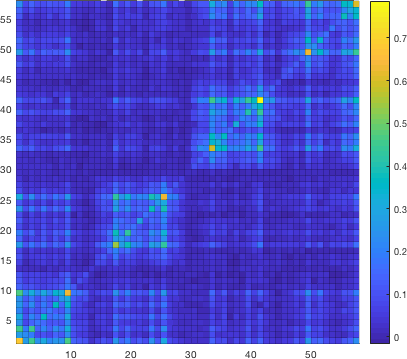} &
\includegraphics[width=1.6in]{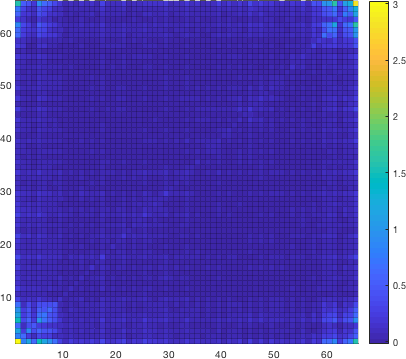} \\
\rotatebox{90}{\hspace{0.5in} $k=7$} &
\includegraphics[width=1.6in]{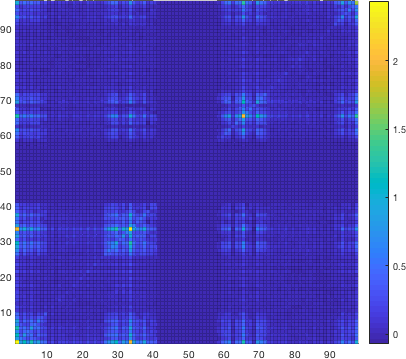}&
\includegraphics[width=1.6in]{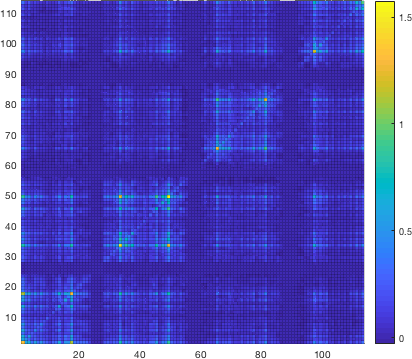}&
\includegraphics[width=1.6in]{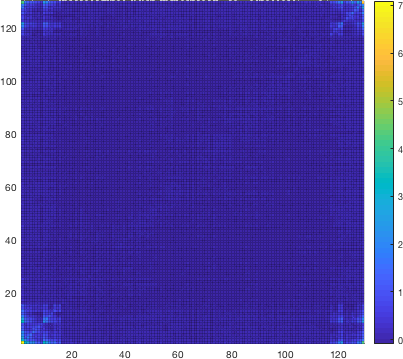}
\end{tabular}
\egroup
\caption{Simulations showing the covariance of the continuous-time IDLA height fluctutations ${\rm Cov}\left(\frac{h(x)}{\sqrt{\log r}},\frac{h(y)}{\sqrt{\log r}}\right)$ for $x,y \in S_o(r)$, where $r=\epsilon 2^k$ with $\epsilon \in \left\{\frac{3}{4}, \frac{7}{8}, 1\right\}$ and $k\in \{5,6,7\}$. Each plot is generated with 1000 simulations.}
\label{fig:IDLAcov}
\end{figure}


\section{Open questions} \label{sec:open}

We close this paper with some outstanding questions and future directions.

\emph{Extensions to nested fractals.}
In light of Theorem \ref{thm:shapeuniv}, it is natural to ask if there are other examples where the limit shapes of the four Laplacian growth models coincide. Some potential candidates are the Vicsek tree (Figure \ref{fig:VT}), the planar Sierpinski gaskets $SG(m)$ (Figure \ref{fig:SG3}, whereupon rotor walks are studied in \cite{FTRotor}), and higher-dimensional Sierpinski simplices (Figure \ref{fig:ST}), all of which are \textbf{nested fractals} as defined by Lindstr\o m \cite{Lindstrom}. 

\begin{figure}
\centering
\begin{subfigure}[t]{0.2\textwidth}
\centering
\includegraphics[width=\linewidth]{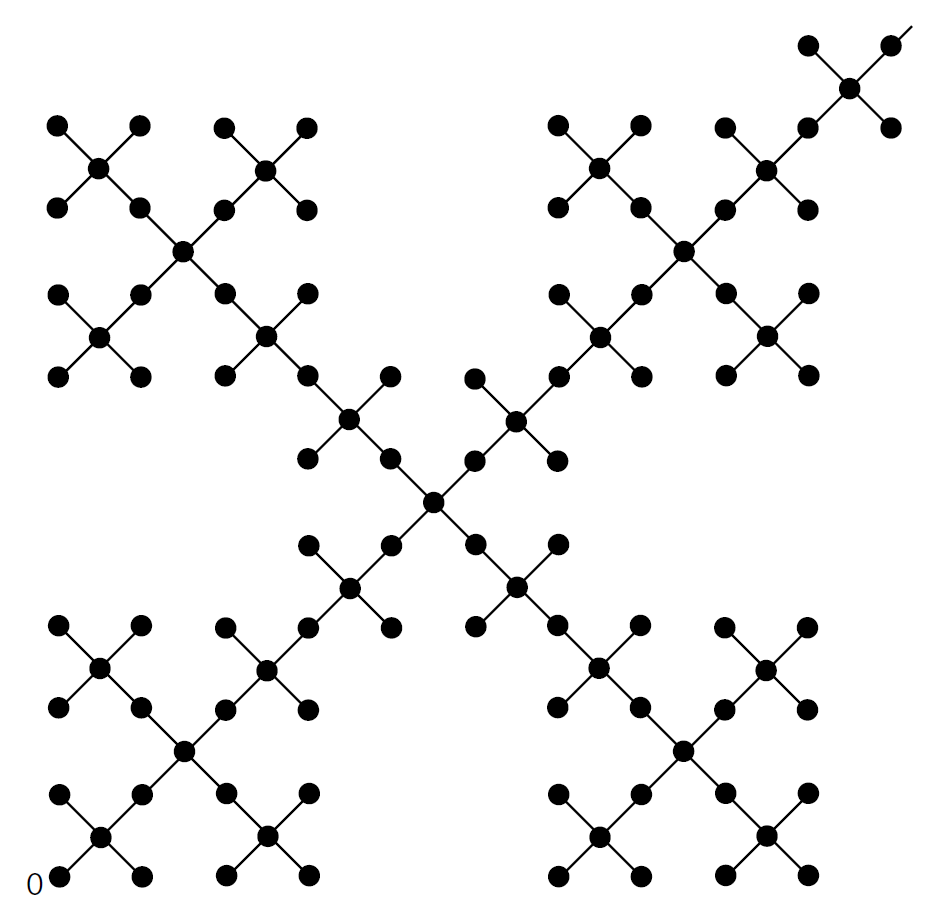}
\caption{Vicsek tree}
\label{fig:VT}
\end{subfigure}
\begin{subfigure}[t]{0.22\textwidth}
\centering
\includegraphics[width=\linewidth]{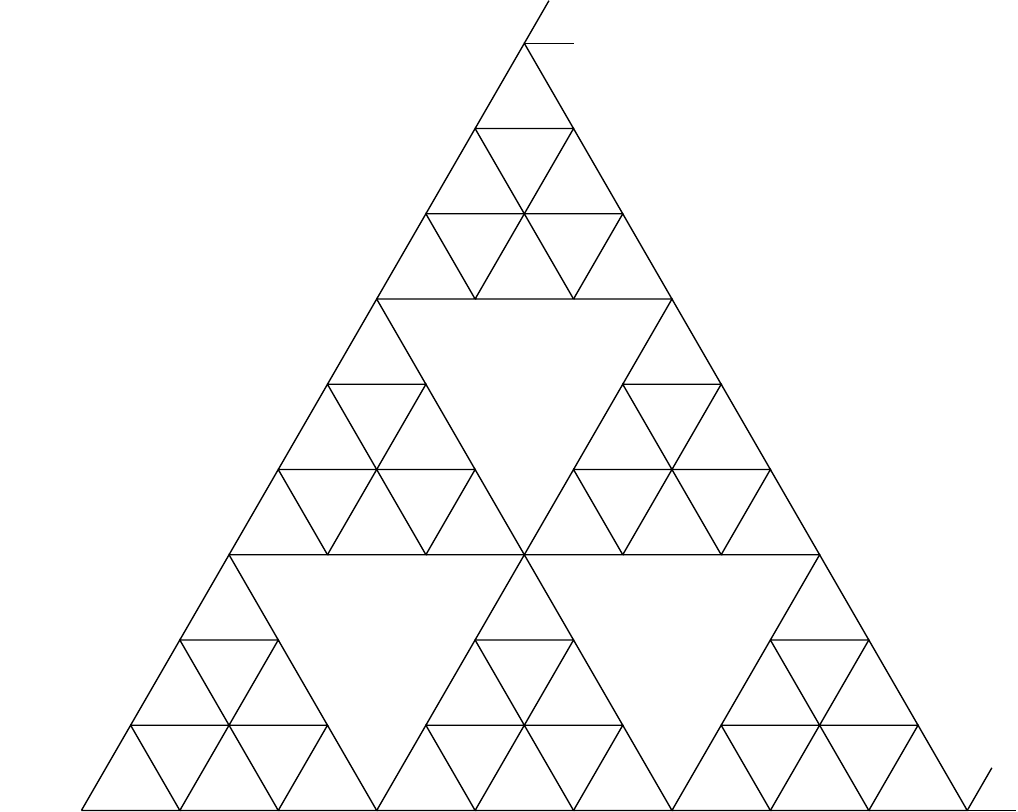}
\caption{$SG(3)$}
\label{fig:SG3}
\end{subfigure}
\begin{subfigure}[t]{0.25\textwidth}
\centering
\includegraphics[width=0.8\linewidth]{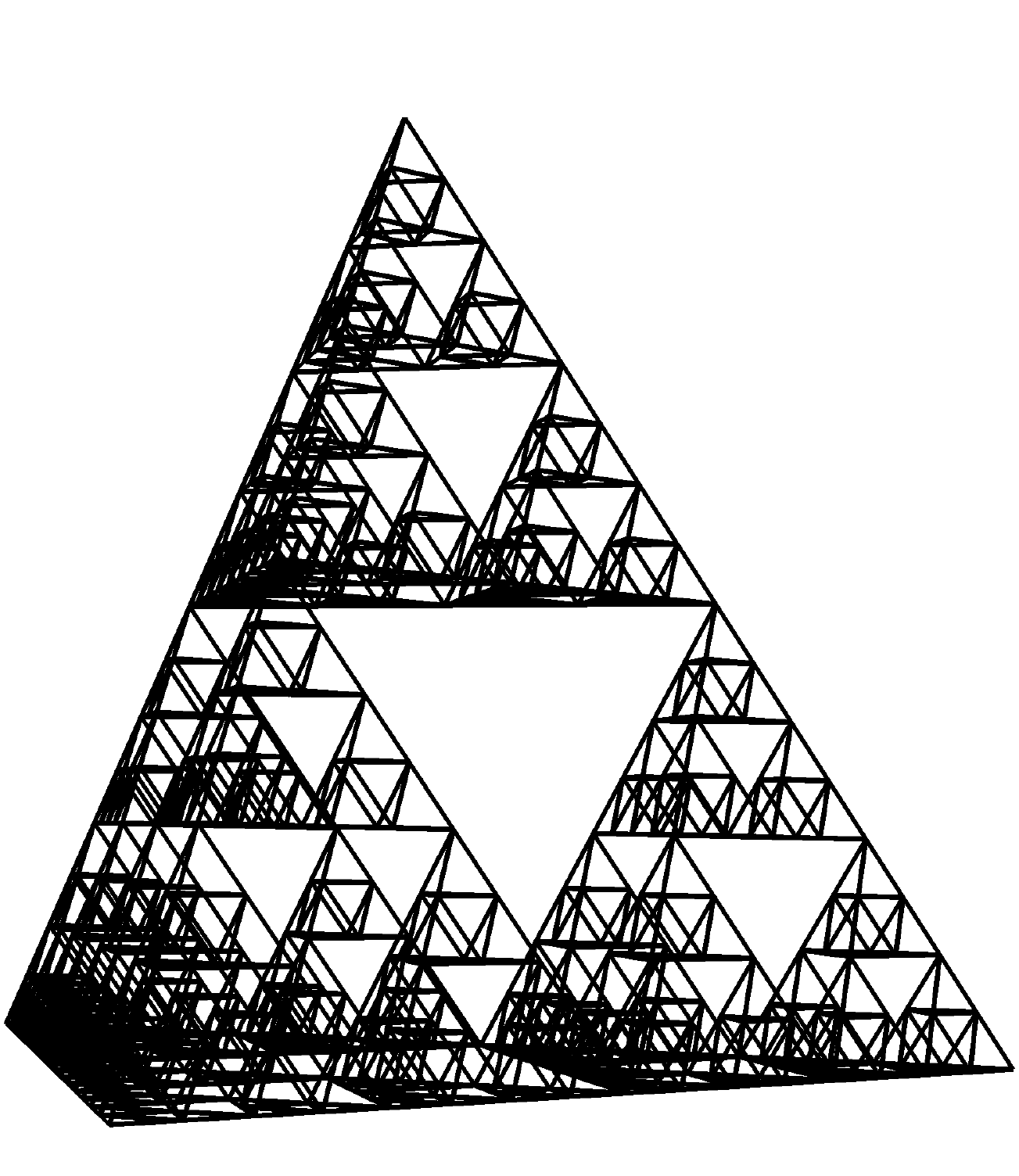}
\caption{Sierpinski tetrahedron}
\label{fig:ST}
\end{subfigure}
\begin{subfigure}[t]{0.25\textwidth}
\centering
\includegraphics[width=0.8\linewidth]{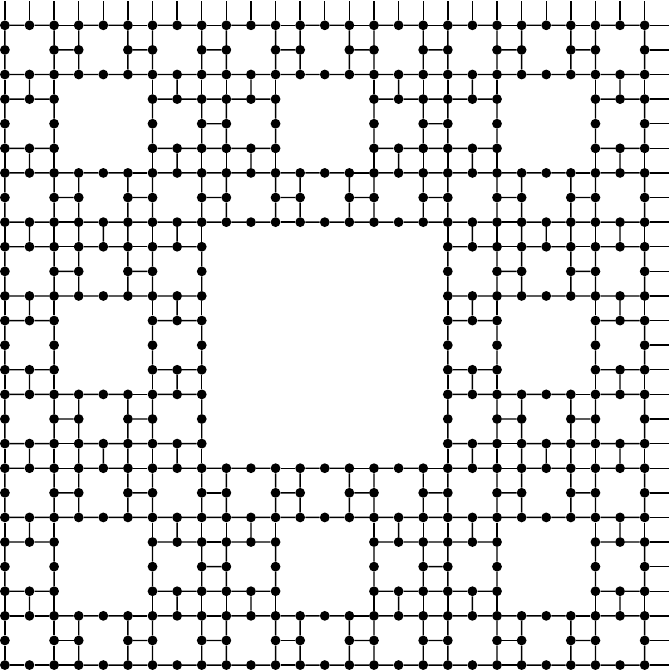}
\caption{Sierpinski carpet graph}
\label{fig:SC}
\end{subfigure}
\caption{A sample of self-similar fractal graphs.}
\label{fig:others}
\end{figure}

\begin{conjecture}
\label{conj:LSUNested}
Limit shape universality of single-source Laplacian growth models holds on nested fractal graphs, if the starting vertex $o$ is chosen such that all balls $B_o(n)$ in the graph metric have spatial symmetry.
\end{conjecture}

\begin{openquestion}
To what extent can one establish log-periodic oscillations in Laplacian growth models on nested fractal graphs?
\end{openquestion}

Recent simulations by Ilias Stitou suggest that Conjecture \ref{conj:LSUNested} holds on the Vicsek tree and on $SG(3)$.
In particular, log-periodic oscillations are present in the sandpile growth.

From a rigorous perspective, we expect Conjecture \ref{conj:LSUNested} to hold amongst divisible sandpiles, IDLA, and rotor-router aggregation at the least. 
Since the harmonic measure on the sphere (in the graph metric) is approximately uniform,
this should imply a spherical shape theorem for the divisible sandpiles.
Consequently, we can use the same strategy outlined in \S\ref{sec:harmonicmeasure} to obtain the same limit shapes for the IDLA or rotor-router cluster, respectively.

The more delicate question concerns the abelian sandpile model.
For tree nested fractal graphs, this can be addressed using results from \cite{LevineTree}.
For non-tree nested fractal graphs, the problem appears to be open to the best of our knowledge.
(See also \cite{ASMSGStr} for numerical results of the sandpile model on various gasket-type fractal graphs.)

\emph{Explosions in sandpile growth.}
An unexpected feature of the sandpile growth on $SG$ is the presence of radial explosions, which do not appear on $\mathbb{Z}^d$ or trees.

\begin{openquestion}
Does radial explosion of sandpile growth appear on other non-tree nested fractal graphs? 
\end{openquestion}

It would seem that the explosion comes from a combination of the loop and the cut point structures of the fractal graph. 
However, Stitou's simulations show there is \emph{no} explosion on $SG(3)$ (Figure \ref{fig:SG3}), despite it having the stated properties. 
It thus appears that $SG$ may be an exception rather than the rule, even among this class of fractal graphs. 
A careful study of the corresponding sandpile group may provide clues to resolving this question.

\begin{openquestion}
Study Laplacian growth models and the single-source abelian sandpile model on the Sierpinski tetrahedron (Figure \ref{fig:ST}). In particular, what does the cluster (and in the case of the abelian sandpile model, the sandpile configuration) look like when restricted to a cross-section of the tetrahedron that is isomorphic to the two-dimensional gasket?    
\end{openquestion}

\emph{Scaling limits of the sandpile patterns on $SG$.}
The appearance of self-similar sandpile tiles (Figure \ref{fig:tailpattern}) suggests that there is a scaling limit of the abelian sandpile patterns restricted to the ``bulk.'' 
A plausible limit statement is as follows: For a fixed $k \in [4,12)$, the sequence of configurations $\left\{2^{-n}\left.(k\cdot 3^n \mathbbm{1}_o)^\circ \right|_{G_n^{(s)}}\right\}_{n=1}^\infty$ converges to a height function on the limit fractal $K$.
Furthermore, we conjecture that this convergence takes place in the weak-$*$ $L^\infty(K,\nu)$ topology ($\nu$ is the standard self-similar measure on $K$) along the full sequence indexed by mass $m\to\infty$.
Roughly speaking this says that local averages of the pattern colors in the bulk converges to a function on $K$.
The conjecture is inspired by the elliptic PDE arguments of Pegden and Smart \cite{PegdenSmart} showing the sandpile pattern convergence on $\mathbb{Z}^d$.

\emph{Other single-source and multi-source growth models on $SG$.}
Throughout this paper we assumed that particles are launched from a single corner vertex $o$ of $SG$. One may ask what happens if they are launched from a fixed vertex which is not $o$, or from multiple vertices. (In the case of $\mathbb{Z}^d$ see \cite{LevinePeres10}.) This is very much an analytic problem as it is probabilistic and combinatorial, requiring fine analysis of the Dirichlet (or obstacle) problem on subsets of $SG$. 
We expect that recent results of Qiu, Strichartz, and collaborators \cites{QiuStr,GKQS} may be applicable for this purpose. 

Another interesting initial condition to explore is the product Bernoulli $\{4-\delta, 4+\delta\}$ cluster: that is, start with a nonempty subset of $SG$ (say, $B_o(r)$ for some $r\in \mathbb{N}$) wherein every vertex independently carries either $4-\delta$ or $4+\delta$ chips with probability $\frac{1}{2}$.
According to simulations by Ahmed Bou-Rabee (Figure \ref{fig:SGBernoulli}), it is conjectured that the sandpile cluster exhibit a two-phase pattern structure, where the patterns in the inner ball (the support of the initial cluster) becomes noisy, while the outer annulus carries patterns reminiscent of those in the deterministic single-source sandpile.

\begin{figure}
\centering
\includegraphics[width=0.9\textwidth]{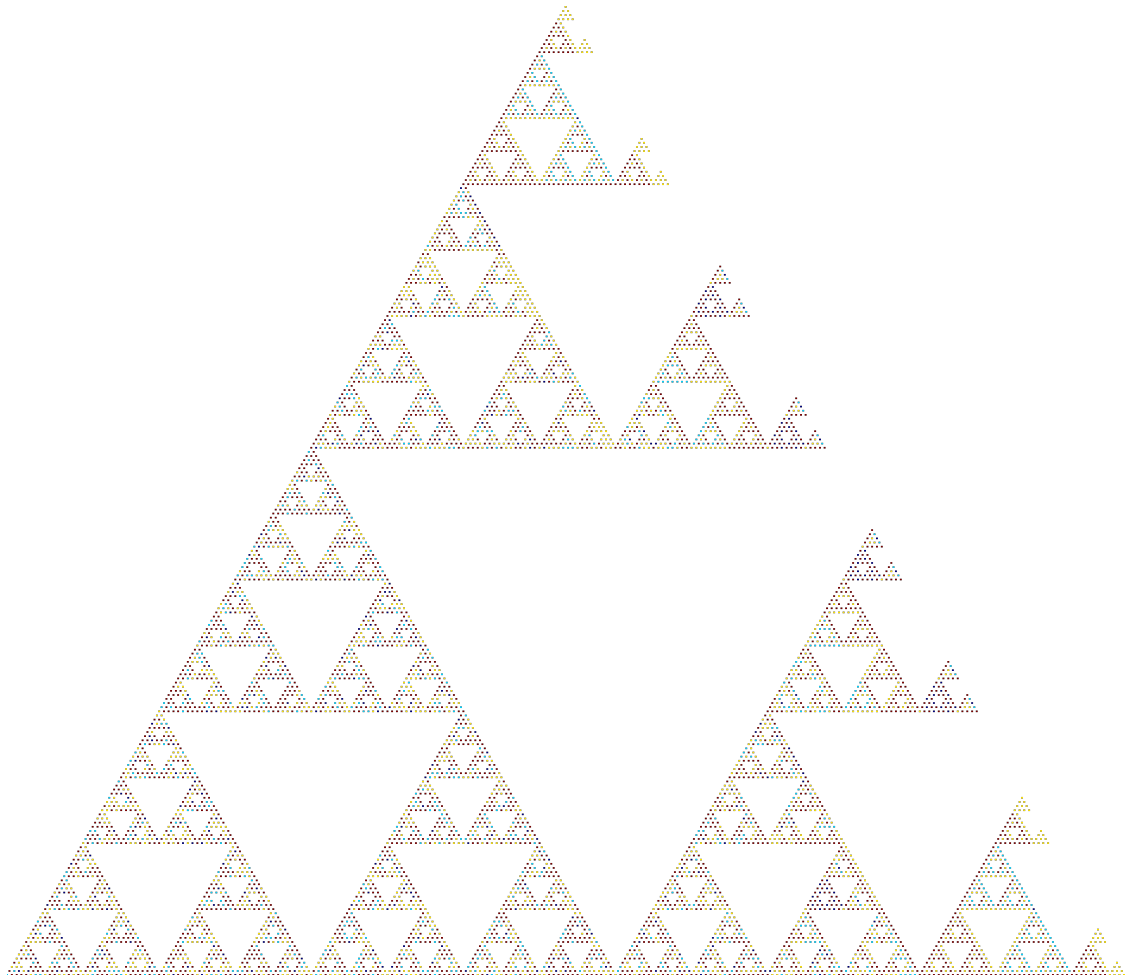}
\caption{A realization of the abelian sandpile configuration starting from the product Bernoulli 3-5 initial configuration on $B_o(160)$. 
The stable configuration exhibits a two-phase pattern structure: noisy patterns on the inner ball (the support of the initial configuration), and patterns reminiscent of the deterministic single-source sandpile on the outer annulus.
Image courtesy of Ahmed Bou-Rabee.}
\label{fig:SGBernoulli}
\end{figure}

\emph{Sandpile Markov chain on fractal graphs.}
In this paper we solved the deterministic abelian sandpile growth problem from a single source.
A related, but different, problem is to study the Markov chain on $\mathcal{R}_G$ under stationarity.
Recalling that $(\mathcal{R}_G, \oplus)$ is a finite abelian group, it is a standard fact that the stationary distribution is a uniform on $\mathcal{R}_G$.
Besides the works of Daerden \emph{et al.\@} \cites{DV98, DPV01}  and Matter \cite{MatterThesis}*{Chapter 5} mentioned in Remark \ref{rem:sandpileSG}, it would be interesting to exploit the burning bijection and make connections with recent results on spanning trees and Laplacian determinants on fractal graphs \cites{CCY, CC, AnemaTsougkas, CTT, TeuflWagner}.
A yet unsolved problem is to compute the average number of chips per vertex, or the sandpile density, on subgraphs of $SG$ under stationarity.

\emph{Laplacian growth models on the Sierpinski carpet.}
One may also perform the same analysis on the Sierpinski carpet graph $SC$ (see Figure \ref{fig:SC}), an infinitely ramified fractal. For concreteness we discuss only the case where $m$ particles are launched from the bottom-left corner vertex of $SC$. Wilfried Huss has performed simulations on both the rotor-router aggregation and the IDLA, \emph{cf.\@} \cite{HussThesis}*{Figures 8.2 and 8.3}. (Both sandpile problems on $SC$ are open.) His key observations were that given a fixed number of particles, the cluster shapes of the two models look qualitatively similar. However, as $m$ increases, a periodic family of limit shapes, as opposed to a single limit shape, appears to develop in both models. Due to the difficulty in the analysis on the Sierpinski carpet (see \cite{BarlowSCReview} for a summary of the state of the art, and references therein for details), we are unable to address Huss' observations rigorously at the moment. 

\emph{Do limit shapes coincide on different graphs approximating the same space?}
We end the paper with the following semi-vague question.

\begin{openquestion}
Let $(G^{(1)}_n)_n$ and $(G^{(2)}_n)_n$ be two sequences of bounded-degree unweighted graphs rooted at a common point $o$. Assume that both sequences converge, in the pointed Gromov-Hausdorff sense (or is there a better mode of convergence?), to the same metric measure space $(X, d, \mu)$ rooted at $o$.
Prove or disprove (with possibly extra conditions) that the clusters associated with any of the Laplacian growth models on $\bigcup_{n\geq 1}G^{(1)}_n$ and $\bigcup_{n\geq 1} G^{(2)}_n$ started from $o$ have the same limit shape.
\end{openquestion}

As an example, take $G^{(1)}_n = (n^{-1}\mathbb{Z})^2$ and $G^{(2)}_n$ to be the hexagonal lattice rooted at $o$ with lattice spacing $1/n$.
Another example is to take $G^{(1)}_n$ to be the level-$n$ Sierpinski gasket graph, and $G^{(2)}_n$ to be the level-$n$ Hanoi tower graph (Figure \ref{fig:Hanoi}). Both sequences of graphs are rooted at the corner vertex $o$ and scaled such that that their diameters stay constant.

It will be useful to investigate this problem for the divisible sandpile model first.

\subsection*{Acknowledgements} 

JPC would like to thank Lionel Levine for providing many useful suggestions and feedbacks; Ilias Stitou for running computations which helped clarify the mechanisms described in the paper; Ahmed Bou-Rabee for useful discussions concerning the open questions in \S\ref{sec:open}; and Elmar Teufl for bringing to his attention the history behind the Sierpinski arrowhead curve.
He also acknowledges inspiring conversations with Richard Kenyon, Tatiana Nagnibeda, Ecaterina Sava-Huss, Robert Strichartz, Alexander Teplyaev, and Wolfgang Woess.
JKF would like to thank Patrick Crotty for his help in setting up simulations on the Ho Computer Cluster at Colgate University.

\end{document}

%% file: sandpiletable.tex
\begin{tabular}{lrrrr}
$\frac{m}{3^n}$ & $m$ & $m'$ & $m-2m'$ &  $\Delta r_m$\\ \hline
& 2&1 & 0 & 1\\
& 8&4  &0 & 1\\ \hline \hline
\multicolumn{5}{c}{${\bf n=1}$} \\ \hline \hline
${\bf 4}$&  {\bf 12} & {\bf 3}&  {\bf 6} &\\
$4\frac{2}{3}$ & 14& 4& 6& 1\\
{\color{red} $6$}& {\color{red} 18} & {\color{red}6} & {\color{red} 6} & \\
$6\frac{2}{3}$ & 20 & 7 & 6& \\
{\color{teal}$8$}& {\color{teal} 24} & {\color{teal}9}& {\color{teal} 6}& \\
$8\frac{2}{3}$ & 26& 10 & 6 & 1\\
{\color{blue} $10$} & {\color{blue} 30} & {\color{blue}12} &{\color{blue} 6}& \\
$10\frac{2}{3}$ & 32 & 13 & 6 & \\ \hline \hline
\multicolumn{5}{c}{${\bf n=2}$} \\ \hline \hline
${\bf 4}$ & \textbf{36} &{\bf 6}& {\bf 24} & {\bf 1}\\
$4\frac{2}{9}$ & 38 & 7 & 24 &  \\
$4\frac{2}{3}$ &42 & 8 & 26 &\\
$5\frac{1}{3}$ &48& 10& 28 &1\\
{\color{red} $6$} &{\color{red}54} &{\color{red} 15}&  {\color{red} 24}&\\
$6\frac{2}{9}$ &56& 16 & 24& 1\\
$6\frac{2}{3}$ & 60 & 17 &26 & \\
$7\frac{1}{3}$ & 66 & 19 & 28 & \\
{\color{teal} $8$} & {\color{teal}72} & {\color{teal} 24}& {\color{teal} 24}&\\
$8\frac{2}{9}$ & 74 & 25 & 24& \\
$8\frac{2}{3}$ & 78 & 26 & 26 & \\
$9\frac{1}{3}$ &84&28& 28& 1\\
{\color{blue} $10$} & {\color{blue} 90} &{\color{blue} 33}& {\color{blue} 24} & \\
$10\frac{2}{9}$ & 92 & 34 & 24& \\
$10\frac{2}{3}$ & 96 & 35 & 26& \\
$11\frac{1}{3}$ & 102 & 37 & 28& \\ \hline \hline
\multicolumn{5}{c}{${\bf n=3}$} \\ \hline \hline
${\bf 4}$ &\textbf{108}& {\bf 15}& {\bf 78} &\textbf{2}\\
$4\frac{2}{27}$ &110& 16& 78&  1\\
$4\frac{4}{9}$ & 120 & 19&82 & \\
$4\frac{2}{3}$ &126 & 20 &86 & \\
$5\frac{1}{3}$ &144& 28& 88& 1\\
{\color{red} $6$} &{\color{red}162}&{\color{red}42}& {\color{red} 78}& {\color{red}1}\\
$6\frac{2}{27}$& 164& 43& 78 & \\
$6\frac{4}{9}$ & 174 & 46& 82 & \\
$6\frac{2}{3}$ & 180 & 47 & 86 & \\
$7\frac{1}{3}$ &198& 55& 88 & 1\\
\end{tabular} \hskip 5pt
\begin{tabular}{lrrrr}
$\frac{m}{3^n}$ & $m$ & $m'$ & $m-2m'$ & $\Delta r_m$\\ \hline
{\color{teal} $8$} &{\color{teal}216}&{\color{teal}69}&{\color{teal} 78} & {\color{teal}1}\\
$8\frac{2}{27}$ & 218 & 70 & 78 & \\
$8\frac{4}{9}$  & 228 & 73 & 82 &\\
$8\frac{2}{3}$ & 234 & 74& 86 &\\
$9\frac{1}{3}$ & 252 & 82& 88 & \\
{\color{blue} $10$} &{\color{blue}270}&{\color{blue}96}& {\color{blue} 78}& {\color{blue}1}\\
$10\frac{2}{27}$ & 272& 97 & 78 & \\
$10\frac{4}{9}$ & 282 & 100& 82 &\\
$10\frac{2}{3}$ & 288 & 101 & 86 & \\
$11\frac{1}{3}$ & 306 & 109&88& \\ \hline \hline
\multicolumn{5}{c}{${\bf n=4}$} \\ \hline \hline
${\bf 4}$ &\textbf{324}& {\bf 42} & {\bf 240} & \textbf{5} \\
$4\frac{2}{81}$ & 326 & 43& 240 & \\
$4\frac{4}{9}$ &360& 55 & 250 &  1\\
$4\frac{2}{3}$ & 378 & 56 & 266 & \\
$5\frac{1}{3}$ &432& 82 & 268 &1\\
{\color{red} $6$}&{\color{red}486}& {\color{red}123}& {\color{red} 240} & {\color{red}4}\\
$6\frac{2}{81}$ & 488 & 124 & 240 & \\
$6\frac{4}{9}$ &522 & 136 &250& \\
$6\frac{2}{3}$ & 540 & 137 & 266 & \\
$7\frac{1}{3}$ &594& 163 &268 &1\\
{\color{teal} $8$} &{\color{teal}648}& {\color{teal} 204} & {\color{teal} 240}& {\color{teal}2}\\
$8\frac{2}{81}$ & 649 & 205 & 240 & \\
$8\frac{4}{9}$ & 684 & 217 & 250 &\\
$8\frac{2}{3}$ &702& 218 &266  &1\\
$9\frac{1}{3}$ & 756 & 244 & 268 & \\
{\color{blue} $10$} &{\color{blue}810}& {\color{blue} 285}& {\color{blue} 240} &{\color{blue}1}\\
$10\frac{2}{81}$ & 812 & 286 & 240 & \\
$10\frac{4}{9}$ & 846 & 298 & 250& \\
$10\frac{2}{3}$ & 864 & 299 & 266& \\
$11\frac{1}{3}$ & 918 & 325&268& \\ \hline \hline
\multicolumn{5}{c}{${\bf n=5}$} \\ \hline \hline
${\bf 4}$ &\textbf{972}& {\bf 123} & {\bf 726} & \textbf{11}\\
$4\frac{2}{3^5}$ & 974 &124&726 & \\
$4\frac{4}{9}$ &1080& 163 & 754 & 1\\
$4\frac{2}{3}$ &1134& 164 & 806& 1\\
$5\frac{1}{3}$ &1296& 244& 808 & 2\\
{\color{red} $6$}&{\color{red}1458}& {\color{red} 366}  & {\color{red}726} & {\color{red}7}\\
$6\frac{2}{3^5}$ &1460 & 367 & 726 &\\
$6\frac{4}{9}$ & 1566 & 406 & 754 &
\end{tabular} \hskip 5pt
\begin{tabular}{lrrrr}
$\frac{m}{3^n}$ & $m$ & $m'$ & $m-2m'$ & $\Delta r_m$\\ \hline
$6\frac{2}{3}$ & 1620 & 407 & 806  &\\
$7\frac{1}{3}$ &1782& 487&808& 1\\
{\color{teal} $8$} &{\color{teal}1944}& {\color{teal} 609}  & {\color{teal} 726} & {\color{teal}5}\\
$8\frac{2}{3^5}$ & 1946 & 610 & 726& \\
$8\frac{4}{9}$ & 2052& 649 & 754 &\\
$8\frac{2}{3}$ &2106& 650& 806& 2\\
$9\frac{1}{3}$ &2268& 730& 808& 1\\
{\color{blue} $10$} &{\color{blue}2430}& {\color{blue} 852} & {\color{blue} 726} & {\color{blue}1}\\
$10\frac{2}{3^5}$ & 2432 & 853 & 726 & \\
$10\frac{4}{9}$ &2538 & 892 & 754& \\
$10\frac{2}{3}$ & 2592  & 893& 806 & \\
$11\frac{1}{3}$ &2754 & 943& 808&\\ \hline \hline
\multicolumn{5}{c}{${\bf n=6}$} \\ \hline \hline
${\bf 4}$ &\textbf{2916}& {\bf 366}  &{\bf 2184} & \textbf{22}\\
$4\frac{2}{3^6}$ & 2918 & 367 & 2184 & \\
$4\frac{4}{9}$ &3240& 487 & 2266& 1\\
$4\frac{2}{3}$ &3402& 488& 2426& 4\\
$5\frac{1}{3}$ &3888& 730& 2428& 4\\
{\color{red} $6$} &{\color{red}4374}& {\color{red} 1095} & {\color{red} 2184} & {\color{red}13}\\
$6\frac{2}{3^6}$ & 4376 & 1096 &2184 & \\
$6\frac{4}{9}$ &4698& 1216 & 2266& 1\\
$6\frac{2}{3}$ &4860 & 1217& 2426& \\
$7\frac{1}{3}$ &5346& 1459& 2428& 2\\
{\color{teal} $8$}&{\color{teal}5832}& {\color{teal} 1824} & {\color{teal} 2184} & {\color{teal}8}\\
$8\frac{2}{3^6}$ &  5834& 1825 &2184 & \\
$8\frac{4}{9}$ & 6156 & 1945 & 2266& \\
$8\frac{2}{3}$ &6318& 1946 & 2426& 5\\
$9\frac{1}{3}$ &6804& 2188 & 2428 & 2\\
{\color{blue} $10$} &{\color{blue}7290}& {\color{blue}2553}&{\color{blue} 2184}  &{\color{blue}2}\\
$10\frac{2}{3^6}$ & 7292& 2554 &2184 & \\
$10\frac{4}{9}$ & 7614  & 2674 & 2266 &\\
$10\frac{2}{3}$ & 7766& 2675 & 2426 & \\
$11\frac{1}{3}$ & 8262 & 2917 & 2428 & \\ \hline \hline
\multicolumn{5}{c}{${\bf n=7}$} \\ \hline \hline
${\bf 4}$ &\textbf{8748}& {\bf 1095} & {\bf 6558} &\textbf{44}\\
$4\frac{2}{3^7}$ & 8750 & 1096 & 6558 & \\
$4\frac{4}{9}$ &9720&  1459 & 6802 & 3\\
$4\frac{2}{3}$ & 10206& 1460 & 7286 & 7\\
$5\frac{1}{3}$ & 11664  & 2188 &7288 &8 \\
{\color{red} 6} & {\color{red}13122} & {\color{red} 3282}& {\color{red} 6558}& {\color{red}25}
\end{tabular}